\documentclass[sigconf]{acmart}

\usepackage{microtype}
\usepackage{graphicx}
\usepackage{booktabs} % for professional tables
\usepackage{amsmath}
\usepackage{amsfonts}

\usepackage{enumerate}

\usepackage[linesnumbered,ruled]{algorithm2e}

\newenvironment{proofsketch}{\noindent\textit{Proof Sketch.}}{\hfill$\square$}

\usepackage{epsf}
\usepackage{fancyhdr}
\usepackage{graphics}
\usepackage{graphicx}
\usepackage{psfrag}
\usepackage{microtype}
\usepackage{mathtools, algpseudocode}
\usepackage{color}

\usepackage{caption,subcaption}
\usepackage{sidecap}
\sidecaptionvpos{figure}{c}

\usepackage{url}
\newtheorem{assumption}{Assumption}

\long\def\comment#1{}
\comment{
\setlength{\topmargin}{0 in}
\setlength{\textwidth}{5.5 in}
\setlength{\textheight}{8 in}
\setlength{\oddsidemargin}{0.5 in}
}

\newtheorem{lemma}{Lemma}
\newtheorem{claim}{Claim}
\newtheorem{theorem}{Theorem}
\newtheorem{definition}{Definition}
\newtheorem{corollary}{Corollary}%opening

\newcommand{\norm}[1]{\left\lVert#1\right\rVert}

\newcommand{\Br}{\mathbb{R}}
\newcommand{\one}{\mathbf{1}}
\newcommand{\zero}{\mathbf{0}}

\def\MDP{\textbf{MDP}}

\def\EE{\mathbb{E}}

\def\QQ{\mathbb{Q}}

\def\Dd{\mathcal{D}}

\def\Ss{\mathcal{S}}

\def\Mm{\mathcal{M}}
\def\Kk{\mathcal{K}}
\def\Ii{\mathcal{I}}

\def\Ll{\mathcal{L}}

\def\Cc{\mathcal{C}}

\def\eE{\mathbf{e}}

\def\bB{\mathbf{b}}
\def\xX{\mathbf{x}}
\def\yY{\mathbf{y}}

\def\pP{\mathbf{p}}
\def\qQ{\mathbf{q}}

\def\lmb{\boldsymbol{\lambda}}
\def\gma{\boldsymbol{\gamma}}

\def\Lmb{\mathbf{\Lambda}}

\def\op{\mathbf{P}}

\def\opti{\mathbf{OP}}

\def\OP{(\mathbf{P})}

\def\AOP{\mathbf{\widetilde{P}}}

\def\QQ{\mathbf{Q}}

\def\rR{\mathbf{r}}
\def\sS{\mathbf{s}}

\graphicspath{{./image/}}

\newcommand{\ba}{\begin{array}}
\newcommand{\ea}{\end{array}}

\newcommand{\argminE}{\mathop{\mathrm{argmin}}} 

\newcommand{\argmaxE}{\mathop{\mathrm{argmax}}}

\def\MDP{\textsc{MDP}}

\def\DP{\textsc{DP}}

\usepackage{amsthm}
\usepackage{amsmath}
\usepackage{bbm}
\AtBeginDocument{%
  }

\copyrightyear{2025}
\acmYear{2025}
 \setcopyright{acmlicensed}
\setcctype{by}
\acmConference[MobiHoc '25]{The Twenty-sixth International Symposium on Theory, Algorithmic Foundations, and Protocol Design for Mobile Networks and Mobile Computing}{October 27--30, 2025}{Houston, TX, USA}
\acmBooktitle{The Twenty-sixth International Symposium on Theory, Algorithmic Foundations, and Protocol Design for Mobile Networks and Mobile Computing (MobiHoc '25), October 27--30, 2025, Houston, TX, USA}
\acmDOI{10.1145/3704413.3764414}
\acmISBN{979-8-4007-1353-8/2025/10}

\begin{document}

\title{A QoS Framework for Service Provision in Multi-Infrastructure-Sharing Networks}

\author{Quang Minh Nguyen}
\email{nmquang@mit.edu}
\affiliation{%
  \institution{Massachusetts Institute of Technology}
   \country{USA}}

\author{Eytan Modiano}
\email{modiano@mit.edu}
\affiliation{%
  \institution{Massachusetts Institute of Technology}
  \country{USA}}


\begin{abstract}
 We propose a framework for resource provisioning with QoS guarantees in shared infrastructure networks. Our novel framework provides tunable probabilistic service guarantees for throughput and delay. Key to our approach is a Modified Dirft-plus-Penalty (MDP) policy that ensures long-term stability while capturing short-term probabilistic service guarantees using linearized upper-confidence bounds. We characterize the feasible region of service guarantees and show that our MDP procedure achieves mean rate stability and an optimality gap that vanishes with the frame size over which service guarantees are  provided. Finally,  empirical simulations validate our theory and demonstrate the favorable performance of our algorithm in handling QoS in   multi-infrastructure networks. 
\end{abstract}





\maketitle

\section{Introduction}
\label{sec:intro}

Next generation (NextG) networks \cite{nextg_overview, quang_mobihoc, sdn_wireless, quang_dist_sdn} are  emerging due to the rapid  advancements of 5G/6G infrastructures 
% \cite{5g_6g_baby1} 
and in response to the stringent Quality-of-Service (QoS) requirements of modern data-intensive applications.
% \cite{qos_nextg1}. 
A key attribute  of NextG networks is the architectural design that permits seamless integration of heterogeneous Infrastructure Providers (InP) to  offer  over-the-top (OTT) services to end clients \cite{bai_ott} (See Figure \ref{fig: network_overview}). Such  utilization of underlying shared infrastructure not only enhances  overall network capacity, but also facilitates robust service provisioning.
% \cite{nextg_strong1}. 
These advances in network capabilities call for the  development of efficient infrastructure sharing mechanisms that can meet emerging service requirements.
% \cite{qos_extensive}. 

% Moreover, toward the transition  to  production-level networks and the Internet of Things (IoT) paradigm  \cite{edgeIOT} with wide applicability, another challenge for  resource allocation is to also sustain granular QoS requirements of a variety of end applications \cite{qos_extensive}. 

In this paper, we focus on designing efficient algorithms for  QoS-based service provision in shared multi-infrastructure networks.
Efficient service provision is essential for maintaining large-scale multi-infrastructure network operation in order to maintain utilization at the InPs, while meeting clients' service requirements  \cite{nextg_ulti1}. From the system perspective, the implementation of a network orchestrator to facilitate the  coordination of  infrastructure sharing  has been the de facto practice in NextG network deployment \cite{nextg_overview}. Therefore, the design of infrastructure sharing for networked systems has garnered significant research interest, encompassing both  generalized  resource allocation 
% \cite{infras_share1, spectrum_stable2} 
and specific contexts such as spectrum sharing \cite{spectrum_stable1} and network slicing \cite{slice_routing1}. 
% \cite{sdn_ran1, slice_routing1}. 
 Most of the  previous work on QoS modeling focused on capturing  simple requirements such as throughput \cite{timevarying_network1} or energy level \cite{6387341}, which can be cast as long-term linear constraints. Unfortunately, more sophisticated requirements such as delay and reliability involve non-linear probabilistic constraints that are difficult to model and provision. While   there has been a rich literature on modeling real-world  QoS  requirements (see  \cite{qos1}
 % \cite{qos1_journal} 
 and the references therein), modeling realistic QoS requirements may result in non-linear  problems whose solutions are often  NP-hard \cite{qos_approx_bad1} and where relaxation techniques result in  drastic sub-optimality \cite{qos_hard1}. 
% Nevertheless, most of the previous work either neglected  QoS requirements or  could only capture  simplistic QoS  that can be cast in linear form \cite{slice_routing1} (e.g. minimum guaranteed average throughput). 
 % On the other hand, using non-linear constraints for rigidly modelling realistic QoS requirements may result in non-linear complicated problems potentially NP-hard \cite{qos_approx_bad1, qos_resource1, qos_hard1}, where relaxation techniques would either result in  drastic sub-optimality or be restrictive to certain applications. 
 % While there has been a rich literature on modelling real-world  QoS, these works are limited in scope: some consider  specific  applications such as packet routing \cite{qos1}, scheduling \cite{qos2} or video streaming \cite{video_latency1}, others consider certain ad-hoc metrics mainly including throughput, delay and reliability  \cite{qos1} (and the references therein). 
 % in the context of OTT service provision prevalent in NextG multi-infrastructure networks, there is a lack of analytical studies and theoretical guarantees on the QoS that can be provided to the end clients and some limited works on empirical evaluation \cite{ott_qos1, ott_qos2}.
 Moreover, there has been very  limited work on QoS-based OTT service provision  for  multi-infrastructure networks, mostly focusing on empirical evaluation  \cite{ott_qos2}. 
  % \cite{ott_qos1, ott_qos2}. 
\begin{figure}[htbp]
\vspace{-0.25cm}
\centerline{\includegraphics[scale=0.35]{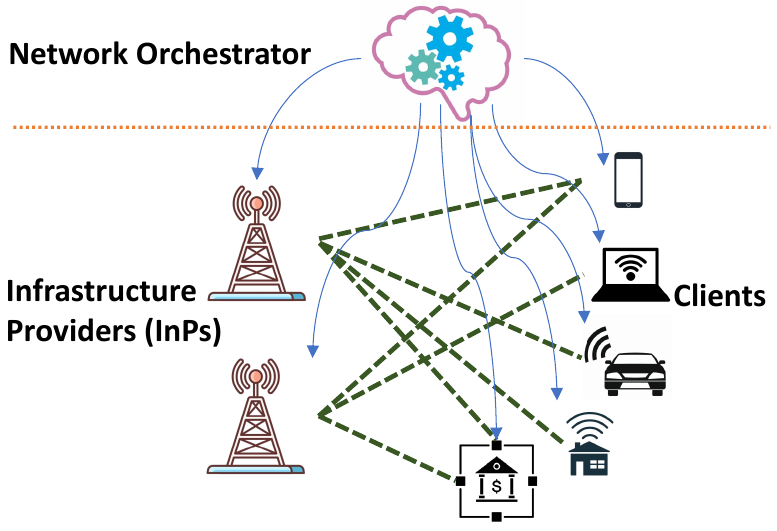}}
\caption{Infrastructure sharing in NextG networks.}
\label{fig: network_overview}
\vspace{-0.25cm}
\end{figure}

% \textcolor{blue}{
To exemplify a realistic shared multi-infrastructure network, consider a setting of two 5G cellular InPs and five clients running various applications, including  video streaming, environmental monitoring  and data backup, as in Figure \ref{fig: network_overview}. 
The InPs dynamically schedule links to serve the clients every 1 ms, which is the supported switching frequency of 5G cellular \cite{5g_speed1}.
 % (e.g. 150 ms time-frame)
Operating at a slower time-scale,  the network orchestrator  gathers the clients' job requests and computes the scheduling decisions for the InPs to serve the clients every time-frame.
% Real-time video streaming  has stringent delay requirement of 150 ms and demands full delivery of the required throughput with negligible failure probability \cite{video_latency1}.  To meet these requirements, the network orchestrator  incorporates throughput requirements in the form of simple linear constraints \cite{num_qos1}  and sets the orchestration time-frame to $T_s = 150$ ms for timely service provision within every time-frame. 
A common approach in the literature is to treat each time-frame independently
% , executing the scheduler periodically \cite{periodic_qos} 
to enforce QoS within that frame  \cite{periodic_qos}. For example, to support real-time video streaming with  delay requirement of 150 ms and negligible failure probability tolerance \cite{video_latency1}, the scheduler can set the time-frame size to be 150 ms and provision full delivery of the required throughput per frame. 
However, this approach neglects inter-frame QoS. 
Specifically, environmental monitoring  can allow delays of seconds to minutes, which correspond to thousands of orchestration time-frames, and does not demand immediate throughput delivery \cite{minutes_delay1}, while data backup typically has no hard deadline requirement,  in which the
data merely needs to be served eventually. 
The design of  QoS-based infrastructure sharing mechanisms  thus requires  a new QoS framework for OTT services that is powerful enough to ensure  broad applicability yet analytically tractable enough to enable the provisioning of service guarantees.

% }

% , which shares similarities with the average throughput requirements of some previous work on QoS based infrastructure sharing, yet enforces it to be 

To address the aforementioned challenges, we propose a  novel QoS framework that captures delay and probabilistic  service guarantees for OTT services. The framework is specified by three parameters: time-frame size, guaranteed delivery throughput  and reliability. Specifically, our service provisioning is within a fixed time-frame, whose length  $T_s$ can be chosen to
capture stringent  latency requirements. For every client $i$ with  type-$c$ jobs, our framework specifies a guaranteed delivery throughput $\gamma^c_i$ that is satisfied with probability at least  $q^c_i$ in each time-frame. Moreover, our model also takes into account    queue stability of job requests to ensure long-term QoS for applications such as data backup, where timely but eventual service fulfillment is sufficient.
 Closest to our work are \cite{timevarying_network1, qos1}. In particular, \cite{timevarying_network1} formulated a carefully designed network utility maximization problem (NUM) with stability constraints to jointly provision fairness and throughput guarantees and \cite{qos1} provides QoS in the form of long-term guaranteed delivery ratio of deadline-constrained packets with probability 1. Our model extends NUM by incorporating both stability constraints and our proposed QoS framework. By setting $q^c_i = 0$  for all clients and jobs, our model reduces to the classical NUM  formulation in \cite{timevarying_network1}. For the required $\gamma^c_i$ and $q^c_i \to 1$, our model aligns with \cite{qos1} in providing guaranteed delivery ratio, while additionally satisfying queue stability. 
 % While \cite{timevarying_network1, qos1} support QoS in the form of long-term guarantees and hard deadline constraints, i.e. within single time-frame,  
 Furthermore, our QoS model with appropriately chosen $\gamma^c_i$ and $q^c_i$ provides multi-time-frame delay guarantees. 
 This  enables the jobs to be served over multiple time-frames to  model applications with soft delay requirements (e.g. average delay) such as environmental monitoring and traditional Internet services. Finally, our framework's modeling flexibility in tuning  $T_s, \gamma^c_i$ and $q^c_i$ allows for various combinations of all these models to provide heterogeneous QoS for different classes of applications. 
 % where the total long-term delay (in terms of time-frames) can be guaranteed by our framework. 
% \textbf{Related work:}
To this end, we present a new QoS-based infrastructure sharing paradigm for NextG  networks. Our contributions are as follows:
 \begin{itemize}
     \item We present a novel QoS model that guarantees a specified OTT service within a bounded time-frame, which can be designed to capture hard delay requirements, with a pre-determined probability. 
     % The  time-frame size can be set to meet the most stringent delay requirement of relevant real-time applications; then the guaranteed OTT service is provisioned with very high probability.  
     % Flexibly modelling also applications with mild delay and service requirements, our  framework allows the jobs to be served in later time-frames, yet can be fine-tuned to provide varying  long-term throughput and delay guarantees.
     We demonstrate  the  application of our QoS model to a number of  realistic scenarios.
     \item We formulate the infrastructure sharing problem to ensure  feasibility of the QoS requirements and stability of the queues. 
     We use  upper-confidence bounds to linearize the probabilistic QoS constraints, and characterize the trade-offs between optimality of the service guarantees and the frame duration over which service is provisioned.
     % To deal with the probabilistic QoS constraint, we linearize it using upper confidence bound and  characterize the tight approximation error that trades-off with the size of time-frame.
     % \item To address the complexity of the stability region incurred by the probabilistic constraints, we characterize it by an approximate stability region that is arbitrarily close to the true stability region as the frame duration increases. We show that the approximate stability region admits an equivalent form of polynomial-size polyhedron, thereby facilitating the development of the efficient feasibility checking  procedure to validate if the clients' requested level of QoS is sustainable.
     % in the form of polynomial size polyhedron
     %Based on the linearization technique, we develop the  feasibility checking procedure, which validates if the clients' requested level of QoS is sustainable, and the efficient algorithm
     \item Based on the linearization technique, we develop the efficient algorithm, which  minimizes the quadratic Lyapunov drift-plus-penalty, hence named Modified Drift-plus-Penalty (MDP), for the infrastructure sharing problem. 
     The MDP policy provably supports a stability region that is arbitrarily close to the original stability region, as the length of the time-frame  increases. Moreover, we prove that the MDP policy attains an optimality gap that vanishes with the frame length.  Empirical simulations demonstrate the efficiency of the MDP policy and its ability to handle realistic QoS requirements.
     % We provide analytical studies on the throughput and delay guarantees implied by our generalized QoS framework as well as its exemplary application to realistic scenario.
 \end{itemize}

% \textbf{Organization:} 
The rest of the paper is organized as follows. We present our system model in Section \ref{sec: system_design}. In Section \ref{sec_qos_infras_sharing}, we develop our new QoS framework and the  formulation of the infrastructure sharing problem. In Section \ref{sec_algos}, we propose the efficient MDP policy to solve the problem. The results on the optimality and stability guarantees of MDP are presented in Section \ref{sec_guarantees_MDP}. In Section \ref{sec_qos_application}, we demonstrate the application of our QoS framework in  realistic  scenarios and  conduct empirical simulation of MDP to validate its ability to handle QoS constraints. We conclude the paper in Section \ref{sec_conclude}. 
% nextG networks

% Despite such system advancement expanding the network capabilities, the algorithmic design to  coordinate efficient infrastructure sharing still lags behind.
% Moreover, besides the integration of heterogeneous InPs, NextG networks support more extensive set of clients with stringent QoS requirements .

% By leveraging unified resource 

% heterogeneous InPs? service? so we need OTT?

% One fundamental challenge toward 

% are increasingly large-scale and complex

\section{Preliminaries}
\label{sec: system_design}

\subsection{ Notations}
We use bold font for vectors (e.g. $\xX$) and write $(x)^+ = \max\{0, x\}$ for brevity. For $1 \leq p \leq \infty$, let  $\norm{\cdot}_p$  be the $\ell_p$-norm of  vector. In particular, we  have $\|\xX\|_1 = \sum_i |x_i|$ and $\|\xX\|_\infty = \max_i \{|x_i|\}$.   The vectors of all zeros and all ones are respectively denoted by $\zero$ and $\one$. For brevity, we denote $[n] = \{0, 1, .., n-1\}$ given an integer $n$.
\subsection{System Model}
\label{sec_sys_model}

\textbf{Network Stakeholders:} We consider a resource allocation scenario with  multiple Infrastructure Providers (InP) represented by $\Kk = \{1, 2, ..., K\}$ and multiple clients represented by $\Ii = \{1, 2, ..., I \}$. 
The InPs serve the clients to fulfill different types of jobs, indexed by $c\in \Cc = \{1, 2, ... ,C \}$. 
 InP $k$ is connected to and thus can serve a set $N(k)\subseteq \Ii$ of clients. Similarly,  client $i$ is connected to and thus can be served by a set $N(i) \subseteq \Kk$ of InPs. Let $E = \{(k, i): k\in \Kk, i \in N(k)\}$ be the set of all the links. 
% Each InP $k$ has a computing capacity of $\pi_k$ (e.g. in Mbps, demonstring the ammount of data the InP can process), while every unit of type-$c$ job requires $\beta_c$ computing resource. 
We consider an uplink  system where the InPs receive jobs from the clients, which models a variety of real-world applications like video broadcasting or sensor monitoring; however, an analogous problem can be formulated for downlink system, and our results can be readily generalized to  setting  with both  uplink and downlink transmissions.

\textbf{Network Orchestrator:} The network orchestrator (NeO) is connected to all the InPs and clients to facilitate a centralized coordination over the network. From the system perspective, the NeO can be implemented via the Software-Defined Networking (SDN) paradigm, which separates the control plane from the data plane and provides a programmable management layer on top of the network; this  approach of placing a SDN centralized controller on top of the underlying network for flexible control has been extensively utilized in resource sharing architectures \cite{sdn_ran1, slice_routing1}. In our system model, the NeO helps facilitate the Service-Level Agreement (SLA) and meet QoS requirements by making scheduling decisions for the InPs.
%on specifying QoS requirements of pre-deployment and, during the actual operation, supports collecting statistics information (of incoming jobs and queued jobs) from the clients and consequently computing the scheduling decisions for the InP to deploy that aggregately optimize for global network metrics like  fairness, bandwidth, or energy consumption. 

\textbf{System Time-Scales:} We consider two time-scales in this system model (see Figure \ref{fig: event_sequence}):
\begin{enumerate}[(i)]
    \item \textit{the orchestration time-frame $t\in [0, T-1]$} over which the NeO gathers information such as incoming job requests and  computes the scheduling decisions for the InPs  to decide which clients and job types to serve. The time horizon $T$ captures the entire lifespan of the system. The  length of each time-frame $T_s$ can be chosen to capture the  latency requirement among  applications of interest.%, from which QoS would be facilitated through our framework.
    \item \textit{the network operational time slot $\tau\in [0, T_s-1]$} over which the InPs perform the actual scheduling  within every orchestration time-frame $t$. Specifically, each time-frame $t$ is divided into $T_s$ equal operational time slots. The actual length of each operational time slot should be chosen in accordance to the transmission scheduling intervals of the InPs; for example, 5G InPs support sub-millisecond scheduling  \cite{5g_speed1}, i.e. changing the schedule every 0.5 ms.  
\end{enumerate}
\begin{figure}[htbp]
\vspace{-0.23cm}
\centerline{\includegraphics[scale=0.75]{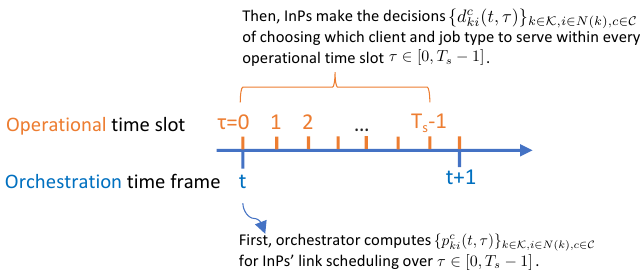}}
\caption{Sequence of events in our two time-scales model.}
\label{fig: event_sequence}
\end{figure}
As the supported scheduling intervals for  NextG network hardware are typically much smaller than the latency requirements of applications, we are interested in the regime of large $T_s$ (of the order of hundreds as shown in Section \ref{sec_example}).

{\textbf{Policy Parameterization and Channel State Model:} }
At every operational time slot, each InP executes a randomized scheduling protocol that chooses the client to be connected to and the type of job to be served. To model the randomized scheduling procedure to be implemented at the InPs, we use $p^c_{ki}(t, \tau)$ to denote the probability that  InP $k$ decides to serve type-$c$ jobs from client $i$ in the  operational time slot $\tau$ within the orchestration time-frame $t$, i.e. $\pP_k(t, \tau) = \{p^c_{ki}(t, \tau)\}_{ i\in N(k), c\in \Cc}$ forms the prior of link activation probabilities for the InP $k$:
\begin{align} 
% & p^c_{ki}(t, \tau) \geq 0, \quad \forall k\in\Kk, i\in\Ii, c\in \Cc, \tau\in [0, T_s-1] \label{prob_constraint0}\\
    &\sum_{c\in \Cc} \sum_{i\in N(k)} p^c_{ki}(t, \tau) = 1,  \forall k\in\Kk,  \tau\in [T_s].\label{prob_constraint}
\end{align}
% $\pP(t, \tau) = \{p^c_{ki}(t, \tau)\}_{k\in \Kk, i\in N(k), c\in \Cc}$
For brevity, we write $\pP(t, \tau) = \{\pP_k(t, \tau)\}_{k\in \Kk}$ to  denote the link activation probabilities of all the InPs at the current  time slot. 
At the start of  time-frame $t$, the NeO computes prior probabilities $\pP(t) = \{\pP(t, \tau)\}_{\tau \in[0, T_s-1]}$ for  scheduling in  the next $T_s$ network operational time slots. 
Then, at every subsequent operational time slot $\tau$, each InP $k$ (independently of other InPs)  activates at random a link $(k, i)$ for serving jobs of type $c$ with probability $p^c_{ki}(t, \tau)$, which is modeled by:
% .following the distribution $\textsc{Distr}(\pP_k(t, \tau))$. We model the actual link schedule of InP $k$ by $\yY_k(t, \tau) = \{ y^c_{ki}(t, \tau)\}_{i\in N(k), c\in \Cc} \sim \textsc{Distr}(\pP_k(t, \tau))$, i.e.
% $y^c_{ki}(t, \tau)$ as follows:
\begin{align*}
d^c_{ki}(t, \tau) = \mathbbm{1}\{ \text{InP $k$ chooses to serve type-$c$ jobs from client $i$} \}
% y^c_{ki}(t, \tau)= 
%     \begin{cases} 1  \text{, if InP $k$ serves type-$c$ jobs from client $i$} & \\
%     0  \text{, otherwise} \end{cases}. 
\end{align*}
with $p^c_{ki}(t, \tau) = P(d^c_{ki}(t, \tau)  = 1)  $. Note that as only one client with the corresponding job type is chosen to be served by InP $k$ in every operational time slot, we also have $ \sum_{c\in \Cc} \sum_{i\in N(k)} d^c_{ki}(t, \tau) = 1$. 
In this model, we assume that each client can connect to  multiple InPs at every time slot $\tau$, a capability commonly supported in NextG networks through carrier aggregation \cite{downlink_ca1, uplink_ca2}. However, the model can be easily modified to restrict clients to one InP per slot.
% \footnote{For model where each client is constrained to send to/receive from a single InP at every time slot $\tau$}.
 When InP $k$ activates the link $(k, i)$ to serve type-$c$ jobs from client $i$ at time slot $\tau$, i.e. $ d^c_{ki}(t, \tau)=1$, the jobs are successfully served with probability $r_{ki} \in (0, 1)$ and a transmission error occurs with probability $1-r_{ki}$. We thus denote the effective service by:
\begin{align*}
% \label{actual_schedule1}
y^c_{ki}(t, \tau) = \mathbbm{1}\{ \text{InP $k$ successfully serves type-$c$ jobs from client $i$} \}
% y^c_{ki}(t, \tau)= 
%     \begin{cases} 1  \text{, if InP $k$ serves type-$c$ jobs from client $i$} & \\
%     0  \text{, otherwise} \end{cases}. 
\end{align*}
with $P(y^c_{ki}(t, \tau)  = 1) = p^c_{ki}(t, \tau) r_{ki} $. Let $r_{max}= \max_{k \in \Kk, i\in \Ii} \{ r_{ki}\}$ be the statistic of the most reliable link.
Figure \ref{fig: event_sequence} summarizes the sequence of events across the two time-scales. For any policy $\pi$, we use the superscript $\pi$ to denote  that a variable is associated with $\pi$.

% For any variable affected by the control of the scheduling decisions, we add the superscript $\pi$ to acknowledge that it is under the action of the policy $\pi$

\textbf{Queue of Requests:} For any job entering the system, if it has not been served yet, it will be queued to be served later. We maintain such queues of requests at the orchestration time-scale. 
For simplicity of technical exposition, we assume that all the links $(k, i)$ have unit link capacity per operational time slot, during which one unit of traffic can be served.
At the start of orchestration time-frame $t$, each client $i$ receives $a^c_i(t)$ incoming new type-$c$ packets with the arrival rate of $\lambda^c_i = \EE[a^c_i(t)]$. 
% We assume that $\lambda_{min}=  \min_{i\in \Ii, c\in \Cc: \lambda^c_i > 0} \{\lambda^c_i\} >0$, i.e. the system would at least receive some arrivals. 
% We assume that $\lambda^c_i > 0, \forall i\in \Ii, c\in \Cc$, where we neglect a job from the system if it has arrival rate 0, and let $\lambda_{min} = \min_{i\in \Ii, c\in \Cc} \{\lambda^c_i\}$.
The number of incoming new packets is bounded by a finite number $a^c_i(t) \leq A_{max}$ at every time-frame $t$.
% For simplicity of technical exposition, we assume that all the links $(k, i)$ have the same link capacity per orchestration time slot of $\alpha$.
Under the unit link capacity, in one time-frame of $T_s$ time slots, the $|\Kk|$ InPs choose  $|\Kk| T_s$ clients and their respective jobs in total to provide service. Thus, $\delta(\lmb) = \min_{i\in \Ii, c\in \Cc: \lambda^c_i > 0} \big\{\frac{\lambda^c_i}{|\Kk| T_s} \big\}$ is the minimum load factor  of the system. We assume that $ \delta(\lmb) \geq \delta_{min}  >0$, i.e. the system is non-empty.
The total number of times the client $i$ has received service for type-$c$ packets from InP $k$ within the  time-frame $t$ is $ y^c_{ki}(t) = \sum_{\tau=0}^{T_s-1} y^c_{ki}(t, \tau),  \forall k\in N(i)$.
% \begin{align}
%     y^c_{ki}(t) = \sum_{\tau=0}^{T_s-1} y^c_{ki}(t, \tau), \quad \forall k\in N(i). \label{service_ki}
% \end{align}
Then, the total number of times it has received service for type-$c$ over all InPs is:
\begin{equation}
    \mu^c_i(t) = \sum_{k\in N(i)}  y^c_{ki}(t) =  \sum_{k\in N(i)}  \sum_{\tau=0}^{T_s-1} y^c_{ki}(t, \tau), \label{service_unit1}
\end{equation}
which is  the  effective  type-$c$ service for client $i$  in  time-frame $t$.
% Thus, the effective service for type-$c$ jobs that client $i$ receives in orchestration time slot $t$ is $\alpha \mu^c_i(t)$. 
The queue of type-$c$ job requests at client $i$ then evolves as:
 \begin{equation}
        Q^c_i(t+1) = \big( Q^c_i(t) + a^c_i(t) -  \mu^c_i(t)\big)^+, \forall i\in \Ii, c\in \Cc. \label{queue_dynamics}
    \end{equation}
% where $(x)^+ = \max\{0, x\}$.

\subsection{ System Requirements}
\label{sec_resource_allo}
In this section, we first present the constraints that we impose on our decision variables  $\pP(t)$ to model the  system requirements. 
% However, capturing merely the system requirements is insufficient to exemplify real-world scenario with specific QoS demands for different applications, which have been not adequately addressed  by the literature.
% \subsubsection{Resource Allocation and Stability Constraints}
% % \label{sec_resource_allo}
% We hereby describe the constraints of interest that model realistic scenario in infrastructure sharing for NextG networks and requires the adherence of the decision variables $\pP(t)$ to be optimized. 
Note that under unit link capacity setting, the link activation probability $r_{ki} p^c_{ki}(t, \tau) = \EE[y^c_{ki}(t, \tau) ]$ is also the expected amount of type-$c$ traffic that  InP $k$ serves  client $i$ in the current time slot. The  management of the NeO, i.e. the control of prior probabilities $\pP(t)$ for scheduling, must take into account available resources at the InPs and clients, and ensure the  stability of the system. 

% that the total workload to be handled by each InP does not exceed its processing capacity and that job dispatch does not exceed the maximum transmission and reception rate supported by the clients.
% , which is upper-bounded by the the InP's computing resource $\pi_k$

\textbf{Transmission Rates of Clients:}
% Recall that each unit type-$c$ job requires $\beta_c$ computing resource and $\pi_k$ is the computing capacity that InP $k$ can sustain. Thus, we impose the constraint that the total amount of jobs allocated to any InP $k$ should not exhaust its computing capacity:
% \begin{align}
%    \frac{1}{T_s}\sum_{c\in \Cc}\sum_{i\in N(k)} \sum_{\tau=0}^{T_s -1 } \beta_c \cdot p^c_{ki}(t, \tau) \leq \pi_k, \quad \forall k\in \Kk. \label{compute_inp1}
% \end{align}
% Here, the LHS of \eqref{compute_inp1} is the average computing workload allocated to InP $k$ within the orchestration time slot $t$. 
At each operational time slot $\tau$ within  time-frame  $t$,  every client $i$ transmits to InP $k$ for type-$c$ job service with probability $p^c_{ki}(t, \tau) $. 
% receives the total average OTT service throughput $ \sum_{c\in \Cc} \sum_{k\in  N(i)} r_{ki} p^c_{ki}(t, \tau) $  from all the InPs. 
Thus, the client would have to make $ \sum_{c\in \Cc} \sum_{k\in  N(i)}  p^c_{ki}(t, \tau)$ transmissions on average, which should not exceed the  transmission  rate $U_i > 0$  of the client:
% . We thus impose that:
\begin{align}
 \sum_{c\in \Cc} \sum_{k\in  N(i)}  p^c_{ki}(t, \tau)  \leq U_i,  \forall i\in \Ii, \tau\in [T_s]. \label{pmax1}
\end{align}
% where we assume for simplicity that the clients has the same maximum transmission  rate of $p_{max} \in (0,1)$. 
Let $U_{max} = \max_{i\in \Ii} \{ U_i\} $ be the maximum transmission rate of the clients.
This constraint \eqref{pmax1} can  be used to capture power control or bandwidth management in infrastructure sharing of  networks. 
% Finally, given the average type-$c$ job arrival rate of $\lambda^c_i$ to client $i$ at every time slot $t$, the NeO aims to provision the average OTT service   throughput $\EE[\mu^c_i(t)]$ to sustain such workload:
% \begin{align}
%     \lambda^c_i \leq \EE[\mu^c_i(t)], \quad \quad \forall i\in \Ii, c\in \Cc.
% \end{align}

\textbf{Stability of Queues of Requests:}  Stability has been an integral requirement in many communication networked systems \cite{timevarying_network1} and  infrastructure sharing networks \cite{spectrum_stable1}. We require  the scheduling decisions to satisfy  mean rate stability of the system \cite{neelybook1}:
\begin{align}
    \lim_{T\to \infty}\frac{\EE\big[ \sum_{i\in\Ii} \sum_{c\in \Cc}Q^c_i(T) \big]}{T} =0, \label{queue_stability1}
\end{align}
or equivalently $\EE[\|\QQ(T)\|_1]=o(T)$ with $\QQ(T) = \{Q^c_i(T)\}_{i\in \Ii, c\in \Cc}$ being the vector of the queue lengths. The constraint \eqref{queue_stability1} requires that the   job requests at all clients should eventually be served.

\textbf{Service-Level Agreement:}
Prior to the system deployment requiring interactions between various network stakeholders, a service-level agreement (SLA) is established to determine the  QoS  provided by the InPs to the clients. This phase of SLA is a common  practice in infrastructure sharing \cite{sla1}. The SLA phase in our model is used to specify all the parameters for system requirements and QoS (Section \ref{qos_model}) to both the InPs and the clients. These parameters  will then be given to the NeO to coordinate the system deployment.

% (Section \ref{sec_resource_allo}) 

% \subsubsection{Insufficient Modelling Capability}

% \cite{ spectrum_stable1, spectrum_stable2}
\subsection{Modeling QoS}
The modeling of system requirements, as presented above, can be well captured by linear constraints in the formulation. 
More generally, this approach of leveraging linear constraints is heavily used by the vast literature on resource allocation for infrastructure sharing in network \cite{ spectrum_stable1} thanks to its tractability. In fact, it can be readily extended to also capture certain simplistic QoS  requirements (e.g. throughput) that can be cast in linear form \cite{slice_routing1}. However, restricting the problem formulation to such linear constraints fails to capture  QoS  requirements such as delay, reliability or guaranteed service. 
% On the other hand, using non-linear constraints for rigidly modelling realistic QoS requirements may result in non-linear complicated problems potentially NP-hard \cite{qos_hard1}, where relaxation techniques would either result in  drastic sub-optimality or be restrictively ad-hoc to certain applications. Consequently, the design of a QoS-based infrastructure sharing primitive with broad applicability requires adopting a new QoS framework that is powerful enough to ensure practicality yet tractable enough for solvability, along with an efficient accompanying algorithm. 
In the next Section, we  present our novel QoS framework that captures delay and probabilistic service guarantees and is specified by three parameters: time-frame size ($T_s$), guaranteed delivery throughput ($\gma$) and reliability metric ($\qQ$). The time-frame size ($T_s$) over which service is provisioned can be chosen to meet the stringent  hard latency requirements of real-time applications. Our QoS model specifies  the guaranteed delivery throughput ($\gma$) analogous to the average throughput constraint of some previous works \cite{thruput_qos1}, yet only requires it to be satisfied with some specified probability ($\qQ$).  
Consequently, our framework  allows the jobs to be served in multiple time-frames, thereby  modeling  applications with soft delay requirements (e.g. environmental monitoring), and provides   long-term delay guarantees (in terms of time-frames) for such  applications. At the same time, our framework can ensure hard delay guarantees of up to the time-frame size  ($T_s$) by setting the reliability metric ($\qQ$) close to 1. Such hard delay requirements and delivery ratios have been considered in the literature \cite{qos1}. 

% Specifically, our service provisioning is within a fixed time-frame, which can be chosen to
% capture service requirements within stringent allowable latency of real-time applications. 
% Our framework specifies  guaranteed delivery throughput, which shares similarities with the average throughput requirements of some previous work on QoS based infrastructure sharing, yet enforces it to be satisfied with some specified probability, i.e. our reliability metric; such flexibility allows the jobs to be served in subsequent time-frames to also model applications with mild delay requirements (e.g. data backup) where the total long-term delay (in terms of time-frames) can be guaranteed by our framework. 
% QoS  of real-time applications such as video streaming which demands certain guaranteed throughput within stringent latency requirements, we 

% in order to capture stringent latency requirement of real-time applications such as video streaming, the time-frame size 

% with broad applicability, and an efficient algorithms tailored to handle the probabilistic QoS constraints in Section \ref{sec_algos}.

% addresses the aforementioned challenges by incorporating our 

% comprised our full problem formulation leveraging 
\section{QoS based Infrastructure Sharing Formulation}
\label{sec_qos_infras_sharing}

In this Section, we present our QoS based infrastructure sharing formulation.
We first propose our new QoS framework that can simultaneously capture delay, reliability and guaranteed service in Section \ref{qos_model} and then provide the full formulation for the infrastructure sharing control problem in Section \ref{full_problem}. .
% , as well as demonstrating the exemplary application of the proposed model to satisfy QoS requirements in realistic scenario in Section \ref{sec_qos_application}. 
% From there, we also establish principled study of how our model of QoS can be fine-tuned to imply the conventional throughput and delay guarantees. Finally, incorporating our QoS model, we formulate  the  infrastructure sharing problem 
% and derive the characterization of its feasibility  region 
% in Section \ref{full_problem}. 

\subsection{Model for QoS}
\label{qos_model}

First, the QoS constraint is specified within every time-frame $t$, i.e. over the $T_s$ operational time slots therein. From the system point of view, the NeO can choose $T_s$ such that the actual length of one orchestration time-frame is upper bounded by the most stringent delay requirements of all relevant applications, so that requested resources are always provisioned within the delay requirements. 
Second, note that (under the unit link capacity) the total OTT service $\mu^c_i(t)$ for type-$c$ jobs at client $i$ as in \eqref{service_unit1} is the sum of up to $|\Kk| T_s$ service opportunities from at most $|\Kk|$ InPs over $T_s$ time slots. Thus, by enforcing a guaranteed delivery ratio $\frac{1}{|\Kk| T_s} \mu^c_i(t) > \gamma^c_i$, we can provision the required service throughput to the client by controlling $\gamma^c_i\in (0,1)$. However, in many applications such as web-browsing or large  data backup, where the delay requirement is milder, such guaranteed delivery ratio does not need to be necessarily satisfied within one time-frame. 
Consequently, we further specify the reliability metric $q^c_i \in [0,1)$ as the  probability that the  delivery ratio is guaranteed and impose the QoS constraint for every time-frame $t$ as:
\begin{align}
    P\bigg( \frac{1}{|\Kk| T_s} \mu^c_i(t) > \gamma^c_i \bigg) \geq q^c_i, \quad \forall i\in \Ii, c\in \Cc. \label{qos_constraint1}
\end{align}
 % \textcolor{blue}{Let   $\gamma_{min} = \min_{i\in \Ii, c\in \Cc} \{\gamma^c_i\}$ be the minium guaranteed delivery ratio of the clients.}
 The set of $(T_s, \gma, \qQ)$, consisting of the frame size $T_s$, the guaranteed delivery ratios $\gma = \{\gamma^c_i\}_{i\in \Ii, c\in \Cc}$ and the reliability metrics $\qQ = \{q^c_i\}_{i\in \Ii, c\in \Cc}$, is part of the SLA between the network stakeholders, i.e. InPs and clients, and the NeO prior to the system deployment. 
 Note that setting $q^c_i = 0$ corresponds to no QoS requirement, where the  constraint \eqref{qos_constraint1} trivially holds.
In   Section \ref{sec_qos_application}, we show through realistic examples how $T_s, \gma$ and $\qQ$ can be chosen to exemplify real-world applications' requirements.

\subsection{The Full Infrastructure Sharing Problem}
\label{full_problem}

% \subsubsection{Problem Formulation} 
Now, we proceed to present our full QoS based infrastructure sharing formulation. Given the guaranteed delivery ratios $\gma$ and reliability metrics $\qQ$ pre-defined in the SLA, we formally define the feasible domain $\Dd(\gma, \qQ) $ of our  link activation probabilities $\pP(t) =   \{p^c_{ki}(t, \tau)\}_{\tau\in[0, T_s-1], (k,i) \in E, c\in \Cc}$ that incorporates both the resource allocation constraints (Section \ref{sec_resource_allo}) and the probabilistic QoS constraints (Section \ref{qos_model}) as:
\begin{align}
\label{def_region}
     \Dd(\gma, \qQ) = \big\{ \pP(t) \in [0, 1]^{T_s |E| |\Cc|}\text{: such that \eqref{prob_constraint}, \eqref{pmax1}, \eqref{qos_constraint1}}\big\}.
\end{align}
After the scheduling  probabilities $\pP(t)\in   \Dd(\gma, \qQ)$ have been made at every time-frame $t$ where $r_{ki} p^c_{ki}(t,\tau) = \EE[ y^c_{ki}(t,\tau)]$ is  the expected job service that client $i$ receives from InP $k$ for type-$c$ jobs, a utility $f(\pP(t))$ is generated. We only require that the utility function $f(.)$ is continuous and concave, which hold  in practical scenarios such as elastic traffic  or fairness \cite{concave_network1}, and bounded within the feasible domain, i.e. $f(\pP(t)) \in [f_{min}, f_{max}]$. We impose no structural assumption on $f(.)$ for generality. As such, $f(.)$ can be designed to capture the InPs' costs \cite{bai_quang_track}, the clients' returns \cite{sdn_ran1} or both, as well as administering additional factors like fairness \cite{fairness1}. Our goal is  to maximize the long-term average 
 network utility while satisfying the realistic requirements, as described by $\Dd(\gma, \qQ)$, and ensuring the stability \eqref{queue_stability1} of the system:
 \begin{align*}
      \mathbf{(P):} \quad   \max_{\pP=\{\pP(t)\}_{t\geq0}}   &F(\pP) \triangleq \liminf_{T\to \infty}\frac{1}{T} \sum_{t=0}^{T-1}f(\pP(t))\\
         \text{such that}\quad &\pP(t) \in \Dd(\gma, \qQ), \\
         &  \EE[\QQ(T)] = o(T).
\end{align*}
The above problem $\OP$ is feasible if there exists a control policy $\pi$ satisfying the feasibility constraints and stabilizing the system \cite{neelybook1}. 
% , i.e. $\pP^\pi(t) \in  \Dd(\gma, \qQ)$, and ensuring the stability  $\EE[\QQ^\pi(T)] = o(T)$ 
Denote by $\Lmb(\gma, \qQ)$ the stability region of the above problem $\OP$, describing the set of job arrival rate vectors $\lmb = \{\lambda^c_i\}_{i\in \Ii, c\in \Cc}$ whereby the problem $\OP$ is feasible.
% , as formally defined below.
\begin{definition}[Stability Region]
The stability region of the set of QoS specifications given by $\gma$ and $\qQ$ is defined as:
    \begin{align*}
    \Lmb(\gma, \qQ)= \{&\lmb = \{\lambda^c_i\}_{i\in \Ii, c\in \Cc}\in \Br_{\geq0}^{|\Ii||\Cc|}:  \text{$\exists \pi$ such that } \\
    & \pP^\pi(t) \in  \Dd(\gma, \qQ), \forall t\geq 0 \text{ and } \EE[\QQ^\pi(T)] = o(T)  \}.
\end{align*}
\end{definition}

In the SLA phase, the network stakeholders share their information such as the resource capacities of InPs and clients and the requested level of QoS (described through $\gma$ and $\qQ$) with the NeO, whose job is to solve $\OP$ to derive real-time scheduling decisions $\pP(t)$ for the InPs. In the next section, we proceed to develop  an efficient control algorithm for the infrastructure sharing problem.

% The NeO thus needs to solve two problems:
% % to finalize the control  problem $\OP$ for further system deployment. Consequently, to facilitate QoS based infrastructure sharing, our problem is two-fold:
% \begin{enumerate}[(i)]
% \item \textit{Efficient control algorithm:} The NeO would be in charge of executing the algorithm solving $\OP$ to derive real-time scheduling decisions $\pP(t)$. Hence, it is necessary to have an  efficient algorithm of $\OP$ for practical deployment.
% \item \textit{Feasibility checking:}  The clients  naturally may want the highest level of QoS as possible, thereby requesting over-stringent QoS guarantees; this overall shrinks the supportable capacity of the system. Given the requested QoS level specified by $T_s, \gma$ and $\qQ$, it is crucial for the NeO to check if the incoming (or envisioned) job arrival rates $\lmb = \{\lambda^c_i\}_{i\in \Ii, c\in \Cc}$ can be supported in the sense that $\lmb $. If it is not, the NeO can either ask the clients to lower their level of QoS, or only agree to support lower job arrival rates.    
% %\in  \Lmb(\gma, \qQ) 
% \end{enumerate}
% In the next Section, we proceed to develop  an efficient control algorithm and a procedure for feasibility checking.
% \begin{itemize}
%     \item 
% \end{itemize}

% , i.e. $\lmb = \{\lambda^c_i\}_{i\in \Ii, c\in \Cc}$,
% \subsubsection{Feasibility Region}

% \textbf{The full infrastructure sharing problem}

% \section{ The Modified Drift-plus-Penalty Policy }
\section{Efficient Algorithm for QoS Based Infrastructure Sharing}
\label{sec_algos}

% aiyo

In this section, we present the Modified Drift-plus-Penalty (MDP) policy as an efficient algorithm for  our spectrum sharing problem $\OP$. We start by providing the preliminaries of the conventional Drift-plus-Penalty (DP) approach as well as its inability to handle the probabilistic constraints in Section \ref{sec_hardness}. 
The algorithmic development of MDP is given in Section \ref{sec_MDP} and is based on a linearization technique to well approximate the probabilistic constraints. 
% The  optimality and stability  guarantees of MDP then follow in Section \ref{sec_feasibility}.

% Central to our spectrum sharing formulation is the probabilistic QoS, which not only differentiates it from the vast literature (mostly with linear constraints) but also grants ours sufficient modeling capabilities to capture realistic requirements (as shown later in Section \ref{sec_qos_application}). However,  dealing directly with such constraint can be hard, thereby motivating our development of a linearization technique to well approximate it via upper confidence bound (UCB) in Section \ref{sec_hardness}. 
% From there, we further derive tight characterization on the approximation error of the linearization procedure. Finally, we present our feasibility checking procedure in Section \ref{sec_feasibility} and the efficient policy for the control problem in Section \ref{sec_MDP}. 

\subsection{Inability of Drift-plus-Penalty to Handle  QoS Constraints}
\label{sec_hardness}

\subsubsection{Complexity of the Probabilistic QoS Constraint}
\label{sec_hardness_qos}
First, from   \eqref{service_unit1}, noting that each InP $k$ independently initiates link activation, we know that $ \mu^c_i(t)= \sum_{k\in N(i)}  \sum_{\tau=0}^{T_s-1} y^c_{ki}(t, \tau)$ is the sum of $N(i) T_s$ independent Bernoulli random variables $y^c_{ki}(t, \tau) \sim Bern( r_{ki} p^c_{ki}(t, \tau))$. Thus, $\mu^c_i(t)$ follows Poisson Binomial distribution \cite{tang2019poisson}, so the probabilistic constraint \eqref{qos_constraint1} can be written as: 
\begin{align}
\nonumber 
     &P\big( \frac{1}{|\Kk| T_s} \mu^c_i(t) >  \gamma^c_{i}  \big) \geq q^c_{i}, \quad \forall i\in I, c\in \Cc\\
     \nonumber
     \Leftrightarrow&\sum_{l=0}^{\lfloor \gamma^c_{i} |\Kk| T_s\rfloor} \sum_{A\in F_l} \prod_{(k, \tau)\in A} r_{ki} p^c_{ki}(t, \tau)  \prod_{(k', \tau')\in A^c} (1- r_{k'i} p^c_{k'i}(t, \tau') ) \\
     &\quad \leq 1- q^c_{i}, \quad \forall i\in I, c\in \Cc, \label{too_hard1}
\end{align}
where $F_l$ is the set of all subsets of size $l$ that can be selected from ${\Kk \times [T_s]} = \{(k, \tau): k\in\Kk, \tau \in [0, T_s-1]\}$.  The constraint \eqref{too_hard1} is  highly non-linear and even takes exponential time  just to evaluate. 
% Such hardness poses immense challenge on the feasibility checking procedure as well as the algorithmic development for the control policy.     
 
\subsubsection{The Drift-plus-Penalty Policy}

 The problem $\OP$ can be considered as Network Utility Maximization (NUM) and solved via the DP framework \cite{neelybook1}. To achieve queue  stability, we consider the  Lyapunov potential:
 \begin{align}
 \label{quadra_lyapu}
    \Ll(\QQ(t)) =\frac{1}{2} \sum_{c\in \Cc}\sum_{i\in I}  Q^c_i(t)^2,
    \end{align}
and the drift of $\Ll(.)$ conditioned on the  queue lengths as:
\begin{align}
\label{lyapu_drift}
    \Delta (t) =\EE\big[ \Ll(\QQ(t+1)) -  \Ll(\QQ(t)) \big| \QQ(t) \big]. 
\end{align}
The next Lemma, whose proof is given in Appendix \ref{appen_lma_driftbound},  establishes the drift bound.
\begin{lemma}
\label{lma_driftbound}
For $B_1 = |\Ii||\Cc| A_{max}^2+ T_s^2 |\Kk|^2$, we have:
\begin{equation}
  \Delta (t) \leq B_1 +   \sum_{c\in \Cc}\sum_{i\in I} Q^c_i(t) \big( \lambda^c_i - \sum_{\tau = 0}^{T_s -1 }\sum_{k\in N(i)} r_{ki} p^c_{ki}(t, \tau)  \big). \label{drifting1}  
\end{equation}
\end{lemma}
Given the hyper-parameter $V>0$ that can be tuned to control the  utility-backlog trade-off, the conventional DP policy minimizes the upper-bound (given by \eqref{drifting1}) of the drift-plus-penalty $ \Delta (t) - V f(\pP(t))$, which is formally given by:
 \begin{equation}
        \pP^{\DP}(t) = \argmaxE_{ \pP(t) \in \Dd(\gma, \qQ) } h_t(\pP(t)),\label{drift_plus_penalty1}
    \end{equation}
 where the objective $h_t(.)$ is defined as:   
\begin{equation}
\label{drift_obj1}
    h_t(\pP(t)) = \sum_{c\in \Cc}\sum_{i\in I} Q^c_i(t) \sum_{\tau = 0}^{T_s -1 }\sum_{k\in N(i)} r_{ki} p^c_{ki}(t, \tau) + V  f( \pP(t)).
\end{equation}
% Thus, the DP policy finds the feasible decision $\pP(t)$ that solves:
%  \begin{equation}
%         \pP^{\DP}(t) = \argmaxE_{ \pP(t) \in \Dd(\gma, \qQ) } h_t(\pP(t)). \label{drift_plus_penalty1}
%     \end{equation}
While the DP policy is known to be an utility-optimal stabilizing algorithm \cite{neelybook1}, solving \eqref{drift_plus_penalty1} is equivalent to maximizing a concave objective over a domain comprised of   highly non-linear  constraints \eqref{too_hard1}, thereby rendering it intractable.
In the next section, we present our algorithmic development that deploys a linearization technique to approximate the probabilistic constraints, thereby resulting in an efficient policy with polynomial time complexity.

\subsection{The Modified Drift-plus-Penalty Algorithm}
\label{sec_MDP}
% \subsubsection{The linearized feasible Domain}
To address the aforementioned challenge, we extend the technique from the robust optimization literature \cite{Bertsimas2004ThePO} to approximate the probabilistic QoS constraint \eqref{qos_constraint1} by  linear constraints on  $\pP(t)$. 
As for $q^c_i = 0$, the constraint \eqref{qos_constraint1} trivially holds and corresponds to no QoS requirement, we consider the set: 
\begin{align}
 \Ss =\{ (i, c): i\in \Ii, c\in \Cc \text{ such that } q^c_i > 0\}   
\end{align}
 of clients and jobs that impose active QoS requirements.
The idea is that it is unlikely for all random variables $y^c_{ki}(t, \tau)$ to simultaneously deviate from their means, so resources are allocated to ``protect" situations where only a number  $\Gamma^c_i$, which is a hyper-parameter to be determined later, of variables  exceed their means. 
To formally describe the linearized constraint, we   define $\hat{p}^c_{ki}(t, \tau)= \max\{r_{ki} p^c_{ki}(t, \tau), 1- r_{ki}  p^c_{ki}(t, \tau)\} > 0$ and:
 % $B^c_i(\Gamma^c_i, \pP(t)) = \max\limits_{S^c_{ti}\cup\{(k_0, t_0)\}: S^c_{ti} \in 2^{N(i) \times [T_s]}, |S^c_{ti}| =\lfloor \Gamma^c_i\rfloor    }  \bigg\{  \sum_{(k,\tau)\in S_{ti}} \hat{p}^c_{ki}(t, \tau)+ (\Gamma^c_i - \lfloor \Gamma^c_i\rfloor) \hat{p}^c_{k_o i}(t, \tau_0)\bigg\}$
\begin{align}
\nonumber
    B^c_i(\Gamma^c_i, &\pP(t)) = \max\limits_{S^c_{ti}\cup\{(k_0, \tau_0)\}: S^c_{ti} \in 2^{N(i) \times [ T_s]}, |S^c_{ti}| =\lfloor \Gamma^c_i\rfloor    }  \\
    &\bigg\{  \sum_{(k,\tau)\in S^c_{ti}} \hat{p}^c_{ki}(t, \tau)+ (\Gamma^c_i - \lfloor \Gamma^c_i\rfloor) \hat{p}^c_{k_0 i}(t, \tau_0)\bigg\}, \label{big_b1}
\end{align}
which serves as the UCB term in the linearized constraint as:
% Then, the linearized constraint is written as:
% to approximate the probabilistic QoS constraint \eqref{qos_constraint1} for $ i\in I, c\in \Cc$ and given $\Gamma^c_i$ is written as:
\begin{equation}
\label{linearized_constraint1}
     |\Kk|T_s \gamma^c_i \leq \sum_{k\in N(i)} \sum_{\tau=0}^{T_s-1} r_{ki}  p^c_{ki}(t, \tau) - B^c_i(\Gamma^c_i, \pP(t)), \forall (i, c)\in \Ss.
\end{equation}
Note that the constraint \eqref{linearized_constraint1} can be equivalently converted into $|N(i)|T_s{|N(i)|T_s-1 \choose \lfloor \Gamma^c_i\rfloor}2^{\lfloor \Gamma^c_i\rfloor+1}$ linear constraints  of $\pP(t)$  by untangling all the max operators in \eqref{big_b1} and in $\hat{p}^c_{ki}(t, \tau)$. 
We  establish the probabilistic guarantee implied by \eqref{linearized_constraint1} in the following Theorem, whose proof is given in Appendix \ref{appen_concen_bound}.
\begin{theorem}
\label{concen_bound}
 If $\pP(t)\in [0, 1]^{T_s |E| |\Cc|}$ satisfies \eqref{linearized_constraint1}, we have:
 \begin{align}
 \label{chernoff}
    P\big( \frac{1}{|\Kk| T_s} \mu^c_i(t) >  \gamma^c_{i}  \big)  \geq  1-e^{-\frac{(\Gamma^c_i)^2}{2|\Kk| T_s} }, \forall (i, c)\in \Ss.
 \end{align}
 % From Theorem \ref{concen_bound},  
 To enforce the probabilistic QoS constraint \eqref{qos_constraint1}, we can  choose $\Gamma^c_i$ such that $1-e^{-\frac{(\Gamma^c_i)^2}{2|\Kk| T_s} } = q^c_i$, or equivalently:
\begin{align}
\label{gamma_choice}
     \Gamma^c_i = \sqrt{2 |\Kk| T_s \log\big( (1-q^c_i)^{-1} \big)} = \Theta(T_s^{1/2}).
\end{align}
\end{theorem}
% \begin{proofsketch}
%     Let $S^{c*}_{ti}$ and $(k^*_0, \tau^*_0)$ be the set and index respectively that achieves the maximum in \eqref{big_b1}, and consider $(k_r, \tau_r)=  \argminE_{k\in N(i), \tau \in [0, T_S-1]}  \hat{p}^c_{ki}(t, \tau)$. We also define the random variable $\eta^c_{ki}(t, \tau) = \frac{ -y^c_{ki}(t, \tau) + p^c_{ki}(t, \tau)}{\hat{p}^c_{ki}(t, \tau)} \in [ -1, 1]$, which is the ``normalized" version of $ 
% z^c_{ki}(t, \tau)=-y^c_{ki}(t, \tau) $ to be used later in the proof. We define the weights:
%     \begin{align*}
%         w^c_{ki}(t, \tau)  = \begin{cases}
%             1, \text{ if $(k, \tau)\in S^{c*}_{ti}$ }\\
%             \frac{\hat{p}^c_{ki}(t, \tau)}{ \hat{p}^c_{k_r i}(t, \tau_r)},\text{   if $(k, \tau)\in N(i) \times [T_s] \setminus S^{c*}_{ti}$ }
%         \end{cases}.
%     \end{align*}
%     Adapting the  technique as in the proof of Proposition 2 and Theorem 2 in \cite{Bertsimas2004ThePO}, we can prove the following two inequalities:
%     % P\big( \frac{1}{|\Kk| T_s} \mu^c_i(t) \leq  \gamma^c_{i}  \big) = 
%     \begin{align*}
%          &P\big( \frac{1}{|\Kk| T_s} \mu^c_i(t) \leq  \gamma^c_{i}  \big) = P\big(\sum_{k\in N(i)}  \sum_{\tau=0}^{T_s-1}z^c_{ki}(t, \tau)  \geq -|\Kk| T_s \gamma^c_i\big) \\
%          &\leq P\big(\sum_{k\in N(i)}  \sum_{\tau=0}^{T_s-1}   w^c_{ki}(t, \tau) \eta^c_{ki}(t, \tau)  \geq \Gamma^c_i\big)\leq e^{-\frac{(\Gamma^c_i)^2}{2|\Kk| T_s} },
%     \end{align*}
%     which   concludes the proof of the required statement \eqref{chernoff}.
% \end{proofsketch}

As discussed in   Section \ref{sec_hardness},  solving $\max_{ \pP(t) \in \Dd(\gma, \qQ) } h_t(\pP(t))$ under the original DP policy is  intractable due to the probabilistic constraint \eqref{qos_constraint1} in the feasible domain $\Dd(\gma, \qQ)$. However, by Theorem \ref{concen_bound}, this probabilistic constraint can be enforced by the linearized constraint \eqref{linearized_constraint1}. 
To make the optimization problem tractable, we consider the linearized feasible domain $ \widetilde{\Dd}(\gma, \qQ)$, which replaces the probabilistic  constraint \eqref{qos_constraint1} in the definition of $\Dd(\gma, \qQ)$ by the above linearized constraint \eqref{linearized_constraint1}:
% We also consider the linearized feasible domain $ \widetilde{\Dd}(\gma, \qQ)$, which replaces the probabilistic QoS constraint \eqref{qos_constraint1} in the definition of $\Dd(\gma, \qQ)$ by the above linearized constraint \eqref{linearized_constraint1}:
\begin{align}
\label{def_approx_region}
    \widetilde{\Dd}(\gma, \qQ) =  \big\{ \pP(t) \in [0, 1]^{T_s |E| |\Cc|}\text{: such that \eqref{prob_constraint}, \eqref{pmax1}, \eqref{linearized_constraint1}}\big\}.
\end{align}
 % While  the linearized constraint \eqref{linearized_constraint1} provably implies the probabilistic QoS constraint with the above choice of $\Gamma^c_i$, 
 While  the linearized constraint \eqref{linearized_constraint1}  implies the probabilistic QoS constraint, 
 it contains $ B^c_i(\Gamma^c_i, \pP(t))$, which is a max operator over   exponentially many linear constraints. 
 In the next Theorem, whose proof is deferred to Appendix \ref{appen_thm_dual},
 we  derive an equivalent form for  $ \widetilde{\Dd}(\gma, \qQ)$  via duality, where the problem size  is of polynomial order. 

\begin{theorem}
\label{thm_dual}
An equivalent form for   $\widetilde{\Dd}(\gma, \qQ)$ is given by:
\begin{align*}
    \widetilde{\Dd}(\gma, \qQ) =  \big\{ &\pP(t) \in [0, 1]^{T_s |E| |\Cc|}\text{: such that \eqref{prob_constraint}, \eqref{pmax1},} \\
    &\quad \quad  \quad |\Kk|T_s \gamma^c_i \leq \sum_{k\in N(i)} \sum_{\tau=0}^{T_s-1} r_{ki} p^c_{ki}(t, \tau)- \Gamma^c_i s^c_i  \\
      &\quad \quad \quad \quad \quad -  \sum_{k\in N(i)} \sum_{\tau=0}^{T_s-1} v^c_{ki \tau},  \forall (i, c)\in \Ss, \\
      & \quad \quad  \quad  s^c_i \geq 0, \forall  (i, c)\in \Ss,\\
      &\text{ and } \forall   (i, c)\in \Ss, k\in N(i), \tau \in[T_s]:\\ 
      & \quad \quad  \quad  s^c_i + v^c_{ki \tau} \geq  r_{ki} p^c_{ki}(t, \tau),\\% \forall   i\in I,k\in N(i), c\in \Cc, \tau \in[0, T_s-1], \\
     &  \quad \quad  \quad s^c_i + v^c_{ki \tau} \geq 1- r_{ki} p^c_{ki}(t, \tau), \\%  \forall i\in I,k\in N(i), c\in \Cc, \tau \in[0, T_s-1],  \\
     &  \quad \quad  \quad v^c_{ki \tau} \geq 0\big\}.%,  \forall i\in I,k\in N(i) ,c\in \Cc, \tau \in[0, T_s-1],  \\
\end{align*}
\end{theorem}
% \begin{proofsketch}
% The proof is done by first recognizing that $ B^c_i(\Gamma^c_i, \pP(t))$ is provably equivalent to:
%   \begin{align*}
%     B^c_i(\Gamma^c_i, \pP(t))  = &\max \quad   \sum_{k\in  N(i)} \sum_{\tau=0}^{T_s-1} \hat{p}^c_{ki}(t, \tau) \cdot z_{k \tau}\\
%      \text{such that }  &  \sum_{k\in  N(i)} \sum_{\tau=0}^{T_s-1} z_{k \tau}  \leq \Gamma^c_i, \\
%      & 0\leq z_{k \tau} \leq 1,  \forall k\in N(i), \tau \in[T_s],
%  \end{align*} 
%  and then deriving the dual form of the above linear program to be plugged back into $ \widetilde{\Dd}(\gma, \qQ)$ to conclude the proof. The full proof is given in Appendix \ref{appen_thm_dual}.  
% \end{proofsketch}

% \subsubsection{The Modified Drift-plus-Penalty Policy}
% \label{sec_the_algo}
 Finally , we propose the Modified Drift-plus-Penalty (MDP) policy that replaces  $\Dd(\gma, \qQ)$ by the linearized feasible domain $\widetilde{\Dd}(\gma, \qQ)$ in the optimization problem \eqref{drift_plus_penalty1} as follows:
 \begin{equation}
        \pP^{\MDP}(t) = \argmaxE_{ \pP(t) \in \widetilde{\Dd}(\gma, \qQ) } h_t(\pP(t)). \label{drift_plus_penalty2}
    \end{equation}
The full  policy is depicted in Algorithm \ref{alg:mdp}. As $\widetilde{\Dd}(\gma, \qQ) $ is a polyhedron of polynomial size by Theorem \ref{thm_dual}, the MDP policy boils down to solving convex optimization \eqref{drift_plus_penalty2} (i.e.  concave maximization)  over linear constraints and thus admits  efficient polynomial-time solvers  \cite{Boyd_Vandenberghe_2004}. 
% \cite{Boyd_Vandenberghe_2004, bertsekas2009convex}. 
% In  Section \ref{sec_guarantees_MDP}, we further derive theoretical guarantees for MDP on its utility-optimality and stability.  
\begin{algorithm}
\caption{Modified Drift-plus-Penalty (MDP) policy}
\label{alg:mdp}
\KwIn{$T_s, \gma$ and $\qQ$.}
\For{$t=0, ..., T-1$}{
 Solve for the link activation probabilities $\pP^{\MDP}(t) = \argmaxE_{ \pP(t) \in \widetilde{\Dd}(\gma, \qQ) } h_t(\pP(t))$. \\
\For{$\tau=0, ..., T_s-1$, within each time frame}{
 \textbf{[Scheduling]} Activate links $\{d^c_{ki}(t, \tau)\}_{k\in \Kk, i\in \Ii, c\in \Cc}$ (based on  $\pP^{\MDP}(t, \tau)$) for serving jobs in  time slot $\tau$. \\
}
}
% $($for $t\in (\tau_j, \tau_{j+1})$, $D_i$ executes as follows$)$\\
 Observe the arrivals $\{a^c_i(t)\}_{i\in\Ii, c\in \Cc}$ and compute the services $\{ \mu^c_i(t)\}_{i\in\Ii, c\in \Cc}$ via \eqref{service_unit1}.\\
 Update the queues as 
$ Q^c_i(t+1) = \big( Q^c_i(t) + a^c_i(t) -  \mu^c_i(t)\big)^+, \forall i\in \Ii, c\in \Cc.$
\end{algorithm}

Our QoS model  permits randomized scheduling via  $\pP(t, \tau)$ in the policy parameterization.
 This raises the question of whether there exists a \textit{feasible} deterministic policy making deterministic scheduling decisions in every time slot, i.e. $ p^c_{ki}(t, \tau) \in \{0, 1\}$, which  resembles a property (i.e. existence of feasible and  optimal policy that is deterministic) inherent in many scheduling problems in the  literature  (e.g. \cite{timevarying_network1, quang_linear_schedule}). 
 % (e.g. \cite{UMW, neelybook1}). 
 We answer this question negatively by constructing an example (Appendix \ref{appen_rand_scheduling}) in which the class of deterministic policies  is insufficient to meet our probabilistic  constraints. In contrast, our linearized feasible domain and thus the  MDP policy can support  these  QoS requirements. In the following sections, we further derive theoretical guarantees for the MDP policy as well as  how the proposed QoS model can provably provide throughput and delay guarantees. 
%  A summary of the main results is given in Table \ref{tab:summary}.}
% \begin{table}[h]
%     \centering
%     \renewcommand{\arraystretch}{1.2}
%     \setlength{\tabcolsep}{4pt} % Adjusts column spacing
%     \small % Reduce font size if needed
%     \begin{tabular}{|p{5.3cm}|p{2.6cm}|} % Adjust width of text column
%         \hline
%         \textbf{ Description} & \textbf{Main Results} \\
%         \hline
%       Stability region characterization (Section \ref{sec_feasibility})& Corollary \ref{corol_approx_region} \\
%       \hline
%         Utility-optimality and stability guarantees of the MDP policy (Section \ref{sec_optimality}) & Theorem \ref{thrm_mdp}, Corollary \ref{corol_V_choice}  \\
%         \hline
%         Application of the QoS framework (Section \ref{sec_implication_qos}) & Corollary \ref{corol_throughput_bound}, Theorem \ref{thrm_delay_bound} \\
%         \hline
%     \end{tabular}
%     \caption{Overview of Main Results}
%     \label{tab:summary}
%     \vspace{-0.8cm}
% \end{table}

\section{ Guarantees for Modified Drift-plus-Penalty} 
\label{sec_guarantees_MDP}

In this section, we show that the MDP policy can support a stability regions arbitrarily close to the true stability region and achieve a vanishing optimality gap, as the time frame size $T_s$ increases. 
In Section\ref{sec_feasibility}, we first establish the stability region characterization of the QoS-based infrastructure sharing problem 
$\OP$ (see Corollary \ref{corol_approx_region}). Then,  we prove the utility-optimality and stability guarantees (Theorem \ref{thrm_mdp} and Corollary \ref{corol_V_choice}) in Section \ref{sec_optimality}.

% We first establish the stability region characterization of the QoS based infrastructure sharing problem $\OP$ (see Corollary \ref{corol_approx_region}) in  Section \ref{sec_feasibility} and then prove the utility-optimality and stability guarantees (Theorem \ref{thrm_mdp} and Corollary \ref{corol_V_choice}) in Section \ref{sec_optimality}. 

\subsection{Stability Region Characterization}
\label{sec_feasibility} 

The MDP policy replaces the original feasible domain $\Dd(\gma, \qQ)$ by the linearized feasible domain $\widetilde{\Dd}(\gma, \qQ)$ in the optimization \eqref{drift_plus_penalty2} for tractability. 
While adopting $\widetilde{\Dd}(\gma, \qQ)$  ensures the probabilistic QoS constraints, the question of how close this approximation is remains. For the regime of large $T_s$ within interest (as discussed in Section \ref{sec_sys_model}), by leveraging the properties of Poisson Binomial distribution, we characterize the tightness of this approximation of order $O\big(\frac{1}{\sqrt{T_s}}\big)$ in the following Theorem, whose proof is deferred to Appendix \ref{appen_approx_bound}.
% \vspace{-0.3cm}
\begin{theorem}
\label{thrm_approx_bound}
% \begin{align}
% \label{approximate_constant}
%    K_1 = \frac{ \sqrt{2  \log\big( (1-\|\qQ\|_\infty)^{-1} \big)}}{\sqrt{|\Kk| }} + \frac{2\sqrt{p_{max}\pi e^3}}{|\Kk| } .
% \end{align}
If the following condition holds:
\begin{align}
\label{assum_charac}
    T_s > \frac{0.795^2}{(1-r_{max}) |\Kk| \gamma^c_i (q^c_i)^3}\triangleq K^c_i,\quad \forall (i, c)\in \Ss,
\end{align}
then we have the characterization:
\begin{align}
\label{characterize_region}
   \widetilde{\Dd}(\gma, \qQ)  \subseteq  {\Dd}(\gma, \qQ) \subseteq  \widetilde{\Dd}(\gma - \frac{K_1}{\sqrt{T_s}} \one, \qQ),
\end{align}
 where the constant $K_1$ is defined as:
    \begin{align}
    \label{k1_const}
         K_1 \triangleq \frac{ \sqrt{2  \log\big( (1-\|\qQ\|_\infty)^{-1} \big)}}{\sqrt{|\Kk| }} + \frac{2\sqrt{U_{max} \pi e^3}}{|\Kk| }.
    \end{align}
\end{theorem}

Theorem \ref{thrm_approx_bound} establishes that, despite its hardness, we can linearize the probabilistic QoS constraint with vanishing approximation error $O\big(\frac{1}{\sqrt{T_s}}\big)$.
% as in  \eqref{characterize_region}. 
We   consider the  approximate stability region  as follows. 
% Furthermore, the mild condition \eqref{assum_charac} can be  satisfied in realistic scenario as shown in Section \ref{sec_example}. 

% Motivated by the results from the previous Section, we proceed to develop an efficient feasibility checking procedure that validates the feasibility of a polyhedron of linear constraints.

\begin{definition}[Approximate Stability Region]
The approximate stability region  given $\gma$ and $\qQ$ is defined as:
    \begin{align*}
    \widetilde{\Lmb}(\gma, \qQ)= \{&\lmb = \{\lambda^c_i\}_{i\in \Ii, c\in \Cc}\in \Br_{\geq0}^{|\Ii||\Cc|}:  \text{$\exists \pi$ such that } \\
    & \pP^\pi(t) \in  \widetilde{\Dd}(\gma, \qQ), \forall t\geq 0 \text{ and } \EE[\QQ^\pi(T)] = o(T)  \}.
\end{align*}
\end{definition}

% Nevertheless, performing feasibility checking on $\widetilde{\Lmb}(\gma, \qQ)$ remains challenging.
In the following Theorem, whose proof is deferred to Appendix \ref{appen_thm_charac_stability_region}, we can further derive  $ \widetilde{\Lmb}(\gma, \qQ)$ in the form of polyhedron. Intuitively, it means that if $\lmb \in  \widetilde{\Lmb}(\gma, \qQ)$, there exists an optimal policy making stationary per-time-frame decision, which is a known property for linearly constrained control problem \cite{neelybook1}; thus, we drop the time frame $t$ from the vector $\pP$  for clarity. 
\begin{theorem}
\label{thm_charac_stability_region}
% The approximate stability region   $ \widetilde{\Lmb}(\gma, \qQ)$ can be characterized as:  
We have the following characterization of $ \widetilde{\Lmb}(\gma, \qQ)$:
\begin{align}
\nonumber
    \widetilde{\Lmb}(\gma, \qQ)= \{\lmb& = \{\lambda^c_i\}_{i\in \Ii, c\in \Cc}\in \Br_{\geq0}^{|\Ii||\Cc|}: \\
    &\lambda^c_i \leq \sum_{k\in N(i)} \sum_{\tau=0}^{T_s-1} r_{ki} p^c_{ki}(\tau) , \forall i\in \Ii, c\in \Cc, \label{super_charac} \\
    \nonumber
    &   \pP = \{p^c_{ki}(\tau) \}_{\tau\in[0, T_s-1], (k,i) \in E, c\in \Cc} \in  \widetilde{\Dd}(\gma, \qQ)\}.  
\end{align}
\end{theorem}
% Our feasibility checking procedure thus performs validating whether $\lmb \in  \widetilde{\Lmb}(\gma, \qQ)$ using the characterization of $\widetilde{\Lmb}(\gma, \qQ)$ from Theorem \ref{thm_charac_stability_region}, which is a polyhedron of polynomial size by Theorem \ref{thm_dual}. 
Finally, we show that the approximate stability region $ \widetilde{\Lmb}(\gma, \qQ)$ is close to the actual stability region $\Lmb(\gma, \qQ)$  as demonstrated in the next Corollary directly following from Theorem \ref{thrm_approx_bound}.
% Then from Theorem \ref{thrm_approx_bound}, we obtain the following Corollary on the relation between $ \widetilde{\Lmb}(\gma, \qQ)$ and $ {\Lmb}(\gma, \qQ)$. 
% A direct Corollary from Theorem \ref{thrm_approx_bound} relating $ \widetilde{\Lmb}(\gma, \qQ)$ and $ {\Lmb}(\gma, \qQ)$ 
\begin{corollary}
\label{corol_approx_region}
    Under the condition \eqref{assum_charac} and for $K_1$ in \eqref{k1_const},
    \begin{align}
         \widetilde{\Lmb}(\gma, \qQ)  \subseteq  {\Lmb}(\gma, \qQ) \subseteq  \widetilde{\Lmb}(\gma - \frac{K_1}{\sqrt{T_s}} \one, \qQ).
    \end{align}
\end{corollary}

% \subsection{Efficient Algorithm for the Control Problem}
% \label{sec_MDP}
% We proceed to develop the efficient algorithm solving the control problem $\OP$ for the approximate stability region $ \widetilde{\Lmb}(\gma, \qQ)$  supported  by our feasibility checking procedure. 

\subsection{Utility-Optimality and Stability Results}
\label{sec_optimality}
% \subsubsection{Performance Analysis}
We  proceed to show that the MDP policy attains competitive performance guarantees on the supported   stability  region that is within $O\big(\frac{1}{\sqrt{T_s}}\big)$ distance from the true stability region. 
For analysis, we make the following  assumption that  $ \widetilde{\Dd}(\gma, \qQ)$ has an interior point that supports the considered arrival rate $\lmb$ in view of Theorem \ref{thm_charac_stability_region}.
\begin{assumption}[Slater Condition]
\label{assum_slater}
% The domain $ \widetilde{\Dd}(\gma, \qQ)$ has an interior point, i.e.
For arrival rate $\lmb$, there exists  $\zeta >0$ and $\pP \in \widetilde{\Dd}(\gma, \qQ)$ such that \eqref{super_charac} holds 
and $\forall (i, c)\in \Ss$:
\begin{align}
\label{slater}
       |\Kk|T_s (\gamma^c_i + \zeta ) \leq \sum_{k\in N(i)} \sum_{\tau=0}^{T_s-1} r_{ki}  p^c_{ki}(\tau) - B^c_i(\Gamma^c_i, \pP(t)).
\end{align}
\end{assumption}
% \cite{Boyd_Vandenberghe_2004, bertsekas2009convex}
The  Slater condition \cite{Boyd_Vandenberghe_2004} is a natural assumption in practice. 
The next Lemma establishes the sub-optimality gap between MDP policy and DP policy caused by replacing ${\Dd}(\gma, \qQ)$ with $\widetilde{\Dd}(\gma, \qQ)$.
% the linearized domain $\widetilde{\Dd}(\gma, \qQ)$. 

\begin{lemma}
\label{lma_obj_gap}
Given fixed $\QQ(t)$, and under Assumption \ref{assum_slater} and condition \eqref{assum_charac},
% and the following condition:
% \begin{align}
% \label{slater_cond}
%     T_s \geq {4 K_1^2}/{\zeta^2},
% \end{align}
the objective gap for solving \eqref{drift_plus_penalty1} and \eqref{drift_plus_penalty2} is bounded by:
% the objective gap between the two policies solving \eqref{drift_plus_penalty1} and \eqref{drift_plus_penalty2} is bounded by:
\begin{align}
\nonumber
    0\leq h_t(\pP^\DP(t)) -  h_t(\pP^\MDP(t)) \leq& K_2  \|\QQ(t)\|_{\infty}\sqrt{T_s}  \\
    &+  \frac{V K_2 (f_{max} -f_{min})}{ |\Kk| \sqrt{T_s}},\label{obj_gap1}
\end{align}
where  $K_2 \triangleq \frac{ |\Cc||\Ii| |\Kk| K_1 }{ \zeta}$ with   $K_1$ defined in \eqref{k1_const}.
\end{lemma}
\begin{proof} 
% Since     $\widetilde{\Dd}(\gma, \qQ) \subseteq {\Dd}(\gma, \qQ)$, we have $h_t(\pP^\MDP(t))  \leq h_t(\pP^\DP(t)) $ and hence the first part of the required statement. To prove the remaining part, 
We  consider:
\begin{align}
\gma' &= \gma - \frac{K_1}{\sqrt{T_s}} \one, \quad  \pP'(t)  = \argmaxE_{ \pP(t) \in \widetilde{\Dd}(\gma', \qQ) } h_t(\pP(t)). \label{lma_obj_gap_proof0}
\end{align}
By Theorem \ref{thrm_approx_bound}, we  have $ \widetilde{\Dd}(\gma, \qQ)  \subseteq  {\Dd}(\gma, \qQ) \subseteq  \widetilde{\Dd}(\gma', \qQ)$. This implies that  $h_t(\pP^\MDP(t)) \leq h_t(\pP^\DP(t))  \leq h_t(\pP'(t)) $ and then:
\begin{align}
% \label{lma_obj_gap_proof1}
\nonumber
    0\leq h_t(\pP^\DP(t)) -  h_t(\pP^\MDP(t)) \leq h_t(\pP'(t))   -  h_t(\pP^\MDP(t)).
\end{align}
It is thus left to bound $h_t(\pP'(t)) - h_t(\pP^\MDP(t))$.
We note that $\pP'(t)$ is the optimal solution to the optimization problem \eqref{lma_obj_gap_proof0} with the same objective and linear constraints as that \eqref{drift_plus_penalty2} of $\pP^\MDP(t)$  except the constraint \eqref{linearized_constraint1}, where $\gma$ is replaced by the ``perturbed" $\gma'$.  We would then be able to bound the objective gap in terms of $\|\gma' - \gma\|_1$ in the next Lemma whose proof is given in Appendix \ref{appen_lma_perturbed_obj}. 
\begin{lemma}
\label{lma_perturbed_obj}
 % Under \eqref{slater_cond} and g
 Under Assumption \ref{assum_slater} and given $ K_3\triangleq  \frac{\| \QQ(t) \|_\infty |\Kk| T_s  }{\zeta} + \frac{V (f_{max} - f_{min})}{ \zeta}$,  we have:
 \begin{align}
 \label{lma_obj_gap_proof2}
   h_t(\pP'(t))   -  h_t(\pP^\MDP(t)) \leq K_3  \|\gma' - \gma\|_1 . 
 \end{align}
\end{lemma}
% The condition \eqref{slater_cond} ensures that if the Slater condition holds for  $\widetilde{\Dd}(\gma, \qQ)$ with constant $\zeta$, it would also hold for $\widetilde{\Dd}(\gma', \qQ)$ with constant $\frac{\zeta}{2}$, which is used in the proof. 
% \ref{assum_slater}
% Note that since $\gamma^c_i = {\gamma^c_i}' + \frac{K_1}{\sqrt{T_s}} > {\gamma^c_i}' $, if the Slater condition holds for  $\widetilde{\Dd}(\gma, \qQ)$, it would also hold for $\widetilde{\Dd}(\gma', \qQ)$. 
By definition of $\gma'$, we have $\|\gma' - \gma\|_1 = \frac{|\Cc||\Ii| K_1}{\sqrt{T_s}}$. Plugging this into \eqref{lma_obj_gap_proof2}, we conclude the proof of the Lemma.
\end{proof}

Finally, we establish the main Theorem on the performance guarantee for the $\MDP$ policy, which characterizes the trade-off between queue stability and utility optimality.

% $\varepsilon = \frac{2 K_2}{\sqrt{T_s}|\Kk|\gamma_{min}}$ 
% Let  $q_{min} = \min_{i\in \Ii, c\in \Cc} \{q^c_i\}$.
\begin{theorem}
\label{thrm_mdp}

 Under Assumption  \ref{assum_slater} and condition \eqref{assum_charac}, and given $\lmb   \in (1- \varepsilon)   \widetilde{\Lmb}(\gma, \qQ)$ with $\varepsilon = \frac{K_4}{\sqrt{T_s}}$, we have:
 \begin{align}
    &\EE[\|\QQ^{\MDP}(T)\|_1] = O\big( \sqrt{VT}+T_s\sqrt{ T}\big), \label{thrm_mdp_eq1}\\
     &0 \leq F(\pP^*) - F(\pP^\MDP) = O\big( \frac{T_s^2}{V} +  \frac{ 1}{\sqrt{T_s}} \big), \label{thrm_mdp_eq2}
 \end{align}
 where $\pP^*$ is from an optimal policy for  $\OP$ and  $K_4 = \frac{2 K_2}{|\Kk|\delta_{min}}$. 
\end{theorem}
\begin{proofsketch}
Besides the optimal policy $\pi^*$ (with decision variables $\pP^*$) that supports the arrival rates $\lmb$, we can show the existence of an optimal policy $\pi_\varepsilon$ that supports  the ``hypothetical" arrival rates $\lmb_\varepsilon = \frac{1}{1-0.5\varepsilon} \, \lmb$ and makes stationary per-time-frame decision $\pP^{\pi_\varepsilon}(t) = \pP$  satisfying the constraints of the characterization in Theorem \ref{thm_charac_stability_region}, i.e. $\pP\in  \widetilde{\Dd}(\gma, \qQ)$ and:
\begin{align}
\label{thrm_mdp_proof1}
    [\lmb_\varepsilon]^c_i = \frac{\lambda^c_i}{1-0.5\varepsilon}  \leq \sum_{k\in N(i)} \sum_{\tau=0}^{T_s-1} r_{ki}  {p^c_{ki}}(\tau) , \forall i\in \Ii, c\in \Cc.
\end{align}
Given $\varepsilon = \frac{K_4}{\sqrt{T_s}}$, we can prove that $F(\pP^*) - F(\pP^{\pi_\varepsilon}) = O\Big(\frac{1}{\sqrt{T_s}}\Big)$. Thus,
\begin{align}
\label{thrm_mdp_proof2}
 0 \leq F(\pP^*) - F(\pP^\MDP) \leq  F(\pP^{\pi_\varepsilon})  - F(\pP^\MDP) +  O\Big(\frac{1}{\sqrt{T_s}}\Big).
\end{align}
By \eqref{drift_plus_penalty1}, we have that $h_t(\pP^\DP(t)) \geq h_t(\pP^{\pi_\varepsilon}(t))$. Combining this with Lemma \ref{lma_obj_gap} gives us the bound on the objective gap $  h_t(\pP^{\pi_\varepsilon}(t))-h_t(\pP^\MDP(t))$. Plugging such bound into \eqref{drifting1} in view of \eqref{thrm_mdp_proof1}, we will be able to show that:
\begin{equation*}
    \Delta^\MDP(t) - V f(\pP^\MDP(t))\leq B_1 -V f(\pP^{\pi_\varepsilon}(t))+   \frac{V K_2 (f_{max} -f_{min})}{ |\Kk| \sqrt{T_s}} .
\end{equation*}
Taking iterated expectation of the above with respect to $\QQ^\MDP(t)$ (recall that $\Delta(t)$ is the expected drift conditioned on queue lengths) and telescoping, we  obtain after simple algebras that:
\begin{align}
    &\EE[ \Ll(\QQ^\MDP(t)) ] =  O(VT + T_S^2 T), \label{thrm_mdp_proof3}\\
      &F(\pP^{\pi_\varepsilon})  - F(\pP^\MDP) = O\Big( \frac{T_s^2}{V} +  \frac{ 1}{\sqrt{T_s}} \Big).  \label{thrm_mdp_proof4}
\end{align}
Via Holder and Jensen inequalities, we  show that $\EE[\|\QQ^{\MDP}(T)\|_1]  \leq \sqrt{2|\Kk| |\Ii| \EE[ \Ll(\QQ^\MDP(t)) ]} $, which, combined with \eqref{thrm_mdp_proof3}, concludes the proof of \eqref{thrm_mdp_eq1}. Plugging \eqref{thrm_mdp_proof4} into \eqref{thrm_mdp_proof2}, we  conclude the proof of \eqref{thrm_mdp_eq2}. The full proof is given in Appendix \ref{appen_thrm_mdp}.
\end{proofsketch}

From  Corollary \ref{corol_approx_region}, we  know that $ \widetilde{\Lmb}(\gma, \qQ) =  \Lmb(\gma, \qQ)- O\big(\frac{1}{\sqrt{T_s}}\big) \one $. 
Theorem \ref{thrm_mdp} shows that the  $\MDP$ policy can support the stability region  $(1-\varepsilon) \widetilde{\Lmb}(\gma, \qQ)$ with $\varepsilon = \frac{K_3}{\sqrt{T_s}}$, which thus approaches arbitrarily close to the true stability region $\Lmb(\gma, \qQ)$ as $T_s$ grows to infinity.
% Given   and $ \widetilde{\Lmb}(\gma, \qQ) =  \Lmb(\gma, \qQ)- O\big(\frac{1}{\sqrt{T_s}}\big) \one $ (by Corollary \ref{corol_approx_region}), which only shrinks the original stability region  $\Lmb(\gma, \qQ)$ by a distance of  $O\big(\frac{1}{\sqrt{T_s}}\big)$. 
Finally, the next Corollary summarizes the choice of $V$ for the $\MDP$ policy to solve the   problem $\OP$ up to  $O\big(\frac{1}{\sqrt{T_s}}\big)$ optimality gap.
\begin{corollary}
\label{corol_V_choice}
    For any choice of $V= \Theta(\min\{T_s^{2.5}, T^\beta\})$ with $\beta \in (0, \frac{1}{2})$, we have the  stability guarantee $\EE[\QQ^{\MDP}(T)] = o(T)$ and the optimality gap $ F(\pP^*) - F(\pP^\MDP) = O\big(\frac{1}{\sqrt{T_s}}\big)$.
\end{corollary}

\section{ Application of the QoS Framework}
\label{sec_qos_application}
In this section, we demonstrate the application of our proposed QoS framework in real-world scenarios.

\subsection{Throughput and Delay Guarantees}
\label{sec_implication_qos}
While our QoS model is designed to be generalized enough to handle a variety of QoS metrics, we show how it can be leveraged to directly imply the conventional throughput and delay guarantees. First, as in the following Corollary, whose proof is given in Appendix \ref{appen_corol_throughput_bound}, we can derive directly from the probabilistic QoS constraint \eqref{qos_constraint1} a bound on guaranteed service throughput in  expectation.

\begin{corollary}
\label{corol_throughput_bound}
The expected total  OTT service for type-$c$ jobs of client $i$ is bounded by:
\begin{align}
     \EE[\mu^c_i(t)] \geq |\Kk| T_s \gamma^c_i q^c_i, \quad \forall i\in \Ii, c\in \Cc. \label{qos_throughput1}
\end{align}
\end{corollary}
The above result demonstrates how $T_s, \gamma^c_i$ and $q^c_i$ can be chosen to enforce overall  service guarantee for throughput-demanding applications. Besides throughput, a common QoS metric required by time-sensitive applications is delay, i.e. the requested job should be served within the required amount of time. 
For applications with deterministic job arrivals such as streaming with stringent delay requirement, the time frame size can be chosen to meet the latency requirement and $q^c_i$ is set close to 1 to ensure negligible failure probability, while $\gamma^c_i$  is chosen in view of \eqref{qos_throughput1} to sustain the required throughput. Such guaranteed delivery ratio with probability 1 was considered in \cite{qos1}.
On the other hand, many applications such as environmental monitoring  only demand soft delay requirements, whereby jobs are allowed to be served in multiple time-frames within  long-term average delay constraint. 
We establish in the next Theorem an upper bound for the average queueing delay $W^c_i$ that is only dependent on $\gamma^c_i, q^c_i$ and the first two  moments of the arrival process. Moreover, real-world applications with stochastic job arrivals either have readily available first and second moments  from  measurement data or can be modeled as Poisson processes, which  inherently provide well-defined first two moments.

% \cite{poisson_wireless1}
% , and holds for any feasible decision variables $\pP(t)$.

\begin{theorem}
\label{thrm_delay_bound}
 Under the assumption that the second moment of the arrival process is finite and for any feasible scheduling  probabilities $\pP(t) \in  \Dd(\gma, \qQ)$, we have the following bound on the average queueing delay of type-$c$ jobs at client $i$ if $\gamma^c_i q^c_i >  \frac{\lambda^c_i}{|\Kk| T_s}$:
 \begin{align}
\label{queue_delay_bound}
    W^c_i \leq \frac{\EE[a^c_i(t)^2] + T_s^2 U_i^2 + T_s U_i- 2\lambda^c_i |\Kk| T_s \gamma^c_i q^c_i}{ 2 \lambda^c_i (  |\Kk| T_s \gamma^{c}_{i}q^c_i  - \lambda^c_i )}.
    % \leq  \bigg( \frac{1}{2|\Kk|T_s \gamma^c_i q^c_i} +\frac{1}{2}p_{max}^2 T_s^2\bigg) \frac{\lambda^c_i}{|\Kk| T_s \gamma^c_i q^c_i -\lambda^c_i}.
\end{align}
Furthermore, a sufficient condition for $W^c_i \leq W^*$, i.e. the queueing delay is bounded by some given $W^* >0$, is:
 \begin{align}
\gamma^c_i q^c_i   \geq& \frac{ \EE[a^c_i(t)^2] + T_s^2 U_i^2 + T_s U_i + 2(\lambda^{c}_i)^2 W^*   }{2\lambda^c_i |\Kk | T_s (W^*+1)}.\label{queue_delay_suff1}
% \geq&   \frac{ p_{max}^2 T_s^2 \lambda^c_i + 2W^* \lambda^c_i }{4W^* |\Kk| T_s}\\
% &+  \frac{ \sqrt{(p_{max}^2 T_s^2 \lambda^c_i + 2W^* \lambda^c_i)^2 + 8W^* |\Kk| T_s \lambda^c_i}}{4W^* |\Kk| T_s}. \label{queue_delay_suff1}
\end{align}
\end{theorem}

\begin{figure*}[h]
\centering
\begin{subfigure}{0.32\textwidth}
\includegraphics[width=\textwidth]{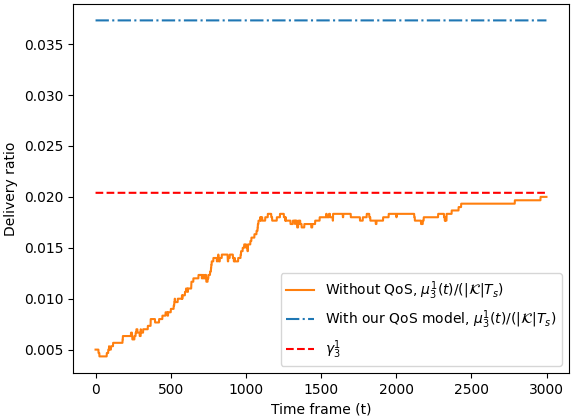}
\caption{Video streaming (type-$1$ job) with requirement of full throughput delivery (i.e.$\gamma^1_1\geq 0.0204$) within 150 ms latency. Conventional method from the literature fails to support the required throughput.}
\label{fig: vid_stream}
\end{subfigure}
\hspace{0.2cm}
\begin{subfigure}{0.31\textwidth}
\includegraphics[width=\textwidth]{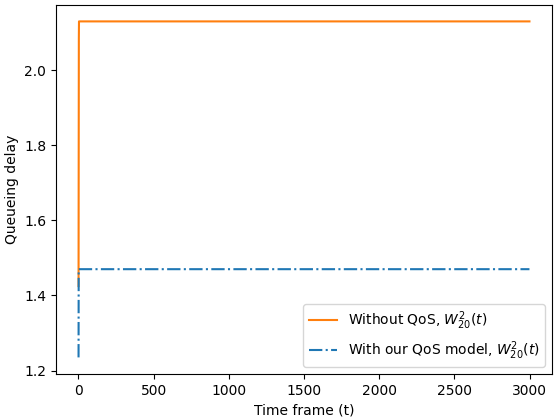}
\caption{Environmental monitoring (type-$2$ job) with delay requirement of 4 minutes. 
The actual delay achieved by our model is even lower than that of the baseline without QoS.}
\label{fig: env_moni}
\end{subfigure}
\hspace{0.2cm}
\begin{subfigure}{0.32\textwidth}
\includegraphics[width=\textwidth]{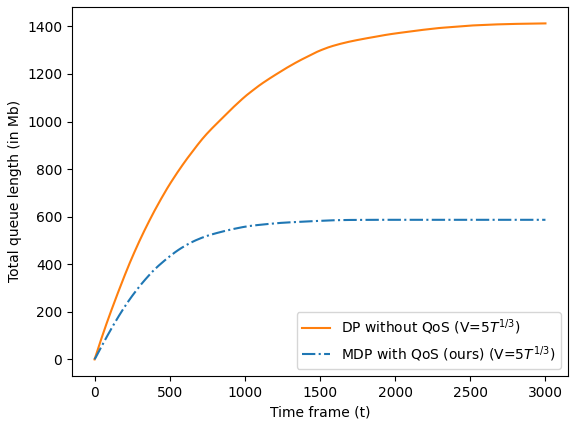}
\caption{Data backup (type-$3$ job) with queue stability requirement.
While both baselines achieve total queue stability,  the DP policy without QoS provisioning has much higher queue backlogs. }
\label{fig: queue}
\end{subfigure}
\caption{ The MDP with our QoS model achieves better QoS than the DP without QoS in all applications.}
\label{fig: qos_model_strong}
\end{figure*}
Theorem \ref{thrm_delay_bound} (proof in Appendix \ref{appen_thrm_delay_bound}) not only provides the delay guarantees via \eqref{queue_delay_bound} that can be given to clients in the SLA phase before the actual  operation, but also elucidates the choices of $\gamma^c_i$ and $q^c_i$ via \eqref{queue_delay_suff1}  to meet the desired delay requirements.

\subsection{Application Examples}
\label{sec_example}
% We demonstrate the adaptation of our QoS model to two examplary realistic applications with different types of QoS requirements: real-time video streaming (VS) \cite{video_latency1} and environmental monitoring (EM) via sensors \cite{minutes_delay1}. We use a simplistic setting of two InPs, comprised of a 5G cellular and a cloud center, fully connected to two clients, with link capacity of 5 Mbps. 

We demonstrate the  applicability of our QoS model on three real-world applications:  real-time video streaming  \cite{video_latency1}, environmental monitoring  \cite{minutes_delay1} and data backup. We use a  setting of ten InPs, which are all 5G cellular base stations and fully connected to twenty clients, with link capacity of 10 Mbps. 
% Five of the 5G networks offer a channel quality of $r_{ki} = 0.8$ for every client, while the remaining five provide more stable connections at $r_{ki} = 0.9$. 
Five of the 5G networks, operating as secondary connections, offer a channel quality of 
$r_{ki}= 0.8$ for every client, while the remaining five, acting as primary networks, provide more stable connections at at $r_{ki} = 0.9$. The transmission rates of the clients are set to $U_i = 1$ for all clients.
% The 5G cellulars  have channel qualities of $r^c_{ki} = 0.9$ for every client with any type of job connected to them. The cloud centers support more stable connections with channel qualities of $r^c_{ki} = 0.95$ for every client with any type of job connected to them. 
% The 5G cellular  supports  switching frequency of 0.5 ms \cite{5g_speed1}  for scheduling than that of cloud infrastructure of up to microseconds; thus, we set the actual length of an operational time slot to 0.5 ms,
We set the actual length of an operational time slot to 0.5 ms, which is the switching frequency supported by 5G cellular suppliers \cite{5g_speed1}.
The assumption on unit link capacity per operational time slot means that one data unit in our setting is equal to $10 \cdot 10^{-3}\cdot 0.5= 5 \cdot 10^{-3} $ Mb. 
% We set $p_{max}=0.95$ in the constraint \eqref{pmax1}, i.e. the clients  have almost full receiving capacity. 
Real-time video streaming is type-$1$ job with stringent delay requirement of 150 ms and requires full delivery of the required throughput with negligible failure probability \cite{video_latency1}. 
Environmental monitoring is type-$2$ job with mild delay requirement of up to 4 minutes, i.e. thus set to 240 s in our setting, and does not require immediate throughput delivery \cite{minutes_delay1}. 
{Data backup is type-$3$ job with  no deadline requirement, in which the data merely needs to be served eventually, and thus corresponds to the requirement of queue stability as in \eqref{queue_stability1}.}
We hence design the actual length of one orchestration time frame to be 150 ms to meet the more stringent delay requirement (of video streaming). Given the length of each time slot being 0.5 ms, this corresponds to $T_s = 300$.
% Similarly, assume that type-$1$ jobs arrive at client 1 and 2 following Poisson distribution respectively at the rate of $1$ Mbps and $0.5$ Mbps, or equivalently $\lambda^2_1 = 60$ and $\lambda^2_2 = 30$ data units per time frame. 
Assume that type-$1$ jobs arrive  at the constant rate of $2$ Mbps (typical rate for  streaming \cite{video_latency1}), or equivalently $\lambda^1_i = \frac{2\cdot T_s\cdot 10^{-3}\cdot0.5}{5\cdot 10^{-3}}=60$ data units per time frame. 
 Similarly, assume that type-$2$ jobs follow Poisson distribution at the rate of $0.4$ Mbps (typical rate for environmental sensors \cite{50kps}), or equivalently $\lambda^2_i = 12$ data units per time frame. {Type-3 jobs are modeled via a Pareto distribution to capture heavy-tailed processes commonly observed in practice.}
% For the InPs' resource constraint, we use $\beta_1 = 1$ and $\beta_2 = 2$ (as video streaming is just data processing without computation, while environmental monitoring might additionally require data analysis), and set $\pi_1 = 10$ Gbps $=\frac{(10\cdot 10^3)\cdot 10^{-3}\cdot 0.5}{2.5\cdot 10^{-3}} = 2000$ data units per time slot \cite{5g_cellular_total1} and  $\pi_2 = 100$ Gbps $=20000$ data units per time slot \cite{cloud_total1}.
{Among the twenty clients, fifteen of them are sensors performing both environmental monitoring and data backup for storage of the collected data, both of which are at the rate of 0.4 Mbps. The remaining five clients perform video streaming,  three of which additionally perform data backup at the rate of 2 Mbps.}
Now, we proceed to design the QoS models:
% for the  applications with QoS requirements:
\begin{itemize}
    \item \textit{Video streaming:} We first set $q^1_i = 0.99$ to ensure negligible failure probability and, from \eqref{qos_throughput1}, aim to ensure $\gamma^1_i q^1_i \geq \frac{\lambda^1_i}{|\Kk| T_s}$ to guarantee full throughput delivery within every time frame, i.e. also the application's delay requirement. Thus, we choose $\gamma^1_i= 0.0204$ for service provisioning.
    % \item \textit{Environmental monitoring:} The 480 s delay requirement corresponds to a tolerable queueing delay of $W^* = 3200$ time slots. From \eqref{queue_delay_suff1} in Theorem \ref{thrm_delay_bound}, we can compute the sufficient conditions for meeting the application's delay requirement as  $\gamma^2_i q^2_i \geq 3.57 \cdot 10^{-3}$. We thus choose the reliability metric to be $q^2_i = 0.8$, and set $\gamma^2_i = 4.4625 \cdot 10^{-3}$.
    % \item 2nd setting $\lambda^2_i = 18$, then $\gamma^2_i q^2_i \geq 6.5205 \cdot 10^{-3}$. set $q^2_i = 0.6$ and $\gamma^2_i = 0.01087$.
       % \item 3rd setting $\lambda^2_i = 12$, then $\gamma^2_i q^2_i \geq 4.78222 \cdot 10^{-3}$. set $q^2_i = 0.7$ and $\gamma^2_i = 6.83175 \cdot 10^{-3}$.
    \item \textit{Environmental monitoring:} The 240 s delay requirement corresponds to a tolerable queueing delay of $W^* = 1600$ time slots. From \eqref{queue_delay_suff1} in Theorem \ref{thrm_delay_bound}, we can compute the sufficient conditions for meeting the application's delay requirement as  $\gamma^2_i q^2_i \geq 4.78222 \cdot 10^{-3}$. We thus choose the reliability metric to be $q^2_i = 0.7$, and set $\gamma^2_i = 6.83175 \cdot 10^{-3}$.
    \item \textit{Data backup:} This application only requires queue stability \eqref{queue_stability1}. We thus set $\gamma^3_i=0$ and $ q^3_i = 0$.
    % The QoS  of this application is inherently captured in the queue stability requirement \eqref{queue_stability1}. We thus set $\gamma^3_i=0$ and $ q^3_i = 0$.
\end{itemize}
Note that these parameters $T_s, \gma$ and $\qQ$ satisfy  condition \eqref{assum_charac} and Assumption  \ref{assum_slater} in Theorem \ref{thrm_mdp}.
In the next section, we empirically verify the benefits of our QoS model under the described setting.
\begin{figure}
\includegraphics[width=0.33\textwidth]{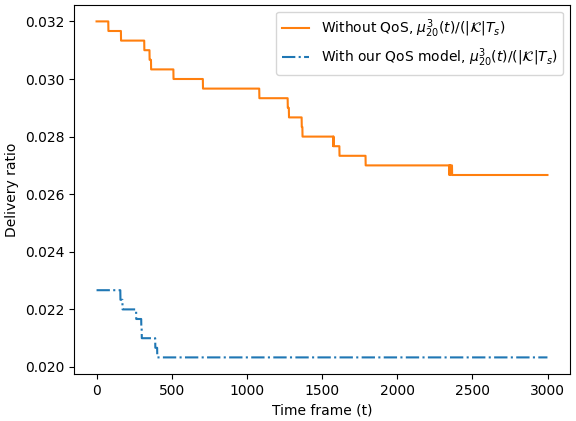}
   % \vspace{-1\baselineskip}
\caption{OTT service of type-$3$ job at client 20. The baseline without QoS over-provisions service for type-$3$ with no strict requirements within a time frame, i.e. $q^3_i = 0$ and $\gamma^3_i=0$.}
% , thereby potentially violating QoS requirements of other types of jobs.}
\label{fig: env_moni_service}
% \aftercaptionskip
\end{figure}

\subsection{ Simulation Results}
\label{sec_exp}
% \subsection{Demonstration of QoS Framework}
% \label{sec_demon}
We perform numerical simulation under the setting in Section \ref{sec_example} of
ten InPs serving twenty clients for three types of jobs with the described  $(T_s, \gma, \qQ)$ configuration. 
We compare our MDP policy solving $\OP$ with QoS framework with the DP policy solving the  problem  \textit{without QoS constraint}, i.e. excluding the constraint \eqref{qos_constraint1}. This baseline represents previous work in the literature of infrastructure sharing that do not incorporate QoS mechanism and typically solve the problem via the  DP policy   \cite{ spectrum_stable1, spectrum_stable2}. We run the simulation  for $T = 3000$ time frames and set $V = 5 \cdot T^{1/3}$ for both policies.
We consider $\alpha$-fair utility \cite{fairness1} where client $i$ receives  utility $f^c_i(x)= w_c \frac{x^{1-\alpha_c}}{1-\alpha_c}$ given the total  service $x$ for type-$c$.  We set $(w_1, w_2, w_3) = (0.106, 0.516, 0.7)$ and $(\alpha_1, \alpha_2, \alpha_3) = (0.5, 0.875, 0.75)$.
Note that $ \EE[\mu^c_i(t)] = \sum_{k\in N(i)}  \sum_{\tau=0}^{T_s-1} r_{ki} p^c_{ki}(t, \tau)$ is the expected  type-$c$ service received by client $i$. The  utility for time frame $t$ is thus $f(\pP(t)) = \sum_{i\in \Ii} \sum_{c\in \Cc} f^c_i\big(  \sum_{k\in N(i)}  \sum_{\tau=0}^{T_s-1} r_{ki} p^c_{ki}(t, \tau)\big)$.

% \vspace{-2.36cm}
First, we empirically verify the benefits of our QoS model in Figure \ref{fig: qos_model_strong} depicting the QoS requirements of all the  applications. We report the QoS measurements for two representative clients, 3 and 20, where similar  patterns are also observed at other clients. In Figure \ref{fig: vid_stream}, the MDP policy consistently  satisfies  the full throughput delivery requirement of type-$1$ job (at least $0.02$ successful delivery ratio), while the baseline without QoS fails to maintain high throughput. In Figure \ref{fig: env_moni}, we plot the actual average queueing delay $W^c_i(t)$ at time $t$.
% , which can be computed in analytical form \cite{mg1}  as a function of $\pP(t)$, 
The average queueing delay of our baseline for type-$2$ job is also much less than that of the baseline without QoS as well as the 4 minutes delay requirement.  Figure \ref{fig: queue} shows that queue stability is achieved by both baselines, thereby ensuring the QoS of all type-$3$ jobs. Nevertheless, the higher queue backlogs of the DP baseline without QoS would result in higher packet delays compared to our algorithm.
Figure \ref{fig: env_moni_service} further demonstrates the inefficiencies of infrastructure sharing without QoS that correspond to the QoS violation of type-$1$ job and overall worse performances for type-$2$ and type-$3$ jobs. Specifically, the DP policy without QoS mechanism over-provisions the OTT service for type-$3$ job. Finally,  our MDP policy  achieves the final average utility of $303.292$, which is close to that of $309.517$ of the DP policy, which ignores the QoS requirements, and thus empirically validates our results on vanishing error of the utility (Theorem \ref{thrm_mdp} and Corollary \ref{corol_V_choice}).

\section{Conclusion}
\label{sec_conclude}

In this paper, we propose a new QoS framework for service provision in multi-infrastructure-sharing networks and the efficient MDP algorithm for the problem. We derive utility-optimality and stability gaurantees for MDP, and  demonstrate the application of our QoS model in realistic scenarios. 
Simulations empirically validate the ability of our algorithm  to handle probabilistic QoS constraints.

% which   facilitates this research direction especially for ad-hoc networks.  
% The control framework is throughput-optimal and  alleviates the communication overhead, which has been a major bottleneck for large-scale SDN.
% Extensive experiments confirm throughput-optimality and favorable scalability of the algorithm. 
% % This work can facilitate further algorithmic development for 

\vspace{-0.05cm}
\section*{Acknowledgment}

This work was supported by NSF grants CNS-2148128 and CNS-2148183.

\bibliographystyle{ACM-Reference-Format}
\bibliography{sample-base}


\newpage

% \appendix

% \input{sections/appendix.tex}

\appendix
\section*{APPENDIX}
% \section{proof}
% \label{yeah}
% Moreover, we consider $\Pi_s$ as the set of policies in $\Pi$ that make \emph{stationary} scheduling decisions. The scheduling of any $\pi \in \Pi_s$ can be further characterized and have additional properties as follows:
% \begin{itemize}
%     \item \textit{Stationary scheduling:} At any time $t$ and for any $i \in [1, m]$, given  $C_{E_i}[t] = \beta_{E_i} \in \{0, 1\}^{|E_i|}$ and  $C_{E}[t-\tau] = \alpha \in \{0, 1\}^{|E|}$, the controller $D_i$ activates the link activation vector $\gamma_{E_i}\in \{0, 1\}^{|E_i|}$   with probability $P(D^\pi_{E_i}(t) = \gamma_{E_i} \big| C_{E_i}[t] = \beta_{E_i}, C_{E}[t-\tau] = \alpha) = p^\pi(\gamma_{E_i}| \beta_{E_i}, \alpha)$
%     \item \textit{Independence across domains:}
% \end{itemize}

% asd

% \section*{Appendix}

\section{Properties of the Linearized Feasible Domain}

\subsection{Proof of Theorem \ref{concen_bound}}
\label{appen_concen_bound}

    Let $S^{c*}_{ti} \subseteq {N(i) \times [T_s]} = \{(k, \tau): k\in N(i), \tau \in [0, T_s-1]\}$ and $(k^*_0, \tau^*_0)$ be the set and index respectively that achieves the maximum in \eqref{big_b1}, and consider $(k_r, \tau_r)= \argminE_{(k,\tau) \in S^{c*}_{ti} \cup (k^*_0, \tau^*_0) }  \hat{p}^c_{ki}(t, \tau)$. We also define the random variable $\eta^c_{ki}(t, \tau) = \frac{ -y^c_{ki}(t, \tau) + r_{ki} p^c_{ki}(t, \tau)}{\hat{p}^c_{ki}(t, \tau)} \in [ -1, 1]$, which is the ``normalized" version of $ 
z^c_{ki}(t, \tau)=-y^c_{ki}(t, \tau) $ to be used later in the proof. We define the weights:
    \begin{align*}
        w^c_{ki}(t, \tau)  = \begin{cases}
            1, \text{ if $(k, \tau)\in S^{c*}_{ti}$ }\\
            \frac{\hat{p}^c_{ki}(t, \tau)}{ \hat{p}^c_{k_r i}(t, \tau_r)},\text{   if $(k, \tau)\in N(i) \times [T_s] \setminus S^{c*}_{ti}$ }
        \end{cases}.
    \end{align*}
% and let $(k_r, \tau_r) = \argminE_{S^{c*}_{ti} \cup (k^*_0, \tau^*_0) }  \hat{p}^c_{ki}(t, \tau)$
Now, we have:
  \begin{align}
  \nonumber
         &P\big( \frac{1}{|\Kk| T_s} \mu^c_i(t) \leq  \gamma^c_{i}  \big) = P\big(\sum_{k\in N(i)}  \sum_{\tau=0}^{T_s-1}z^c_{ki}(t, \tau)  \geq -|\Kk| T_s \gamma^c_i\big) \\
         \nonumber
         &=P\big(\sum_{k\in N(i)}  \sum_{\tau=0}^{T_s-1} \big[ \eta^c_{ki}(t, \tau) \hat{p}^c_{ki}(t, \tau) - r_{ki} p^c_{ki}(t, \tau) \big] \geq -|\Kk| T_s \gamma^c_i\big) \\
         \nonumber
         &\overset{\eqref{linearized_constraint1}}{\leq} P\big(\sum_{k\in N(i)}  \sum_{\tau=0}^{T_s-1}  \eta^c_{ki}(t, \tau) \hat{p}^c_{ki}(t, \tau) \geq  B^c_i(\Gamma^c_i, \pP(t)) \big)\\
         \nonumber
         &=   P\big(\sum_{k\in N(i)}  \sum_{\tau=0}^{T_s-1}  \eta^c_{ki}(t, \tau) \hat{p}^c_{ki}(t, \tau) \geq   \\
         \nonumber
         &\quad \sum_{(k,\tau)\in S^{c*}_{ti}} \hat{p}^c_{ki}(t, \tau)+ (\Gamma^c_i - \lfloor \Gamma^c_i\rfloor) \hat{p}^c_{k^*_o i}(t, \tau^*_0) \big)\\
         \nonumber
          &=   P\big(\sum_{(k,\tau)\in N(i) \times [T_s] \setminus S^{c*}_{ti}} \eta^c_{ki}(t, \tau) \hat{p}^c_{ki}(t, \tau) \geq   \\
          \nonumber
         &\quad \sum_{(k,\tau)\in S^{c*}_{ti}} \hat{p}^c_{ki}(t, \tau) (1-\eta^c_{ki}(t, \tau)) + (\Gamma^c_i - \lfloor \Gamma^c_i\rfloor) \hat{p}^c_{k^*_o i}(t, \tau^*_0) \big)\\
         \nonumber
          &\leq   P\big(\sum_{(k,\tau)\in N(i) \times [T_s] \setminus S^{c*}_{ti}} \eta^c_{ki}(t, \tau) \hat{p}^c_{ki}(t, \tau) \geq   \\
         &\quad \quad \hat{p}^c_{k_r i}(t, \tau_r) \cdot \big[ \sum_{(k,\tau)\in S^{c*}_{ti}}  (1-\eta^c_{ki}(t, \tau)) + (\Gamma^c_i - \lfloor \Gamma^c_i\rfloor) \big] \big) \label{lma_concen_bound_proof1}\\
         \nonumber
         &=   P\big(  \sum_{(k,\tau)\in N(i) \times [T_s] \setminus S^{c*}_{ti}} \eta^c_{ki}(t, \tau) \frac{\hat{p}^c_{ki}(t, \tau) }{ \hat{p}^c_{k_r i}(t, \tau_r)} \\
         \nonumber
         &\quad \quad +  \sum_{(k,\tau)\in S^{c*}_{ti}}\eta^c_{ki}(t, \tau)  \geq    \Gamma^c_i \big)\\
         &=  P\big(  \sum_{(k,\tau)\in N(i) \times [T_s] } \eta^c_{ki}(t, \tau) w^c_{ki}(t, \tau)  \geq  \Gamma^c_i \big), \label{lma_concen_bound_proof2}
    \end{align}
    where \eqref{lma_concen_bound_proof1} follows from $ \hat{p}^c_{ki}(t, \tau) \geq  \hat{p}^c_{k_r i}(t, \tau_r)$ by definition and that $1 -  \eta^c_{ki}(t, \tau)\geq 0$. 
    
    Next, we proceed to show that  $ w^c_{ki}(t, \tau)  \in [0, 1] $ for all $ (k, \tau) \in N(i) \times [T_s]$. Suppose for the sake of contradiction that there exists some $ (k', \tau') \in N(i) \times [T_s] \setminus S^{c*}_{ti}$ such that $ w^c_{ki}(t, \tau)  > 1$ .  If $ (k', \tau') \neq (k^*_0, \tau^*_0)$, then  $ (k', \tau') \in N(i) \times [T_s] \setminus (S^{c*}_{ti}  \cup (k^*_0, \tau^*_0))$. Since $(k_r, \tau_r) \in S^{c*}_{ti} \cup (k^*_0, \tau^*_0) $, we can  increase the objective value in the maximization problem \eqref{big_b1} by exchanging $(k', \tau')$ with $(k_r, \tau_r)$ from the set $S^{c*}_{ti} \cup (k^*_0, \tau^*_0)$. Likewise,  if $ (k', \tau') = (k^*_0, \tau^*_0)$, then  we first have $ (k_r, \tau_r) \neq (k^*_0, \tau^*_0)$ (otherwise, $w^c_{ki}(t,\tau)$ = 1)  and hence $ (k_r, \tau_r) \in S^{c*}_{ti} $. Then, we  can  increase the objective value in the maximization problem \eqref{big_b1} by exchanging $ (k^*_0, \tau^*_0)$ with$(k_r, \tau_r)$. In both cases, we obtain a new solution other than $S^{c*}_{ti} \cup (k^*_0, \tau^*_0)$ that increases the objective value of \eqref{big_b1}, which is a contradiction.

Given that $ w^c_{ki}(t, \tau)  \in [0, 1] $ for all $ (k, \tau) \in N(i) \times [T_s]$ and $\eta^c_{ki}(t, \tau) $'s are independent and symmetrically distributed in $[-1, 1]$, we use \cite[Theorem 2]{Bertsimas2004ThePO} to obtain that:
\begin{align}
     P\big(  \sum_{(k,\tau)\in N(i) \times [T_s] } \eta^c_{ki}(t, \tau) w^c_{ki}(t, \tau)  \geq  \Gamma^c_i \big) \leq e^{-\frac{(\Gamma^c_i)^2}{2|\Kk| T_s} }. \label{lma_concen_bound_proof3}
\end{align}
Finally, combining \eqref{lma_concen_bound_proof2} and \eqref{lma_concen_bound_proof3}, we conclude the proof of the Theorem.
    % Adapting the  technique as in the proof of Proposition 2 and Theorem 2 in \cite{Bertsimas2004ThePO}, we can prove the following two inequalities:
    % % P\big( \frac{1}{|\Kk| T_s} \mu^c_i(t) \leq  \gamma^c_{i}  \big) = 
    % \begin{align*}
    %      &P\big( \frac{1}{|\Kk| T_s} \mu^c_i(t) \leq  \gamma^c_{i}  \big) = P\big(\sum_{k\in N(i)}  \sum_{\tau=0}^{T_s-1}z^c_{ki}(t, \tau)  \geq -|\Kk| T_s \gamma^c_i\big) \\
    %      &\leq P\big(\sum_{k\in N(i)}  \sum_{\tau=0}^{T_s-1}   w^c_{ki}(t, \tau) \eta^c_{ki}(t, \tau)  \geq \Gamma^c_i\big)\leq e^{-\frac{(\Gamma^c_i)^2}{2|\Kk| T_s} },
    % \end{align*}
    % which   concludes the proof of the required statement \eqref{chernoff}.

\subsection{Proof of Theorem \ref{thm_dual}}
\label{appen_thm_dual}

First, the following lemma characterizes $B^c_i(\Gamma^c_i, \pP(t)) $ into linear optimization problem. 

\begin{lemma}
 We have the following: 
 \begin{align}
  B^c_i(\Gamma^c_i, \pP(t))  = \max &\quad   \sum_{k\in N(i)} \label{appen_thm_dual_proof1}
  \sum_{\tau=0}^{T_s-1} \hat{p}^c_{ki}(t, \tau) \cdot z_{k \tau}\\
  \nonumber
     \text{such that }\quad  &  \sum_{k\in N(i)} \sum_{\tau=0}^{T_s-1} z_{k \tau}  \leq \Gamma^c_i\\
     \nonumber
     & 0\leq z_{k \tau} \leq 1, \quad \forall k\in N(i), \tau \in[T_s]
 \end{align}
\end{lemma}
\begin{proof}
Apparently, the optimal solution to the above problem consists of $\lfloor \Gamma^c_i\rfloor$ variables at $1$ and one variable at $\Gamma^c_i - \lfloor \Gamma^c_i\rfloor$. This is equivalent to the subset selection problem as in the original problem, where we find the subset 
$$\{S^c_{ti}\cup\{(k_0, \tau_0)\}: S^c_{ti} \in 2^{N(i) \times [0, T_s-1]}, |S^c_{ti}| =\lfloor \Gamma^c_i\rfloor    \}$$
that maximizes the corresponding objective $\sum_{(k,\tau)\in S^c_{ti}} \hat{p}^c_{ki}(t, \tau)+ (\Gamma^c_i - \lfloor \Gamma^c_i\rfloor) \hat{p}^c_{k_0 i}(t, \tau_0)$. 
\end{proof}

The above problem \eqref{appen_thm_dual_proof1} is feasible and bounded for finite $\Gamma^c_i > 0$.
Since strong duality holds for any feasible linear optimization problem, $B^c_i(\Gamma^c_i, \pP(t)) $ admits the following dual form of \eqref{appen_thm_dual_proof1}:
\begin{align}
    B^c_i(\Gamma^c_i, \pP(t))  = \min &\quad   \sum_{k\in \Kk} \sum_{\tau=0}^{T_s-1} v^c_{i, k \tau} + \Gamma^c_i s^c_i \label{appen_thm_dual_proof2}\\
    \nonumber
     \text{such that }\quad  &  s^c_i + v^c_{ ki \tau} \geq r_{ki} p^c_{ki}(t, \tau),\\
     \nonumber
     &\quad \quad \forall  k\in N(i), \tau \in[T_s] \\
     \nonumber
     &  s^c_i + v^c_{ ki \tau} \geq 1- r_{ki} p^c_{ki}(t, \tau), \\
     \nonumber
     &\quad \quad \forall  k\in N(i), \tau \in[T_s] \\
     \nonumber
     &  v^c_{ki \tau} \geq 0, \quad \forall  k\in N(i), \tau \in[T_s] \\
     \nonumber
     & s^c_i \geq 0.
\end{align}
Substituting \eqref{appen_thm_dual_proof2} into the linearized constraint \eqref{linearized_constraint1} of the definition \eqref{def_approx_region} of the approximate feasible region $  \widetilde{\Dd}(\gma, \qQ)$, we conclude the equivalent form for $  \widetilde{\Dd}(\gma, \qQ)$ as presented in the Theorem's statement.

\subsection{Proof of Theorem \ref{thrm_approx_bound}}
\label{appen_approx_bound}
 As under \eqref{gamma_choice}, the linearized constraint \eqref{linearized_constraint1} implies \eqref{qos_constraint1}, we obtain that $ \widetilde{\Dd}(\gma, \qQ)  \subseteq  {\Dd}(\gma, \qQ)$. To prove the remaining part of \eqref{characterize_region}, we consider any $\pP(t) \in {\Dd}(\gma, \qQ)$ and show that $\pP(t) \in  \widetilde{\Dd}(\gma - \frac{K_1}{\sqrt{T_s}} \one, \qQ)$  under  \eqref{assum_charac}. We first have the following Lemma.

 \begin{lemma}
 \label{lma_exp_var}
  For $\pP(t) \in {\Dd}(\gma, \qQ)$, we have $\forall i\in \Ii, c\in \Cc$:
   \begin{align}
        \EE[\mu^c_i(t)] &= \sum_{k\in N(i)}  \sum_{\tau=0}^{T_s-1} r_{ki}  p^c_{ki}(t, \tau), \label{exp_service1}\\
        \mathbf{Var}[\mu^c_i(t)] &= \sum_{k\in N(i)}  \sum_{\tau=0}^{T_s-1} r_{ki}  p^c_{ki}(t, \tau)(1- r_{ki}  p^c_{ki}(t, \tau)), \label{var_service1}\\
      (1- r_{max})   |\Kk| T_s& \gamma^{c}_{i}q^c_i   \leq  \mathbf{Var}[\mu^c_i(t)]\leq  \EE[\mu^c_i(t)]  \leq T_s U_i.\label{var_bound1}
    \end{align}
 \end{lemma}
 \begin{proof}
     The  probabilities $\pP(t)$ would induce the  Poisson Binomial random variable $ \mu^c_i(t)= \sum_{k\in N(i)}  \sum_{\tau=0}^{T_s-1} y^c_{ki}(t, \tau)$ with $y^c_{ki}(t, \tau) \sim Bern( r_{ki}  p^c_{ki}(t, \tau))$. 
 % We can compute the expectation and variance of $\mu^c_i(t)$ as follows.
 % The next Lemma, whose proof is omitted for brevity,  provides 
 Thus, we obtain the expression for the  expectation and variance of $\mu^c_i(t)$ as in \eqref{exp_service1} and \eqref{var_service1}.  Since $r_{ki} \in (0, r_{max}] $ and $p^c_{ki}(t, \tau) \in [0,1]$, we obtain from  \eqref{var_service1} that:
 \begin{align}
 \nonumber
      \mathbf{Var}[\mu^c_i(t)] &\geq  (1-r_{max})  \sum_{\tau=0}^{T_s-1} r_{ki}  p^c_{ki}(t, \tau) = (1-r_{max} )  \EE[\mu^c_i(t)] \\
      &\overset{\eqref{qos_throughput1}}{\geq}  (1-r_{max} ) |\Kk| T_s \gamma^c_i q^c_i. \label{var_bound0}
 \end{align}
 % By \eqref{var_service1} and  the fact that $\pP(t)$ must satisfy the inequalities of the feasible domain $\Dd(\gma, \qQ)$, we can prove that:
 By \eqref{pmax1}, \eqref{exp_service1}, \eqref{var_service1} and \eqref{var_bound0}, we have:
 \begin{equation*}
      (1- r_{max})   |\Kk| T_s \gamma^{c}_{i}q^c_i   \leq  \mathbf{Var}[\mu^c_i(t)]\leq  \EE[\mu^c_i(t)]  \leq T_s U_i,
 \end{equation*}
 which is the required statement \eqref{var_bound1}.
 \end{proof}

 % \begin{lemma}
 % We have the following bounds on $ \mathbf{Var}[\mu^c_i(t)]$:
 % \begin{align}
 %      (1- p_{max})   |\Kk| T_s \gamma^{c}_{i}q^c_i   \leq  \mathbf{Var}[\mu^c_i(t)] \leq T_s p_{max}
 % \end{align}
 % \end{lemma}
 Back to the main proof, using \eqref{qos_constraint1} and invoking \cite[Theorem 3.3]{tang2019poisson} on the Gaussian approximation error of Poisson Binomial distribution, we have $\forall (i, c)\in \Ss$:
 \begin{align}
 \nonumber
 1-q^c_i &\geq  P\big(  \mu^c_i(t) \leq  |\Kk| T_s\gamma^c_{i}  \big) \\
 \nonumber
 &\geq \Phi\bigg(\frac{|\Kk|  T_s \gamma^{c}_{i} -\EE[\mu^{c}_i(t)]   }{\sqrt{ \mathbf{Var}[\mu^c_i(t)] }} \bigg) - \frac{0.7915}{\sqrt{ \mathbf{Var}[\mu^c_i(t)] }}   \\
 &\overset{\eqref{var_bound1}}{\geq} \Phi\bigg(\frac{|\Kk|  T_s \gamma^{c}_{i} -\EE[\mu^{c}_i(t)]   }{\sqrt{ \mathbf{Var}[\mu^c_i(t)] }} \bigg) - \frac{q^c_i \sqrt{K^c_i}}{\sqrt{T_s  }}. \label{charac_proof1}
 \end{align}
 Under the condition \eqref{assum_charac}, we have that $p^c_i = 1-q^c_i + \frac{q^c_i \sqrt{K^c_i}}{\sqrt{T_s  }} \in (0, 1)$. Thus, we further obtain from \eqref{charac_proof1} that:
 \begin{align}
\label{charac_proof2}
    \EE[\mu^c_i(t)]&\geq |\Kk|T_s \gamma^{c}_{i} - \sqrt{ \mathbf{Var}[\mu^c_i(t)] } \Phi^{-1}(p^c_i).
    % \\
   % \sum_{k\in N(i)}  \sum_{\tau=0}^{T_s-1} p^c_{ki}(t, \tau) &\geq |\Kk|T_s \gamma^{c}_{i} - \sqrt{ T_s p_{max}}\mathit{erf}^{-1}(2p^c_i - 1) 
\end{align}
We  can express  $\Phi^{-1}(.)$ in terms of the inverse error function $\mathit{erf}^{-1}(.)$, whose upper-bound is given by \cite{erf_bound}, as follows:
\begin{align*}
    \Phi^{-1}(p^c_i) =\sqrt{2} \mathit{erf}^{-1}(2p^c_i-1)&\leq 2\sqrt{\tfrac{\pi}{e}}  e^{2 (2p^c_i-1)^2}<   2 e^{1.5}\pi^{0.5}. %\label{erf_bound1}
\end{align*}
Plugging the above  into \eqref{charac_proof2} in view of \eqref{exp_service1} and \eqref{var_bound1}:
\begin{align}
\label{charac_proof3}
    \sum_{k\in N(i)}  \sum_{\tau=0}^{T_s-1} r_{ki}  p^c_{ki}(t, \tau) \geq |\Kk|T_s \gamma^{c}_{i} - 2\sqrt{T_s U_{max} \pi e^3} 
\end{align}
Consider the vector $\gma' = \gma - \frac{K_1}{\sqrt{T_s}} \one $, i.e. $ {\gamma^c_i}'= \gamma^c_i - \frac{K_1}{\sqrt{T_s}}, \forall i\in \Ii, c\in \Cc$. Since $\|\qQ\|_\infty \geq q^c_i$, we obtain the following bound:

\begin{align}
\label{charac_proof4}
\gamma^c_i  \geq {\gamma^c_i}' +  \frac{ \sqrt{2  \log\big( (1-q^c_i )^{-1} \big)}}{\sqrt{|\Kk| T_s }} + \frac{2\sqrt{U_{max}  \pi e^3}}{|\Kk|\sqrt{T_s} }
\end{align}
Applying \eqref{charac_proof4} to \eqref{charac_proof3} and using \eqref{gamma_choice}, we have:
% {\small
\begin{align}
\nonumber
    \sum_{k\in N(i)}  \sum_{\tau=0}^{T_s-1} r_{ki}  p^c_{ki}(t, \tau) &\geq |\Kk|T_s  {\gamma^c_i}' + \Gamma^c_i\\
    &\geq |\Kk|T_s  {\gamma^c_i}' + B^c_i(\Gamma^c_i, \pP(t)), \label{charac_proof5}
\end{align}
% }%
where the last line follows directly from  \eqref{big_b1}. Since \eqref{charac_proof5} is just the  constraint \eqref{linearized_constraint1} with $ {\gamma^c_i}'$, we obtain that $\pP(t) \in\widetilde{\Dd}(\gma', \qQ) = \widetilde{\Dd}(\gma - \frac{K_1}{\sqrt{T_s}} \one, \qQ) $, which concludes the proof the Theorem.

\section{Proof of Lemma \ref{lma_driftbound} }
\label{appen_lma_driftbound}

From the queue dynamics \eqref{queue_dynamics}, we first obtain that $\forall i\in \Ii, c\in \Cc, t\geq0$:
\begin{align*}
 Q^c_i(t+1)^2  &\leq \big( Q^c_i(t) + a^c_i(t) -  \mu^c_i(t)\big)^2\\   
 &=  Q^c_i(t)^2 + (a^c_i(t) -  \mu^c_i(t))^2  + 2 Q^c_i(t) (a^c_i(t) -  \mu^c_i(t))\\
 &\leq  Q^c_i(t)^2 + 2 A^2_{max} +2 \mu^c_i(t)^2  + 2 Q^c_i(t) (a^c_i(t) -  \mu^c_i(t)),
\end{align*}
where in the last inequality, we use $0\leq a^c_i(t)\leq A_{max}$. Summing the above over $i\in \Ii, c\in \Cc$, we obtain that:

\begin{align}
\nonumber
    &\frac{1}{2} \sum_{c\in \Cc}\sum_{i\in I} Q^c_i(t+1)^2 -\frac{1}{2}  \sum_{c\in \Cc}\sum_{i\in I} Q^c_i(t)^2  \\
    \nonumber
    &\leq |\Ii| |\Cc| A^2_{max}  +  \sum_{c\in \Cc}\sum_{i\in I} \mu^c_i(t)^2 - \sum_{c\in \Cc}\sum_{i\in I}Q^c_i(t) (a^c_i(t) -  \mu^c_i(t))\\
    \nonumber
    &\leq |\Ii| |\Cc| A^2_{max}  +  [\sum_{c\in \Cc}\sum_{i\in I} \mu^c_i(t)]^2  - \sum_{c\in \Cc}\sum_{i\in I}Q^c_i(t) (a^c_i(t) -  \mu^c_i(t))\\
    \nonumber
    &= |\Ii| |\Cc| A^2_{max}  + |\Kk|^2 T_s^2- \sum_{c\in \Cc}\sum_{i\in I}Q^c_i(t) (a^c_i(t) -  \mu^c_i(t)) \\
    &= B_1 - \sum_{c\in \Cc}\sum_{i\in I}Q^c_i(t) (a^c_i(t) -   \sum_{\tau = 0}^{T_s -1 }\sum_{k\in N(i)} y^c_{ki}(t, \tau)), \label{lma_driftbound_proof1}
\end{align}
where in the second inequality, we use: 
$$\sum_{c\in \Cc}\sum_{i\in I} \mu^c_i(t) =  \sum_{k\in \Kk} \sum_{\tau=0}^{T_s-1} \sum_{c\in \Cc}\sum_{i\in N(k)} y^c_{ki}(t, \tau) =  \sum_{k\in \Kk} \sum_{\tau=0}^{T_s-1} 1= |\Kk|T_s .$$ 
Finally, taking  expectation of \eqref{lma_driftbound_proof1} conditioned on $\QQ(t)$ and noting that $\EE[a^c_i(t) \big| \QQ(t) ] = \EE[a^c_i(t)] = \lambda^c_i$ and $\EE[y^c_{ki}(t) \big| \QQ(t) ] =r_{ki} p^c_{ki}(t)$, we conclude 
 the required statement of the Lemma.

\section{Proof of Lemma \ref{lma_perturbed_obj}}
\label{appen_lma_perturbed_obj}

% Note that we can equivalently write the optimization step \eqref{drift_plus_penalty2} of the MDP policy as solving a minimization problem via:
% $$ \max_{ \pP(t) \in \widetilde{\Dd}(\gma, \qQ) } h_t(\pP(t)) = -  \min_{ \pP(t) \in \widetilde{\Dd}(\gma, \qQ) } -h_t(\pP(t)) $$
% with $ -h_t(\pP(t)) $ being a convex objective. Thus, we have $\pP^{\MDP}(t) = \argminE_{ \pP(t) \in \widetilde{\Dd}(\gma, \qQ) } -h_t(\pP(t)) $.
% Similarly, we can write \eqref{lma_obj_gap_proof0} equivalently as:
%  $\pP'(t)= \argminE_{ \pP(t) \in \widetilde{\Dd}(\gma', \qQ) } -h_t(\pP(t)). $

% thm_perturbation_bound_no_v_in_statement

First, we consider the  decompositions $ \widetilde{\Dd}(\gma, \qQ) = S \cap M(\gma)$ and $ \widetilde{\Dd}(\gma', \qQ) = S \cap M(\gma')$, where the convex set $$  S =\big\{ \pP(t) \in [0, 1]^{T_s |E| |\Cc|}\text{: such that \eqref{prob_constraint}, \eqref{pmax1}}\big\} $$ is comprised of all the mutual linear constraints in $\widetilde{\Dd}(\gma, \qQ)$ and $\widetilde{\Dd}(\gma', \qQ)$ and the convex set $M(\gma)$ is formally defined by:
 \begin{align}
 \nonumber
    M(\gma)  = \{&  \pP(t) \in [0, 1]^{T_s |E| |\Cc|}:  \forall (i, c)\in \Ss: \\
    &|\Kk|T_s \gamma^c_i \leq \sum_{k\in N(i)} \sum_{\tau=0}^{T_s-1} r_{ki}  p^c_{ki}(t, \tau) - B^c_i(\Gamma^c_i, \pP(t)).\}. \label{m_gamma_def1}
 \end{align}
 Next, we consider the following indicator function:
 \begin{align}
 \label{indicator1}
     \mathbbm{1}_S(\pP(t)) = \begin{cases} 0 \quad \text{, if $\pP(t)\in S$}\\ +\infty,  \text{otherwise}\end{cases},
 \end{align}
 % which is convex in $\pP$ by the convexity of $S$. By defining the following new concave objective:
 % \begin{align}
 %     \bar{h}_t(\pP) = h_t(\pP) - \mathbbm{1}_S(\pP),
 % \end{align}
 % We have:
 % \begin{align}
 %     h_t(\pP') &= \max_{\pP(t) \in \widetilde{\Dd}(\gma, \qQ)} h_t(\pP(t)) \\
 %     &= \max_{\pP(t) \in M(\gma)} \{ \bar{h}_t(\pP(t)) = h_t(\pP(t)) - \mathbbm{1}_S(\pP(t)) \} \label{proove_distort1}\\
 %     h_t(\pP'') &= \max_{\pP(t) \in \widetilde{\Dd}(\gma'(t), \qQ(t))} h_t(\pP(t)) = \max_{\pP(t) \in M(\gma'(t))} \{ \bar{h}(\pP(t)) = h_t(\pP(t)) - \mathbbm{1}_S(\pP(t)) \}  ,\label{proove_distort2}
 % \end{align}
 which is convex in $\pP$ by the convexity of $S$. By defining the following new concave objective:
\begin{align}
    \bar{h}_t(\pP(t)) = h_t(\pP(t)) - \mathbbm{1}_S(\pP(t)), \label{def_h_bar_t}
\end{align}
we have:
\begin{align}
\nonumber
    h_t(\pP^\MDP(t)) &= \max_{\pP(t) \in \widetilde{\Dd}(\gma, \qQ)} h_t(\pP(t)) \\
    &= \max_{\pP(t) \in M(\gma)} \{ \bar{h}_t(\pP(t)) = h_t(\pP(t)) - \mathbbm{1}_S(\pP(t)) \}, \label{proove_distort1} \\
    \nonumber
    h_t(\pP'(t)) &= \max_{\pP(t) \in \widetilde{\Dd}(\gma', \qQ)} h_t(\pP(t)) \\
    &= \max_{\pP(t) \in M(\gma')} \{ \bar{h}_t(\pP(t)) = h_t(\pP(t)) - \mathbbm{1}_S(\pP(t)) \}. \label{proove_distort2}
\end{align}
% where $M(\gma)$ and $M(\gma')$ contain only linear inequality constraints. 
The following Lemma, whose proof is given in Appendix \ref{appen_lma_bound_h}, provides bounds for the objective $\bar{h}_t(\pP(t))$.
\begin{lemma}
\label{lma_bound_h}
  For any $\pP(t)$, we have:
  \begin{align}
      \bar{h}_t(\pP(t))   \leq \| \QQ(t) \|_\infty  T_s  |\Kk| + V f_{max}. \label{bound_h1}
  \end{align}
\end{lemma}

Next, we penalize such linear inequality constraints by the following penalty function for some hyper-parameter $r>0$ to be determined later:
\begin{align}
\nonumber
    &P_{r, \gma}(\pP(t)) = r \sum_{(i, c)\in \Ss} \min \bigg\{ 0, \\
    &\frac{1}{|\Kk|T_s}\sum_{k \in N(i)} \sum_{\tau=0}^{T_s-1} r_{ki} p^c_{ki}(t, \tau) - \frac{1}{|\Kk|T_s}B^c_i(\Gamma^c_i, \pP(t)) -  \gamma^c_i \bigg\}. \label{penalty}
\end{align}
The following Lemma, whose proof is deferred to Appendix \ref{appen_penalty_diff}, provides the bound on the variation of $P_{r, \gma}(\pP(t))$ for perturbing $\gma$.
\begin{lemma}
\label{penalty_diff}
For any fixed $r>0$ and $\pP(t)$, we have:
\begin{align}
    | P_{r, \gma}(\pP(t)) -  P_{r, \gma'}(\pP(t))| \leq r  \|\gma  -\gma'  \|_1.
\end{align}
\end{lemma}
Next, leveraging results from \cite{penalty1}, we convert \eqref{proove_distort1} and \eqref{proove_distort2} into unconstrained optimization problems by incorporating the penalty \eqref{penalty} into their objective for large enough $r$.
\begin{lemma}
\label{lma_penalty_main}
Under Assumption \ref{assum_slater} and for  $r \geq K_3 = \frac{\| \QQ(t) \|_\infty |\Kk| T_s  }{\zeta} + \frac{V (f_{max} - f_{min})}{ \zeta}$, we have:
\begin{align}
    h_t(\pP^\MDP(t)) = \max_{\pP(t)} \{ \bar{h}_t(\pP(t)) + P_{r, \gma}(\pP(t)) \}, \label{lma_penalty_main_eq1} \\
    h_t(\pP'(t)) = \max_{\pP(t)} \{ \bar{h}_t(\pP(t)) + P_{r, \gma'}(\pP(t)) \}.\label{lma_penalty_main_eq2}
\end{align}
\end{lemma}
The proof of Lemma \ref{lma_penalty_main} is given in Appendix \ref{appen_lma_penalty_main}. Now, by setting $r= K_3$ and from \eqref{lma_penalty_main_eq2} and Lemma \ref{penalty_diff}, we have:
\begin{align*}
    h_t(\pP'(t))  &= \bar{h}_t(\pP'(t)) + P_{r, \gma'}(\pP'(t)) \\
    &\leq  \bar{h}_t(\pP'(t)) + P_{r, \gma}(\pP'(t)) + K_3 \| \gma - \gma_1\|_1 \\
    &\leq  \max_{\pP(t)} \{ \bar{h}_t(\pP(t)) + P_{r, \gma}(\pP(t)) \}  + K_3 \| \gma - \gma_1\|_1\\
    &\overset{\eqref{lma_penalty_main_eq1}}{=}   h_t(\pP^\MDP(t))+ K_3 \| \gma - \gma_1\|_1,
\end{align*}
which concludes the proof of the Lemma.

% From \eqref{proove_distort1}, \eqref{proove_distort2} and the fact that $M(\gma) \subseteq M(\gma')$, we obtain that: 
% \begin{align}
%      h_t(\pP^\MDP(t))  \leq  h_t(\pP'(t)) \iff 0\leq  h_t(\pP'(t)) - h_t(\pP^\MDP(t)) ,
% \end{align}
% which is the first part of the required statement 

% Finally, combining Lemma \ref{penalty_diff} and Lemma \ref{lma_penalty_main}, we conclude the proof of this main Lemma \ref{lma_hard_bound}.

% Plugging the inequality in Lemma \ref{lma_hard_bound} into \eqref{final_rela3}, we conclude the proof of this Theorem \ref{sub_optimality}.

% \section{Proof of Theorem \ref{thm_charac_stability_region}}

\section{Stability and Performance Guarantees}

This Appendix provides proofs of Theorem \ref{thm_charac_stability_region} and \ref{thrm_mdp}. 
In addition to the original control problem $\OP$, we consider the  control problem $(\AOP_{\lmb, \gma, \qQ})$ where the controller makes the decision $\pP(t) \in \widetilde{\Dd}(\gma, \qQ)$ within the linearized feasible domain to support the arrival rate $\lmb$. Formally, we write:
 \begin{align*}
      (\AOP_{\lmb, \gma, \qQ}): \quad   \max_{\pP=\{\pP(t)\}_{t\geq0}}   &F(\pP) \triangleq \liminf_{T\to \infty}\frac{1}{T} \sum_{t=0}^{T-1}f(\pP(t))\\
         \text{such that}\quad &\pP(t) \in \widetilde{\Dd}(\gma, \qQ) \\
         &  \EE[\QQ(T)] = o(T).
\end{align*}
Here, we explicitly specify the parameters $\lmb, \gma$ and $ \qQ$ of the problem. Similarly, within the scope of this Appendix, we use the more tailored notation  $(\op_{\lmb, \gma, \qQ})$ (instead of just $\OP$) for the original control problem. 
Recall that for $q^c_i = 0$, i.e. the case of  no QoS requirement, the constraint \eqref{qos_constraint1} trivially holds for any value of $\gamma^c_i$. 
% Without loss of generality, we assume that  whenever $q^c_i = 0$,  the delivery ratio  is set to $\gamma^c_i =  \frac{\lambda^c_i}{|\Kk|T_s}, \forall (i, c)\in \Ss$\footnote{This delivery ratio $\frac{\lambda^c_i}{|\Kk|T_s}$ corresponds to the service throughput of $\lambda^c_i$, i.e. the arrival rate of type-$c$ jobs at client $i$. }. Note that since $\lambda^c_i$ is the total arrival rate over $T_s$ time slots, we assume that $\lambda^c_i = \Theta(T_s)$. 
% We remark that all the following sets $\Dd(\gma, \qQ), \widetilde{\Dd}(\gma, \qQ), \Lmb(\gma, \qQ)$ and $\widetilde{\Lmb}(\gma, \qQ)$ would remain the same if one varies  $\gamma^c_i \in (0,1)$ for any $(i,c) \in \Ss$.
We next present some useful properties for studying stability and performance guarantees in the following Appendix.

% Recall that for $q^c_i = 0$, i.e. the case of  no QoS requirement, the constraint \eqref{qos_constraint1} trivially holds for any value of $\gamma^c_i$ and $\gamma_{min} = \min_{ (i, c) \in \Ss}\{ \gamma^c_i\}$. Without loss of generality, we assume that  whenever $q^c_i = 0$,  the delivery ratio  is set to $\gamma^c_i = \min\big\{ \gamma_{min}, \frac{\lambda^c_i}{|\Kk|T_s} \big\}, \forall (i, c)\in \Ss$. 
% We remark that all the following sets $\Dd(\gma, \qQ), \widetilde{\Dd}(\gma, \qQ), \Lmb(\gma, \qQ)$ and $\widetilde{\Lmb}(\gma, \qQ)$ would remain the same if one varies  $\gamma^c_i \in (0,1)$ for any $(i,c) \in \Ss$.
% \footnote{This delivery ratio $\frac{\lambda^c_i}{|\Kk|T_s}$ corresponds to the service throughput of $\lambda^c_i$, i.e. the arrival rate of type-$c$ jobs at client $i$. }

\subsection{Useful Properties of the Approximate Stability Region and the Utility Optimization}
\label{appen_helpful_stability_utility}

First, we provide for the following Skorokhod mapping representation of the queue dynamics $\{\QQ(\tau_j)\}_{j\geq 0}$.

\begin{lemma}[\cite{UMW}]
\label{skorokhod}
Define $a^c_i(t_1, t_2) = \sum_{q= t_1}^{t_2-1} a^c_i(q)$ and $\mu^c_i(t_1, t_2) = \sum_{q= t_1}^{t_2-1} \mu^c_i(q)$ . The discrete-time  Skorokhod map representation of the  queue process $\{\QQ(t)\}_{t\geq 0}$ can be expressed as:
\begin{align*}
    Q^c_i(t) = \Big( \sup_{1\leq q \leq t} \big( a^c_i(t-q,t)- \mu^c_i(t-q,t) \big) \Big)^+, \forall c\in \Cc, i\in \Ii.
\end{align*}
\end{lemma}

The next Lemma shows the existence of an optimal policy for the problem $(\AOP_{\lmb, \gma, \qQ})$ that is staionary randomized.

\begin{lemma}
\label{lma_linear_control}
For any $\lmb \in   \widetilde{\Lmb}(\gma, \qQ)$, i.e. there exists $\pi$ such that $ \pP^\pi(t) \in  \widetilde{\Dd}(\gma, \qQ), \forall t\geq 0 \text{ and } \EE[\QQ^\pi(T)] = o(T)$, then the followings hold:
\begin{enumerate}[(i)]
    \item There is some  $\pP = \{p^c_{ki}(\tau) \}_{\tau\in[T_s], (k,i) \in E, c\in \Cc} \in  \widetilde{\Dd}(\gma, \qQ)$ such that $\lambda^c_i \leq \sum_{k\in N(i)} \sum_{\tau=0}^{T_s-1} r_{ki} 
 p^c_{ki}(\tau) , \forall i\in \Ii, c\in \Cc$.
    \item Furthermore, for any $\pP $ satisfying the conditions in (i), we consider a stationary randomized policy $\pi^*$ such that $p^{\pi^* c}_{ki}(t, \tau) = p^c_{ki}(\tau), \forall t\geq 0$. Then $\pi^*$ is a feasible policy that yields the utility that is not worse than that of $\pi$. Formally, the stationary randomized policy $\pi^*$ satisfies  $ \pP^{\pi^*}(t) \in  \widetilde{\Dd}(\gma, \qQ), \forall t\geq 0,\EE[\QQ^{\pi^*}(T)] = o(T)$ and $F(\pP^{\pi^*}) \geq F(\pP^\pi)$.
\end{enumerate}
\end{lemma}
\begin{proof}
Theorem \ref{thm_dual} provides the characterization of $\widetilde{\Dd}(\gma, \qQ)$ in linear form, which can be equivalently expressed as:
\begin{align}
    \widetilde{\Dd}(\gma, \qQ) = \{\pP(t) \in [0, 1]^{T_s |E| |\Cc|}: \text{ such that } A \begin{bmatrix}\pP(t) \\ \sS \\ \rR \end{bmatrix} \leq \bB  \}, \label{linear_charac}
\end{align}
where $\sS = \{s^c_i\}_{(i, c)\in \Ss}$ and $\rR = \{r^c_{ ki\tau}\}_{(i, c)\in \Ss, k\in N(i), \tau \in [T_s]}$ are variable vectors, $A$ is the coefficient matrix and $\bB$ is the boundary vector. 
For any $\lmb \in   \widetilde{\Lmb}(\gma, \qQ)$, i.e. there exists $\pi$ such that $ \pP^\pi(t) \in  \widetilde{\Dd}(\gma, \qQ), \forall t\geq 0 \text{ and } \EE[\QQ^\pi(T)] = o(T)$. From \eqref{linear_charac}, this means that there exists $\sS(t)$ and $\rR(t)$ such that $\forall t\geq 0$:
\begin{align}
    A \begin{bmatrix}\pP^\pi(t) \\ \sS(t) \\ \rR(t) \end{bmatrix} \leq \bB.  \label{linear_control1}
\end{align}
% $\pP = \{p^c_{ki}(\tau) \}_{\tau\in[0, T_s-1], (k,i) \in E, c\in \Cc} \in  \widetilde{\Dd}(\gma, \qQ)$ 

Now, we consider $\bar{\pP}(t) = \{\bar{p}^{c}_{ki}(t, \tau) \}_{\tau\in[T_s], (k,i) \in E, c\in \Cc}, \bar{\sS}(t)$ and $\bar{\rR}(t)$ as follows:
\begin{align}
   \bar{\pP}(t) &= \frac{1}{t} \sum_{h= 0}^{t-1} \pP^\pi(h) \label{linear_control2} \\
    \bar{\sS}(t) &= \frac{1}{t}  \sum_{h= 0}^{t-1}  \sS(h) \\ 
    \bar{\rR}(t) &= \frac{1}{t}  \sum_{h= 0}^{t-1}  \rR(h). 
\end{align}
Clearly, we have that:
\begin{align}
    A \begin{bmatrix}\bar{\pP}(t) \\ \bar{\sS}(t) \\ \bar{\rR}(t) \end{bmatrix}  = \frac{1}{t}\sum_{h=0}^{t-1} A \begin{bmatrix}\pP^\pi(h) \\ \sS(h) \\ \rR(h) \end{bmatrix} 
\overset{\eqref{linear_control1}}{ \leq} \bB. 
\end{align}
Therefore, $\bar{\pP}(t)   \in  \widetilde{\Dd}(\gma, \qQ), \forall t\geq 0$. Since $p^{\pi c}_{ki}(h, \tau) \geq 0$ and, by \eqref{prob_constraint}, $p^{\pi c}_{ki}(h, \tau) \leq 1$, we know that $\pP^\pi(t) \in [0, 1]^{T_s |E||\Cc|}$. Thus, the vector $\bar{\pP}(t) \in [0, 1]^{T_s |E||\Cc|}$ has all of its entries bounded. By the Bolzano–Weierstrass Theorem, there exists a convergent subsequence $\{\bar{\pP}(t_h) \}_{h= 1, 2,...} $ where $\lim_{h\to \infty} t_h = \infty$. Define the limit of the convergent subsequence by:
\begin{align}
    \pP = \lim_{h\to \infty} \bar{\pP}(t_h). \label{convergent_subseries1}
\end{align}
Now, we proceed to prove the the statement (i) of the Lemma. From Lemma \ref{skorokhod}, we first obtain that under policy $\pi$:
\begin{align*}
    \sum_{q=0}^{t_h-1} a^c_i(q) -\sum_{q=0}^{t_h-1} \mu^c_i(q)  = a^c_i(0, t_h)- \mu^c_i(0, t_h) \leq Q^c_i(t_h).
\end{align*}
Taking expectation of both sides in view of \eqref{service_unit1}, we have:
\begin{align*}
  t_h \lambda^c_i    -\sum_{q=0}^{t_h-1}  \sum_{k\in N(i)}  \sum_{\tau=0}^{T_s-1} r_{ki}  p^{\pi c}_{ki}(q, \tau)  \leq \EE[\QQ^\pi(t_h)].
\end{align*}
Dividing both sides by $t_h$ and using \eqref{linear_control2}, we have:
\begin{align*}
    \lambda^c_i    \leq   \sum_{k\in N(i)}  \sum_{\tau=0}^{T_s-1} r_{ki}  \bar{p}^{ c}_{ki}(t_h, \tau)  + \frac{\EE[\QQ^\pi(t_h)]}{t_h}.
\end{align*}
Taking $h\to \infty$ and using $\EE[\QQ^\pi(T)] = o(T)$ and \eqref{convergent_subseries1}, we obtain that:
\begin{align}
    \lambda^c_i \leq \sum_{k\in N(i)} \sum_{\tau=0}^{T_s-1} r_{ki}  p^c_{ki}(\tau), \label{linear_control5}
\end{align}
which proves the required statement (i). 

Next, we proceed to prove the statement (ii). By \eqref{convergent_subseries1} and that $\bar{\pP}(t)   \in  \widetilde{\Dd}(\gma, \qQ), \forall t\geq 0$, we have $\pP \in \widetilde{\Dd}(\gma, \qQ)$. Thus, $ \pP^{\pi^*}(t)= \pP \in  \widetilde{\Dd}(\gma, \qQ), \forall t\geq 0$.
Using Lemma \ref{lma_driftbound} and given $B_1 = |\Ii||\Cc| A_{max}^2+ T_s^2 |\Kk|^2$, we have the following drift bound under the action of $\pi^*$:
\begin{align}
\nonumber
  \Delta^{ \pi^*} (t) &\leq B_1 +   \sum_{c\in \Cc}\sum_{i\in I} Q^{\pi^* c}_i(t) \big( \lambda^c_i - \sum_{\tau = 0}^{T_s -1 }\sum_{k\in N(i)} r_{ki}  p^{ \pi^* c}_{ki}(t, \tau)  \big) \\
  &\overset{\eqref{linear_control5}}{\leq} B_1   \label{linear_control6}
\end{align}
Taking iterated expectation of the above in view of \eqref{lyapu_drift}, we obtain that:
\begin{align*}
    \EE[ \Ll(\QQ^{\pi^*}(t+1))] - \EE[ \Ll(\QQ^{\pi^*}(t))] \leq B_1.
\end{align*}
Telescoping for $t = 0 \to T-1$, we have:
\begin{align}
    \EE[ \Ll(\QQ^{\pi^*}(T))] \leq T B_1. \label{linear_control7}
\end{align}
From Lemma \ref{lma_cauchy}, we have: 
\begin{align}
\nonumber
     \EE[\|\QQ^{\pi^*}(T)\|_1] &\leq \sqrt{|\Cc| |\Ii| \cdot \EE[\Ll(\QQ^{\pi^*}(T))]}\\
     \nonumber
     &\overset{\eqref{linear_control7}}{\leq } \sqrt{|\Cc| |\Ii| B_1 T}\\
     \therefore  \EE[\|\QQ^{\pi^*}(T)\|_1] &= o(T),
\end{align}
which shows that the stationary randomized policy $\pi^*$ stabilizes the queues. It is left to show that the policy $\pi^*$ attains better objective than that of the policy $\pi$. By the concavity of the utility $f(.)$, we have:
% f(\pP^{\pi^*}(t))
\begin{align}
\nonumber
    f(\bar{\pP}(t_h)) \overset{\eqref{linear_control2}}{=} f\big( \frac{1}{t_h} \sum_{q= 0}^{t_h-1} \pP^\pi(q)\big) \geq \frac{1}{t_h} \sum_{q= 0}^{t_h-1} f (\pP^\pi(q))
\end{align}
Taking $h\to \infty$ and noting that $\lim_{h\to \infty} f( \bar{\pP}(t_h)) = f(\pP) = f(\pP^{\pi^*}(t))$ by the continuity of $f(.)$, we obtain that:
\begin{align}
\nonumber
    &f(\pP^{\pi^*}(t)) \geq F(\pP^\pi)\\
    \nonumber
    \therefore &F(\pP^{\pi^*} ) = \liminf_{T\to \infty}\frac{1}{T} \sum_{t=0}^{T-1}f(\pP^\pi(t)) \geq F(\pP^\pi),
\end{align}
which concludes the proof of the Lemma. 

% For the policy $\pi$ serving  such that $ \pP^\pi(t) \in  \widetilde{\Dd}(\gma, \qQ), \forall t\geq 0$, 
    
%     by the characterization of $\widetilde{\Dd}(\gma, \qQ)$ in Theorem \ref{thm_dual}, we have the followings:

% \begin{align*}
%     \widetilde{\Dd}(\gma, \qQ) =  \big\{ &\pP(t) \in [0, 1]^{T_s |E| |\Cc|}\text{: such that \eqref{prob_constraint}, \eqref{pmax1},} \\
%     &\quad \quad  \quad |\Kk|T_s \gamma^c_i \leq \sum_{k\in N(i)} \sum_{\tau=0}^{T_s-1} p^c_{ki}(t, \tau)- \Gamma^c_i s^c_i  \\
%       &\quad \quad \quad \quad \quad -  \sum_{k\in N(i)} \sum_{\tau=0}^{T_s-1} r^c_{i, k \tau},  \forall i\in I, c\in \Cc,\\
%       & \quad \quad  \quad  s^c_i \geq 0, \forall  i\in I, c\in \Cc,\\
%       &\text{ and } \forall   i\in I,k\in N(i), c\in \Cc, \tau \in[0, T_s-1]:\\ 
%       & \quad \quad  \quad  s^c_i + r^c_{i,k \tau} \geq p^c_{ki}(t, \tau),\\% \forall   i\in I,k\in N(i), c\in \Cc, \tau \in[0, T_s-1], \\
%      &  \quad \quad  \quad s^c_i + r^c_{i, k \tau} \geq 1- p^c_{ki}(t, \tau), \\%  \forall i\in I,k\in N(i), c\in \Cc, \tau \in[0, T_s-1],  \\
%      &  \quad \quad  \quad r^c_{i, k \tau} \geq 0\big\}.%,  \forall i\in I,k\in N(i) ,c\in \Cc, \tau \in[0, T_s-1],  \\
% \end{align*}

%     \begin{align}
%          &\sum_{c\in \Cc} \sum_{i\in N(k)} p^{\pi c}_{ki}(t, \tau) = 1  \forall k\in\Kk,  \tau\in [0, T_s-1].\label{linear_control1}\\
%           &\sum_{k\in  N(i)} p^{\pi c}_{ki}(t, \tau)  \leq p_{max},  \forall i\in \Ii, c\in \Cc, \tau\in [0, T_s-1]\label{linear_control2}\\
%     \end{align}
\end{proof}

The following Lemma provides an equivalent form of the problem $(\AOP_{\lmb, \gma, \qQ})$.

\begin{lemma}
\label{lma_control_opti_convert}
The  optimal utility of the problem $(\AOP_{\lmb, \gma, \qQ})$ is equal to the optimal objective of the following optimization problem:
\begin{align}
\nonumber
      (\opti_{\lmb, \gma, \qQ}): \quad   \max_{\pP}   &f(\pP)\\
      \nonumber
         \text{such that}\quad &\pP \in \widetilde{\Dd}(\gma, \qQ) ,\\
          &\lambda^c_i \leq \sum_{k\in N(i)} \sum_{\tau=0}^{T_s-1}r_{ki}  p^c_{ki}(\tau) , \forall i\in \Ii, c\in \Cc.\label{lmb_constraint1}
\end{align}
\end{lemma}

\begin{proof}
Consider the optimal policy $\tilde{\pi}$ to the problem $(\AOP_{\lmb, \gma, \qQ})$. By Lemma \ref{lma_linear_control}, there exists some $\pP \in \widetilde{\Dd}(\gma, \qQ)$, which satisfies $\lambda^c_i \leq \sum_{k\in N(i)} \sum_{\tau=0}^{T_s-1} p^c_{ki}(\tau) , \forall i\in \Ii, c\in \Cc$ and thus is a feasible solution to the above problem $(\opti_{\lmb, \gma, \qQ})$, and a stationary randomized policy $\pi^*$ such that $\pP^{\pi^*}(t) = \pP$ that is feasible for  the problem $(\AOP_{\lmb, \gma, \qQ})$ and achieves the utility $F(\pP^{\pi^*}) \geq F(\pP^{\tilde{\pi}})$. Since $\pi^*$ makes stationary randomized decision $\pP^{\pi^*}(t) = \pP$, we have $F(\pP^{\pi^*}) = \liminf_{T\to \infty} \frac{1}{T}\sum_{t=0}^{T-1} f(\pP) = f(\pP)$.    
As $\tilde{\pi}$ is optimal for $(\AOP_{\lmb, \gma, \qQ})$, we  have $F(\pP^{\pi^*}) \leq F(\pP^{\tilde{\pi}})$ and thus  $F(\pP^{\pi^*}) = F(\pP^{\tilde{\pi}}) = f(\pP)$. Assume for the sake of contradiction that $\pP$ is not an optimal solution to the problem  $(\opti_{\lmb, \gma, \qQ})$, i.e. there exists some optimal solution $\pP'$ to the problem $(\opti_{\lmb, \gma, \qQ})$ with $f(\pP') > f(\pP)$. Then, similarly to the proof of Lemma  \ref{lma_linear_control}, we can construct a stationary randomized policy $\tilde{\pi}'$ via $\pP^{\tilde{\pi}' }(t) = \pP'$ that is feasible for  the problem $(\AOP_{\lmb, \gma, \qQ})$ and achieves the utility of $F(\pP^{\tilde{\pi}' }) = f(\pP') > F(\pP^{\tilde{\pi}}) = f(\pP)$; this is contradiction because $\tilde{\pi}$ is the optimal policy to the problem $(\AOP_{\lmb, \gma, \qQ})$. Therefore, $\pP$ is the optimal solution to the problem $(\opti_{\lmb, \gma, \qQ})$. With $F(\pP^{\tilde{\pi}}) = f(\pP)$ being the optimal utility of the problem $(\AOP_{\lmb, \gma, \qQ})$, we conclude the proof of the Lemma.  
\end{proof}

Finally, the next Lemma establishes the variation on the optimal objective of $(\AOP_{\lmb, \gma, \qQ})$ for varying $\lmb$ and $\gma$.

\begin{lemma}
\label{lma_num_distort}
Let $\pP^*$ and ${\pP^*}'$ respectively be the optimal solution to the optimization problems  $(\opti_{\lmb, \gma, \qQ})$ and $(\opti_{\lmb', \gma', \qQ})$ for some $\lmb' = \psi \lmb$ (where $\psi > 1$ is some constant) and $\gma' \leq \gma$. Define $\eE = \{ \mathbbm{1}\{\lambda^c_i >0\}\}_{i\in \Ii, c\in \Cc}$.
If Assumption \ref{assum_slater}  holds  and there exists  some $\varepsilon'>0$ such that $\lmb + \varepsilon' \eE \in   \widetilde{\Lmb}(\gma, \qQ)$ and $\lmb' + \varepsilon' \eE \in   \widetilde{\Lmb}(\gma', \qQ)$,  the following holds:
% for $ K_5 \triangleq  \frac{ f_{max}}{ \min\{\zeta, \varepsilon'\}}$:
\begin{align}
\nonumber
    |f(\pP^*) - f({\pP^*}')| \leq &  \frac{f_{max} - f_{min}}{\varepsilon'} \| \lmb -\lmb'\|_\infty \\
    &+ \frac{f_{max} - f_{min}}{\zeta}\| \gma - \gma'\|_1. \label{num_distort1}
\end{align}
Consequently, if $\pi$ and ${\pi}'$ are respectively optimal policies to the control problems $(\AOP_{\lmb, \gma, \qQ})$ and $(\AOP_{\lmb', \gma', \qQ})$, then we have:
\begin{align}
 \nonumber
| F(\pP^{\pi}) - F(\pP^{\pi'}) | \leq&  \frac{f_{max} - f_{min}}{\varepsilon'} \| \lmb -\lmb'\|_\infty \\
&+ \frac{f_{max} - f_{min}}{\zeta}\| \gma - \gma'\|_1.\label{num_distort2}
\end{align}
\end{lemma}
\begin{proof}
We additionally consider $\bar{\pP}^*$ as the optimal solution to the optimization problem $(\opti_{\lmb, \gma', \qQ})$.
    Denote by $\Mm(\lmb, \gma, \qQ)$ the feasible domain of $(\opti_{\lmb, \gma, \qQ})$.  In order to prove \eqref{num_distort1}, we first have: 
\begin{align}
     |f(\pP^*) - f({\pP^*}')| \leq  |f(\pP^*) - f(\bar{\pP}^*)| +  |f(\bar{\pP}^*) - f({\pP^*}')|, \label{proof_lma_num_distort0}
\end{align}
and proceed to bound each term on the RHS of \eqref{proof_lma_num_distort0} in the following. 
    
 \textbf{Bounding $|f(\pP^*) - f(\bar{\pP}^*)|$:}   We consider  the  decompositions $ \Mm(\lmb, \gma, \qQ) = S(\lmb) \cap M(\gma)$ and $ \Mm(\lmb, \gma', \qQ) = S(\lmb) \cap M( \gma')$, where the convex set: $$  S(\lmb) =\big\{ \pP \in [0, 1]^{T_s |E| |\Cc|}\text{: such that \eqref{prob_constraint}, \eqref{pmax1} and \eqref{lmb_constraint1}}\big\} $$ is comprised of  the mutual linear constraints in $\Mm(\lmb, \gma, \qQ)$ and $\Mm(\lmb, \gma', \qQ)$ and the convex set $M( \gma)$ is defined as in \eqref{m_gamma_def1}.
 %    formally defined by:
 % \begin{align*}
 %   M(\lmb, \gma)  &= \{ \pP \in [0, 1]^{T_s |E| |\Cc|}:  \forall i\in I, c\in \Cc: \\
 %    &|\Kk|T_s \gamma^c_i \leq \sum_{k\in N(i)} \sum_{\tau=0}^{T_s-1} r_{ki}  p^c_{ki}( \tau) - B^c_i(\Gamma^c_i, \pP),\\
 %   & \lambda^c_i \leq \sum_{k\in N(i)} \sum_{\tau=0}^{T_s-1} r_{ki}  p^c_{ki}(\tau)\}. 
 % \end{align*}
% Similar to the proof of Lemma \ref{lma_perturbed_obj},
We consider the indicator function $  \mathbbm{1}_{S(\lmb)}(\pP)$ as follows:
\begin{align}
 \label{indicator1.5}
     \mathbbm{1}_{S(\lmb)}(\pP) = \begin{cases} 0 \quad \text{, if $\pP\in S(\lmb)$}\\ +\infty,  \text{otherwise}\end{cases},
 \end{align}
Similar to the proof of Lemma \ref{lma_perturbed_obj}, we consider 
the  penalty function similar to \eqref{penalty} as follows:
\begin{align}
\nonumber
    &P_{r, \gma}(\pP)= r \sum_{(i, c)\in \Ss} \min \bigg\{ 0, \\
    &\frac{1}{|\Kk|T_s}\sum_{k \in N(i)} \sum_{\tau=0}^{T_s-1} r_{ki} p^c_{ki}( \tau) - \frac{1}{|\Kk|T_s}B^c_i(\Gamma^c_i, \pP) -  \gamma^c_i \bigg\}, \label{penalty_num}
\end{align}
 where $r>0$ is some hyper-parameter to be determined later. 
% Nevertheless, in order to penalize the linear inequality constraints in $ M(\lmb, \gma, \qQ) $, we use the following different penalty function  for some hyper-parameter $r>0$ to be determined later:
% \begin{align}
% \nonumber
%     &P_{r, \lmb, \gma}(\pP) = r \sum_{i \in I, c \in \Cc} \min \bigg\{ 0, \\
%     \nonumber 
%     &\frac{1}{|\Kk|T_s}\sum_{k \in N(i)} \sum_{\tau=0}^{T_s-1}r_{ki}  p^c_{ki}(\tau) - \frac{1}{|\Kk|T_s}B^c_i(\Gamma^c_i, \pP) -  \gamma^c_i \bigg\}\\
%     &+ r \sum_{i \in I, c \in \Cc} \min \bigg\{ 0,  \sum_{k\in N(i)} \sum_{\tau=0}^{T_s-1}r_{ki}  p^c_{ki}(\tau) - \lambda^c_i 
%  \bigg\}. \label{penalty_num}
% \end{align}
By defining the following new concave objective:
\begin{align}
    \bar{f}(\pP) = f(\pP) - \mathbbm{1}_{S(\lmb)}(\pP), \label{def_f_bar}
\end{align}
we have:
\begin{align}
\nonumber
    f(\pP^*) &= \max_{\pP \in \Mm(\lmb, \gma, \qQ)} f(\pP) \\
    &= \max_{\pP \in M( \gma)} \{ \bar{f}(\pP) = f(\pP) - \mathbbm{1}_{S(\lmb)}(\pP) \}, \label{proof_lma_num_distort} \\
    \nonumber
    f(\bar{\pP}^*) &= \max_{\pP \in \Mm(\lmb, \gma', \qQ)} f(\pP) \\
    &= \max_{\pP \in M( \gma')} \{ \bar{f}(\pP) = f(\pP) - \mathbbm{1}_{S(\lmb)}(\pP) \}, \label{proof_lma_num_distort2}
\end{align}
% where $M( \gma)$ and $M(\gma')$ contain only linear inequality constraints. 
Furthermore, the next Lemma, whose proof is given in Appendix \ref{appen_lma_bound_f}, establishes a bound on the objective $\bar{f}(\pP)$ for any $\pP$. 
\begin{lemma}
\label{lma_bound_f}
    We have the following bound:
    \begin{align}
        \bar{f}(\pP) \leq f_{max}. \label{bound_f_eq}
    \end{align}
\end{lemma}
% The following Lemma, whose proof is deferred to Appendix \ref{appen_big_penalty}, provides the bound on the variation of $P_{r, \lmb, \gma}(\pP)$ for perturbing $\lmb$ and $\gma$.
% \begin{lemma}
% \label{lma_big_penalty}
%     For any fixed $r>0$ and $\pP$, we have:
% \begin{align}
%     | P_{r, \lmb, \gma}(\pP) -  P_{r, \lmb', \gma'}(\pP)| \leq r \big( \|\lmb - \lmb'  \|_1 +   \|\gma  -\gma'  \|_1\big).\label{big_penalty_eq}
% \end{align}
% \end{lemma}
Next, leveraging results from \cite{penalty1}, we convert \eqref{proof_lma_num_distort} and \eqref{proof_lma_num_distort2} into unconstrained optimization problems by incorporating the penalty \eqref{penalty_num} into their objective for large enough $r$.

\begin{lemma}
\label{lma_penalty_num}
For  $r \geq  \frac{ f_{max} - f_{min}}{\zeta}$, we have:
\begin{align}
    f(\pP^*) = \max_{\pP} \{ \bar{f}(\pP) + P_{r,  \gma}(\pP) \}, \label{lma_penalty_num_eq1} \\
    f(\bar{\pP}^*) = \max_{\pP} \{ \bar{f}(\pP) + P_{r,  \gma'}(\pP) \}.\label{lma_penalty_num_eq2}
\end{align}
\end{lemma}
The proof of Lemma \ref{lma_penalty_num} is given in Appendix \ref{appen_lma_penalty_num}.  Now, by setting $r= \frac{ f_{max} - f_{min}}{\zeta}$ and from \eqref{lma_penalty_num_eq2} and Lemma \ref{penalty_diff}, we have:
\begin{align}
\nonumber
    f(\bar{\pP}^*)  &= \bar{f}(\bar{\pP}^*) + P_{r,\gma'}(\bar{\pP}^*) \\
    \nonumber
    &\leq  \bar{f}(\bar{\pP}^*) + P_{r, \gma}(\bar{\pP}^*)  + \frac{ f_{max} - f_{min} }{\zeta} \| \gma - \gma'\|_1 \\
    \nonumber
    &\leq  \max_{\pP} \{ \bar{f}(\pP )+ P_{r, \gma}(\pP) \}  +  \frac{ f_{max} - f_{min}}{\zeta} \| \gma - \gma'\|_1\\
    &\overset{\eqref{lma_penalty_num_eq1}}{=}   f({\pP^*})+ \frac{ f_{max} - f_{min}}{\zeta} \| \gma - \gma'\|_1. \label{proof_lma_num_distort3}
\end{align}
Similarly, we have:
\begin{align}
     f({\pP^*}) \leq   f(\bar{\pP}^*) + \frac{ f_{max} - f_{min}}{\zeta} \| \gma - \gma'\|_1.  \label{proof_lma_num_distort4}
\end{align}
Combining \eqref{proof_lma_num_distort3} and \eqref{proof_lma_num_distort4}, we obtain that:
\begin{align}
    |  f({\pP^*}) -   f(\bar{\pP}^*) | \leq \frac{ f_{max} - f_{min}}{\zeta} \| \gma - \gma'\|_1.  \label{proof_lma_num_distort5}
\end{align}

\textbf{Bounding $|f(\bar{\pP}^*) - f({\pP^*}')|$:} 
We consider the convex domain of the linear constraints \eqref{lmb_constraint1} as follows:
\begin{align*}
  M(\lmb) =  \{\pP &\in [0, 1]^{T_s |E| |\Cc|}: \text{ such that } \\   
     &\lambda^c_i \leq \sum_{k\in N(i)} \sum_{\tau=0}^{T_s-1}r_{ki}  p^c_{ki}(\tau) , \forall i\in \Ii, c\in \Cc: \lambda^c_i > 0.\}
\end{align*}
Note that any such constraint with  $\lambda^c_i=0$ is a redundant constraint (i.e. trivially satisfied given that $\pP \in [0, 1]^{T_s |E| |\Cc|}$) and is thus neglected. Then, we have the decompositions $ \Mm(\lmb, \gma', \qQ) =\widetilde{\Dd}(\gma', \qQ) \cap M(\lmb)$ and $ \Mm(\lmb', \gma', \qQ) = \widetilde{\Dd}(\gma', \qQ) \cap M( \lmb')$. We consider the indicator function $  \mathbbm{1}_{\widetilde{\Dd}(\gma', \qQ)}(\pP)$ as follows:
\begin{align}
 \label{indicator2}
     \mathbbm{1}_{\widetilde{\Dd}(\gma', \qQ)}(\pP) = \begin{cases} 0 \quad \text{, if $\pP\in \widetilde{\Dd}(\gma', \qQ)$}\\ +\infty,  \text{otherwise}\end{cases}.
 \end{align}
By defining the following new concave objective:
\begin{align}
    \tilde{f}(\pP) = f(\pP) - \mathbbm{1}_{\widetilde{\Dd}(\gma', \qQ)}(\pP),  \label{def_tilde_f}
\end{align}
we have:
\begin{align}
 \nonumber
    f(\bar{\pP}^*) &= \max_{\pP \in \Mm(\lmb, \gma', \qQ)} f(\pP) \\
    &= \max_{\pP \in M( \lmb)} \{ \tilde{f}(\pP) = f(\pP) - \mathbbm{1}_{\widetilde{\Dd}(\gma', \qQ)}(\pP) \}, \label{proof_lma_num_distort6}\\
\nonumber
    f({\pP^*}') &= \max_{\pP \in \Mm(\lmb', \gma', \qQ)} f(\pP) \\
    &= \max_{\pP \in M( \lmb')} \{ \tilde{f}(\pP) = f(\pP) - \mathbbm{1}_{\widetilde{\Dd}(\gma', \qQ)}(\pP) \}. \label{proof_lma_num_distort7} 
\end{align}
Furthermore, the next Lemma, whose proof is given in Appendix \ref{appen_lma_bound_f_tilde}, establishes a bound on the objective $\bar{f}(\pP)$ for any $\pP$. 
\begin{lemma}
\label{lma_bound_f_tilde}
    We have the following bound:
    \begin{align}
        \tilde{f}(\pP) \leq f_{max}. \label{bound_f_tilde_eq}
    \end{align}
\end{lemma}
Note that we can equivalently write \eqref{proof_lma_num_distort6} and \eqref{proof_lma_num_distort7} as solving convex minimization over linear constraints via: 
\begin{align}
    \max_{\pP \in M( \lmb)}  \tilde{f}(\pP) = - \min_{\pP \in M( \lmb)} - \tilde{f}(\pP), \label{proof_lma_num_distort8}\\
    \max_{\pP \in M( \lmb')}  \tilde{f}(\pP) = - \min_{\pP \in M( \lmb')} - \tilde{f}(\pP). \label{proof_lma_num_distort9}
\end{align}
Since $\lmb + \varepsilon' \eE \in   \widetilde{\Lmb}(\gma, \qQ)$ and $\lmb' + \varepsilon' \eE \in   \widetilde{\Lmb}(\gma', \qQ)$, the Slater condition holds with constant $\varepsilon'$ for $ M(\lmb)$ and $ M(\lmb')$, i.e. by Theorem \ref{thm_charac_stability_region}, there exists $\bar{\pP} \in  \widetilde{\Dd}(\gma', \qQ) $ and $\bar{\pP}' \in  \widetilde{\Dd}(\gma', \qQ) $ such that:
\begin{align}
    &\lambda^c_i  + \varepsilon'\leq \sum_{k\in N(i)} \sum_{\tau=0}^{T_s-1}r_{ki}  \bar{p}^c_{ki}(\tau) , \forall i\in \Ii, c\in \Cc: \lambda^c_i>0,  \label{proof_lma_num_distort10}  \\
    &{\lambda^c_i}' + \varepsilon' \leq \sum_{k\in N(i)} \sum_{\tau=0}^{T_s-1}r_{ki}  {\bar{p}^c_{ki}}{}'(\tau) , \forall i\in \Ii, c\in \Cc: \lambda^c_i>0.  \label{proof_lma_num_distort11}
\end{align}
Note that we have:
\begin{align}
    &-\tilde{f}(\bar{\pP}) \overset{\eqref{def_tilde_f}}{= } -f(\bar{\pP}) \leq -f_{min}, \label{proof_lma_num_distort12}\\
    &-\tilde{f}(\bar{\pP}') \overset{\eqref{def_tilde_f}}{= } -f(\bar{\pP}') \leq -f_{min} \label{proof_lma_num_distort13}
\end{align}
Applying Theorem \ref{thm_perturbation_bound_no_v_in_statement} into \eqref{proof_lma_num_distort8} and \eqref{proof_lma_num_distort9}, we obtain that:
\begin{align}
    |f(\bar{\pP}^*) - f({\pP^*}')| \leq \kappa \|\gma -\gma' \|_\infty \label{proof_lma_num_distort14}
\end{align}
with
\begin{align}
\nonumber
    \kappa &= \frac{1}{\varepsilon'} \max\big\{- \tilde{f}(\bar{\pP}) + \tilde{f}( \bar{\pP}^*), - \tilde{f}(\bar{\pP}') +\tilde{f}( {\pP^*}')  \big\}\\
    &\leq \frac{f_{max} - f_{min}}{\varepsilon'} , \label{proof_lma_num_distort15}
\end{align}
where the inequality follows from \eqref{bound_f_tilde_eq}, \eqref{proof_lma_num_distort12} and \eqref{proof_lma_num_distort13}. Plugging \eqref{proof_lma_num_distort15} into \eqref{proof_lma_num_distort14}, we obtain that:
\begin{align}
   |f(\bar{\pP}^*) - f({\pP^*}')| \leq   \frac{f_{max} - f_{min}}{\varepsilon'} \, \|\gma -\gma' \|_\infty.  \label{proof_lma_num_distort16}
\end{align}
Applying \eqref{proof_lma_num_distort5} and \eqref{proof_lma_num_distort16} to \eqref{proof_lma_num_distort0}, 
we conclude the required statement \eqref{num_distort1} of the Lemma. 
Now   note that by Lemma \ref{lma_control_opti_convert},  we have $F(\pP^\pi) = f(\pP^*)$ and $F(\pP^{\pi'}) = f({\pP^*}')$. Plugging these into \eqref{num_distort1}, we conclude the required statement \eqref{num_distort2} of the Lemma.
\end{proof}

\subsection{Proof of Theorem \ref{thm_charac_stability_region}}
\label{appen_thm_charac_stability_region}

The Theorem is the immediate consequence of Lemma \ref{lma_linear_control}. For brevity, we denote the set characterization of $ \widetilde{\Lmb}(\gma, \qQ)$ by:
\begin{align*}
    S \triangleq \{\lmb& = \{\lambda^c_i\}_{i\in \Ii, c\in \Cc}\in \Br_{\geq0}^{|\Ii||\Cc|}: \\
    &\lambda^c_i \leq \sum_{k\in N(i)} \sum_{\tau=0}^{T_s-1}r_{ki}  p^c_{ki}(\tau) , \forall i\in \Ii, c\in \Cc,\\
    &   \pP = \{p^c_{ki}(\tau) \}_{\tau\in[0, T_s-1], (k,i) \in E, c\in \Cc} \in  \widetilde{\Dd}(\gma, \qQ)\},  
\end{align*}
and proceed to show that $ \widetilde{\Lmb}(\gma, \qQ) = S$. For any $\lmb \in \widetilde{\Lmb}(\gma, \qQ)$, by statement (i) of Lemma \ref{lma_linear_control}, there exists some $\pP \in \widetilde{\Dd}(\gma, \qQ)$ such that: 
\begin{align}
 \lambda^c_i \leq \sum_{k\in N(i)} \sum_{\tau=0}^{T_s-1} r_{ki}  p^c_{ki}(\tau) , \forall i\in \Ii, c\in \Cc. \label{charac_stability_region1}
\end{align}
This implies  $\lambda \in S$ and thus $ \widetilde{\Lmb}(\gma, \qQ) \subseteq S$. Conversely, for any $\lmb \in S$, i.e. there exists $\pP\in  \widetilde{\Dd}(\gma, \qQ)$ satisfying  \eqref{charac_stability_region1}, by statement (ii) of Lemma \ref{lma_linear_control}, the policy $\pi^*$ with $\pP^{\pi^*}(t) = \pP$ is a feasible policy that stabilizes the queues, i.e. $\EE[\QQ^{\pi^*}(T)] = o(T)$. This implies $\lmb\in \widetilde{\Lmb}(\gma, \qQ)$ and thus $S \subseteq \widetilde{\Lmb}(\gma, \qQ)$. We hence conclude that $\widetilde{\Lmb}(\gma, \qQ) = S$.

% Recall that the problem $  \mathbf{(P)}$ seeks the policy that solves the utility maximization problem given the arrival rate $\lmb = \{ \lambda^c_i\}_{c\in \Cc, i\in \Ii}$, and subject to the feasible domain $\Dd(\gma, \qQ)$ and queue stability constraint.
% Within the scope of the proof, in addition to the problem $  \mathbf{(P)}$, we now consider the problem $  \mathbf{(\widetilde{P}_{\lmb, \gma, \qQ}) }$ that instead finds the optimal scheduling decision within the linearized feasible domain $\widetilde{\Dd}(\gma, \qQ)$.
%  \begin{align*}
%       \mathbf{(\widetilde{P}_{\lmb, \gma, \qQ}) :} \quad   \max_{\pP=\{\pP(t)\}_{t\geq0}}   &F(\pP) = \lim_{T\to \infty}\frac{1}{T} \sum_{t=0}^{T-1}f(\pP(t))\\
%          \text{such that}\quad &\pP(t) \in \widetilde{\Dd}(\gma, \qQ) \\
%          &  \EE[\QQ(T)] = o(T).
% \end{align*}

% ; we explicitly specify the parameters $\lmb, \gma$ and $\qQ$ for this problem $\mathbf{(\widetilde{P}_{\lmb, \gma, \qQ}) }$  as we will consider alternative values of them to assist the proof. 

% given the arrival rate $\lmb = \{ \lambda^c_i\}_{c\in \Cc, i\in \Ii}$ and subject to the feasibility domain $\widetilde{\Dd}(\gma, \qQ)$ and queue stability constraint. 

\subsection{Proof of Theorem \ref{thrm_mdp}}
\label{appen_thrm_mdp}

Besides the optimal policy $\pi^*$ of the problem $(\op_{\lmb, \gma, \qQ})$, we also consider the following policies of their respective problems:
\begin{itemize}
    \item the policy $\pi_1$ that is optimal for the problem $(\AOP_{\lmb, \gma, \qQ})$
    \item the policy $\pi_2$ that is optimal for the problem $(\AOP_{\lmb, \gma - \frac{K_1}{\sqrt{T_s}} \one, \qQ})$
    \item the policy $\pi_\varepsilon$  that is optimal for the problem $(\AOP_{\lmb_\varepsilon, \gma, \qQ})$ with $\lmb_\varepsilon = \frac{1}{1-0.5 \varepsilon} \lmb \in  \frac{1-\varepsilon}{1-0.5 \varepsilon} \widetilde{\Lmb}(\gma, \qQ) $
\end{itemize}

Since under condition \eqref{assum_charac}, we have by Theorem \ref{thrm_approx_bound} that  $\widetilde{\Dd}(\gma, \qQ)  \subseteq  {\Dd}(\gma, \qQ) \subseteq  \widetilde{\Dd}(\gma - \frac{K_1}{\sqrt{T_s}} \one, \qQ)$, we have $\pi_1$ being a feasible policy for the problem $(\op_{\lmb, \gma, \qQ})$ (whose optimal policy is $\pi^*$) and $\pi^*$ being a feasible policy for the problem $(\AOP_{\lmb, \gma - \frac{K_1}{\sqrt{T_s}} \one, \qQ})$ (whose optimal policy is $\pi_2$). Therefore, we obtain that: 
\begin{align}
\label{thrm_mdp1}
    F(\pP^{\pi_1}) \leq F(\pP^{*}) \leq F(\pP^{\pi_2}). 
\end{align}
Next, we break down the objective gap $ F(\pP^*) - F(\pP^\MDP)$ as follows:
\begin{align}
\nonumber
     &F(\pP^*) - F(\pP^\MDP) = \big( F(\pP^*) -  F(\pP^{\pi_1})\big) \\
     &\quad + \big( F(\pP^{\pi_1}) -  F(\pP^{\pi_\varepsilon}) \big) + \big(  F(\pP^{\pi_\varepsilon}) -  F(\pP^\MDP) \big), \label{thrm_mdp2}
\end{align}
and bound the first two terms in the following Lemma, whose proof is given in Appendix \ref{appen_lma_final_bounds}.

\begin{lemma}
\label{lma_final_bounds}
    We have the following bounds:
\begin{align}
    F(\pP^*) -  F(\pP^{\pi_1}) &\leq O\bigg(\frac{1}{\sqrt{T_s}}\bigg), \label{lma_final_bounds_1}\\
    F(\pP^{\pi_1}) -  F(\pP^{\pi_\varepsilon}) &\leq O\bigg(\frac{1}{{T_s}}\bigg). \label{lma_final_bounds_2}
\end{align}
\end{lemma}

Applying \eqref{lma_final_bounds_1} and \eqref{lma_final_bounds_2} into \eqref{thrm_mdp2}, we  obtain that:
\begin{align}
    F(\pP^*) - F(\pP^\MDP) \leq O\bigg(\frac{1}{\sqrt{T_s}}\bigg) +\big(  F(\pP^{\pi_\varepsilon}) -  F(\pP^\MDP) \big).  \label{thrm_mdp3}
\end{align}
It is thus left to bound  $ F(\pP^{\pi_\varepsilon}) -  F(\pP^\MDP)$. Since $\lmb_\varepsilon = \frac{1}{1-0.5 \varepsilon} \lmb \subseteq \frac{1-\varepsilon}{1-0.5 \varepsilon} \widetilde{\Lmb}(\gma, \qQ) \subseteq  \widetilde{\Lmb}(\gma, \qQ)$, by Lemma \ref{lma_linear_control}, there exists an stationary randomized policy that is optimal for the problem $(\AOP_{\lmb_\varepsilon, \gma, \qQ})$. Without loss of generality (WLOG), we assume that the considered $\pi_\varepsilon$ is one such stationary randomized optimal policy that follows from Lemma \ref{lma_linear_control}. Thus, we have $p^{\pi_\varepsilon c}_{ki}(t, \tau) = p^c_{ki}(\tau) , \forall t\geq 0$ for some $\pP \in  \widetilde{\Dd}(\gma, \qQ)$  that supports the arrival rate $\lmb_\varepsilon = \frac{1}{1-0.5 \varepsilon} \lmb $, i.e. 
\begin{align}
\nonumber
  [\lambda_\varepsilon]^c_i=  \frac{1}{1-0.5 \varepsilon} \lambda^c_i \leq \sum_{k\in N(i)} \sum_{\tau=0}^{T_s-1} r_{ki} 
 p^c_{ki}(\tau) , \forall i\in \Ii, c\in \Cc \\
 \lambda^c_i  \leq (1-0.5\varepsilon) \sum_{\tau=0}^{T_s-1} r_{ki} 
 p^{\pi_\varepsilon c}_{ki}(t, \tau), \forall i\in \Ii, c\in \Cc. \label{thrm_mdp4}
\end{align}
Furthermore, since $\pP^{\pi_\varepsilon}(t)\in  \widetilde{\Dd}(\gma, \qQ) \subseteq  {\Dd}(\gma, \qQ)$ and  $ \pP^{\DP}(t) = \argmaxE_{ \pP(t) \in \Dd(\gma, \qQ) } h_t(\pP(t))$ from \eqref{drift_plus_penalty1}, we have:
\begin{align}
    h_t(\pP^\DP(t)) \geq h_t(\pP^{\pi_\varepsilon}(t)).  \label{thrm_mdp5}
\end{align}
Next, we proceed to analyze the drift-plus-penalty $\Delta^\MDP (t) - V f(\pP^\MDP(t))$ of the MDP policy.  We consider the drift bound \eqref{drifting1} from Lemma \ref{lma_driftbound} under the effect of the MDP policy.
Now, subtracting $V f(\pP^\MDP(t))$ from both sides of \eqref{drifting1} and using \eqref{drift_obj1}, we have:
\begin{align}
\nonumber
     \Delta^\MDP (t) -  V f(\pP^\MDP(t))&\leq B_1 +     \sum_{c\in \Cc}\sum_{i\in I} Q^c_i(t)  \lambda^c_i - h_t(\pP^\MDP(t))\\
     \nonumber
     &\leq  B_1 +     \sum_{c\in \Cc}\sum_{i\in I} Q^c_i(t)  \lambda^c_i - h_t(\pP^\DP(t))\\
     \nonumber
     &\quad + K_2  \|\QQ(t)\|_{\infty}\sqrt{T_s}  +  \frac{V K_2 (f_{max} -f_{min})}{ |\Kk| \sqrt{T_s}}, \\
     \nonumber
     &\overset{\eqref{thrm_mdp5} }{\leq}  B_1 +     \sum_{c\in \Cc}\sum_{i\in I} Q^c_i(t)  \lambda^c_i - h_t(\pP^{\pi_\varepsilon}(t))\\
     &\quad  + K_2  \|\QQ(t)\|_{\infty}\sqrt{T_s}  +  \frac{V K_2 (f_{max} -f_{min})}{ |\Kk| \sqrt{T_s}},
     \label{thrm_mdp6}
\end{align}
where the second inequality follows from Lemma \ref{lma_obj_gap}. We next bound the term $  \sum_{c\in \Cc}\sum_{i\in I} Q^c_i(t)  \lambda^c_i - h_t(\pP^{\pi_\varepsilon}(t))$ in the RHS of \eqref{thrm_mdp6}:
\begin{align}
\nonumber
    &\sum_{c\in \Cc}\sum_{i\in I} Q^c_i(t)  \lambda^c_i - h_t(\pP^{\pi_\varepsilon}(t))  \\
    \nonumber
    &\overset{\eqref{drift_obj1}}{=} \sum_{c\in \Cc}\sum_{i\in I} Q^c_i(t) ( \lambda^c_i -\sum_{\tau=0}^{T_s-1} r_{ki} 
 p^{\pi_\varepsilon c}_{ki}(t, \tau) ) - V f(\pP^{\pi_\varepsilon}(t))\\
 &\overset{\eqref{thrm_mdp4}}{\leq} -0.5 \varepsilon  \sum_{c\in \Cc}\sum_{i\in I} Q^c_i(t) \big( \sum_{\tau=0}^{T_s-1} r_{ki} 
 p^{\pi_\varepsilon c}_{ki}(t, \tau) \big)- V f(\pP^{\pi_\varepsilon}(t)). \label{thrm_mdp7}
\end{align}
From \eqref{thrm_mdp4} and noting that $ \delta(\lmb) = \min_{i\in \Ii, c\in \Cc: \lambda^c_i > 0}\big\{ \frac{\lambda^c_i}{|\Kk| T_s}\big\} \geq \delta_{min}$, we have:
\begin{align}
    \sum_{\tau=0}^{T_s-1} r_{ki} p^{\pi_\varepsilon c}_{ki}(t, \tau) \geq |\Kk| T_s \delta(\lmb) \geq |\Kk| T_s \delta_{min}, \forall i\in \Ii, c\in \Cc: \lambda^c_i > 0. \label{thrm_mdp8.1}
\end{align}
If $\| \QQ(t)\|_\infty= Q^{c_0}_{i_0}(t) > 0$, then $\lambda^{c_0}_{i_0} > 0$; otherwise, we have $a^{c_0}_{i_0}(t) = 0, \forall t\geq 0$ and thus $Q^{c_0}_{i_0}(t) = 0, \forall t\geq 0$ with probability 1. Then, we have the following for the case $\| \QQ(t)\|_\infty>0$:
\begin{align}
\nonumber
    \sum_{c\in \Cc}\sum_{i\in I} Q^c_i(t) \big( \sum_{\tau=0}^{T_s-1} r_{ki}  p^{\pi_\varepsilon c}_{ki}(t, \tau) \big) &\geq \| \QQ(t)\|_\infty \big( \sum_{\tau=0}^{T_s-1} r_{ki_0}  p^{\pi_\varepsilon c_0}_{ki_0}(t, \tau) \big) \\
    &\overset{ \eqref{thrm_mdp8.1}}{\geq} \| \QQ(t)\|_\infty |\Kk| T_s \delta_{min}. \label{thrm_mdp8}
\end{align}
For the case $\| \QQ(t)\|_\infty=0$, \eqref{thrm_mdp8} trivially holds. 
% From Lemma \ref{lma_lower_bound_service} and noting that $\pP^{\pi_\varepsilon}(t) = \pP \in \widetilde{\Dd}(\gma, \qQ)$, we have:
% \begin{align}
%     \sum_{\tau=0}^{T_s-1} r_{ki} 
%  p^{\pi_\varepsilon c}_{ki}(t, \tau)  \geq  |\Kk| T_s \gamma_{min}, \forall i\in \Ii, c\in \Cc. \label{thrm_mdp8}
% \end{align}
Plugging \eqref{thrm_mdp8} into \eqref{thrm_mdp7}, we obtain that:
\begin{align}
    \nonumber
       &\sum_{c\in \Cc}\sum_{i\in I} Q^c_i(t)  \lambda^c_i - h_t(\pP^{\pi_\varepsilon}(t)) \\
       % \nonumber
       % &\leq  -0.5 \varepsilon  \sum_{c\in \Cc}\sum_{i\in I} Q^c_i(t) |\Kk| T_s \gamma_{min} -V f(\pP^{\pi_\varepsilon}(t)) \\
       \nonumber 
       &\leq  -0.5 \varepsilon T_s  \|\QQ(t) \|_\infty |\Kk|  \delta_{min}  -V f(\pP^{\pi_\varepsilon}(t)) \\
       &= - K_2  \|\QQ(t)\|_{\infty}\sqrt{T_s} -V f(\pP^{\pi_\varepsilon}(t)),\label{thrm_mdp9}
\end{align}
where the last line is by $\varepsilon =\frac{K_4}{\sqrt{T_s}} = \frac{2 K_2}{\sqrt{T_s}|\Kk|\delta_{min}}$. Plugging \eqref{thrm_mdp9} into \eqref{thrm_mdp6}, we obtain that: 
\begin{align}
\nonumber
      \Delta^\MDP (t) -  V f(\pP^\MDP(t))&\leq B_1 -V f(\pP^{\pi_\varepsilon}(t)) +  \frac{V K_2 (f_{max} -f_{min})}{ |\Kk| \sqrt{T_s}}. %\label{thrm_mdp10}
\end{align}
Noting that $\Delta^\MDP (t) = \EE\big[ \Ll(\QQ(t+1)) -  \Ll(\QQ(t)) \big| \QQ(t) \big]$ and taking iterated expectation of both sides of the above, we have: 
\begin{align*}
   &\EE\big[ \Ll(\QQ(t+1)) \big] -  \EE\big[\Ll(\QQ(t)) \big]\\
   &\leq B_1 + V (f(\pP^\MDP(t)) -f(\pP^{\pi_\varepsilon}(t)) ) +  \frac{V K_2 (f_{max} -f_{min})}{ |\Kk| \sqrt{T_s}}.
\end{align*}
Telescoping the above for $t=0 \to T-1$ and noting that $  \EE\big[ \Ll(\QQ(0))\big]=0$, we get:
\begin{align}
\nonumber
    \EE\big[ \Ll(\QQ(T))\big] &\leq B_1 T  + V \sum_{t=0}^{T-1} f(\pP^\MDP(t)) - V \sum_{t=0}^{T-1}f(\pP^{\pi_\varepsilon}(t)) \\
    &\quad +\frac{V T K_2 (f_{max} -f_{min})}{ |\Kk| \sqrt{T_s}} .\label{thrm_mdp11}
    % &= O\big(VT + T_s^2 T). \label{thrm_mdp11}
\end{align}
Using $f(\pP(t)) \in [f_{min}, f_{max}]$ and noting that $B_1 = |\Ii||\Cc| A_{max}^2+ T_s^2 |\Kk|^2$, we have:
\begin{align}
\nonumber
    \EE\big[ \Ll(\QQ(T))\big] &\leq B_1 T  + VT(f_{max} -f_{min})+ \frac{V T K_2 (f_{max} -f_{min})}{ |\Kk| \sqrt{T_s}}\\
    &= O\big(VT + T_s^2 T). \label{thrm_mdp12}
\end{align}
From Lemma \ref{lma_cauchy}, we have:
\begin{align*}
     \EE[\|\QQ(T)\|_1] \leq \sqrt{2|\Cc| |\Ii| \, \EE[\Ll(\QQ(T))]} \overset{\eqref{thrm_mdp12}}{=} O\big( \sqrt{VT} + T_s \sqrt{T}\big), 
\end{align*}
which is the required statement \eqref{thrm_mdp_eq1} of the Theorem. Now, dividing both sides of  \eqref{thrm_mdp11} by $VT$ and noting that $   \EE\big[ \Ll(\QQ(T))\big]\geq 0$, we have:
\begin{align*}
       \frac{\sum_{t=0}^{T-1}f(\pP^{\pi_\varepsilon}(t)) }{T}\leq \frac{\sum_{t=0}^{T-1} f(\pP^\MDP(t)) }{T} +\frac{B_1}{V} + \frac{ K_2 (f_{max} -f_{min})}{ |\Kk| \sqrt{T_s}}.
\end{align*}
Taking $\liminf$ for $T\to \infty$ of both sides and noting that  $B_1 = |\Ii||\Cc| A_{max}^2+ T_s^2 |\Kk|^2$, we obtain that:
\begin{align}
    F(\pP^{\pi_\varepsilon}(t)) \leq  F(\pP^\MDP(t)) + O\bigg(\frac{T_s^2}{V} + \frac{1}{\sqrt{T_s}}\bigg). \label{thrm_mdp13}
\end{align}
Plugging \eqref{thrm_mdp13} into \eqref{thrm_mdp3}, we get: 
\begin{align*}
      F(\pP^*) - F(\pP^\MDP) = O\bigg(\frac{T_s^2}{V} + \frac{1}{\sqrt{T_s}}\bigg),
\end{align*}
which is the required statement \eqref{thrm_mdp_eq2} of the Theorem.

\section{Necessity of Randomized Scheduling and Single Server Example}
\label{appen_rand_scheduling}

Consider a single server serving 101 clients with only one type of job. For brevity, we omit the superscript $c$ indexing the job type. 
 The channel statistics are the same, i.e. $r_{1i} = r_{max} = 0.8$ for all $ i\in[1, 101]$
 The orchestration time frame is $T_s = 300$ and the probabilistic constraints are specified by $\gamma_i = 6.72 \cdot 10^{-3}$ and $q_i = 2 \cdot 10^{-4}$ for all $i \in[1,101]$. Even under this simple setting with mild probabilistic QoS constraints, any deterministic policy fails to meet all the probabilistic constraints, as shown in the following Claim \ref{claim1}. 
 \begin{claim}
 \label{claim1}
     Under the considered setting and for any deterministic policy, at any time frame $t$, there exists some $i\in [1, 101]$ such that $P(\mu_i(t) > T_s \gamma_i) = 0$.
 \end{claim}
 \begin{proof}
   For any deterministic policy, within time frame $t$, let $n_i(t)$ be the number of time slots that the server activate link $(1, i)$ with client $i$. 
   As there are $T_s$ time slots, we have $\sum_{i=1}^{|I|} n_i(t) = T_s$. Thus, there exists some $i$ such that $n_i(t) \leq T_s / |I| <3$, or $n_i(t) \leq 2$ (as $n_i(t)$ is an integer). Within  time frame $t$, since the server attempts to serve client $i$ for $n_i(t)$ times, each of which is successful with probability $r_{max}$, the service that client $i$ receives  follows Binomial distribution as $\mu_i(t) \sim B(n_i(t), r_{max})$. Given that $n_i(t) \leq 2$, we obtain that  $P(\mu_i(t) > T_s \gamma_i) = P(\mu_i(t) > 2.016)= 0$.
 \end{proof}
We know that the MDP policy seeks for the scheduling decision $\pP(t)\in \widetilde{\Dd}(\gma, \qQ)$, which, if exists, by Theorem \ref{concen_bound} would meet all the probabilistic QoS constraints. The following claim shows that, in fact, such  feasible scheduling decision always exists.
\begin{claim}
    \label{claim2}
   The domain $\widetilde{\Dd}(\gma, \qQ)$ is non-empty.
\end{claim}
\begin{proof}
Consider the simple (randomized) scheduling decision \(\pP(t)\) given by
    $p_{1i}(t, \tau) = \frac{1}{|I|} = \frac{1}{101},  \forall \tau \in [T_s]$.
By direct verification, this choice of \(\pP(t)\) satisfies the constraints 
defining \(\widetilde{\Dd}(\gamma, \qQ)\), and hence the set is non-empty.
\end{proof}

\section{Application of the QoS framework: Throughput and Delay Guarantees}

\subsection{Proof of Corollary \ref{corol_throughput_bound}}
\label{appen_corol_throughput_bound}

The guarantee is obtained directly from the probabilistic QoS constraint. In particular, we have $\forall i\in \Ii, c\in \Cc$:
\begin{align*}
 \EE[\mu^c_i(t)] \geq |\Kk| T_s \gamma^c_i P\big(\mu^c_i(t) >  |\Kk| T_s \gamma^c_i \big)       \overset{\eqref{qos_constraint1}}{\geq} |\Kk| T_s \gamma^c_i q^c_i.
\end{align*}

\subsection{Proof of Theorem \ref{thrm_delay_bound}}
\label{appen_thrm_delay_bound}

From the queue dynamics \eqref{queue_dynamics}, we  obtain that $\forall i\in \Ii, c\in \Cc, t\geq0$:
\begin{align*}
 Q^c_i(t+1)^2  &\leq \big( Q^c_i(t) + a^c_i(t) -  \mu^c_i(t)\big)^2\\   
 &=  Q^c_i(t)^2 + (a^c_i(t) -  \mu^c_i(t))^2  + 2 Q^c_i(t) (a^c_i(t) -  \mu^c_i(t)).
\end{align*}
Taking conditional expectation of the above, we obtain that:
\begin{align}
\nonumber
&\EE\big[Q^c_i(t+1)^2 - Q^c_i(t)^2 \big|  Q^c_i(t)\big] \\
\nonumber
&\leq \EE[a^c_i(t)^2] - 2 \lambda^c_i \EE[ \mu^c_i(t)] + \EE[\mu^c_i(t)^2]+ 2  Q^c_i(t) (  \lambda^c_i  -  \EE[ \mu^c_i(t)]  )\\
\nonumber
&= \EE[a^c_i(t)^2] - 2 \lambda^c_i \EE[ \mu^c_i(t)] + Var[\mu^c_i(t)] + \EE[\mu^c_i(t)]^2 \\
\nonumber
&\quad+ 2  Q^c_i(t) (  \lambda^c_i  -  \EE[ \mu^c_i(t)]  )\\
\nonumber
&\leq   \EE[a^c_i(t)^2] - 2 \lambda^c_i  |\Kk| T_s \gamma^{c}_{i}q^c_i + T_s U_i +T_s^2  U_i^2\\
\nonumber
& \quad +  2  Q^c_i(t) \big(  \lambda^c_i  -    |\Kk| T_s \gamma^{c}_{i}q^c_i  \big), 
\end{align}
where the last inequality follows from \eqref{qos_throughput1} and \eqref{var_bound1}. Taking iterated expectation of the above, we get:
\begin{align*}
    \EE\big[Q^c_i(t+1)^2\big] -\EE\big[Q^c_i(t)^2\big] &\leq   \EE[a^c_i(t)^2] - 2 \lambda^c_i  |\Kk| T_s \gamma^{c}_{i}q^c_i + T_s U_i \\
& \quad +T_s^2  U_i^2+  2  \EE\big[ Q^c_i(t) \big] \big(  \lambda^c_i  -    |\Kk| T_s \gamma^{c}_{i}q^c_i  \big). 
\end{align*}
 Telescoping for $t=0 \to T-1$ and noting that $  \EE\big[ Q^c_i(0)\big]=0$, we obtain that:
 \begin{align*}
\nonumber
     0\leq  \EE\big[Q^c_i(T)^2\big] &\leq T \EE[a^c_i(t)^2]- 2 T \lambda^c_i  |\Kk| T_s \gamma^{c}_{i}q^c_i + T T_s U_i \\
     &\quad +T T_s^2  U_i^2+ 2 \big(  \lambda^c_i  -    |\Kk| T_s \gamma^{c}_{i}q^c_i  \big) \sum_{t=0}^{T-1} \EE\big[ Q^c_i(t) \big]\\
     \therefore \frac{1}{T}\sum_{t=0}^{T-1} \EE\big[ Q^c_i(t) \big]&\leq \frac{ \EE[a^c_i(t)^2]- 2  \lambda^c_i  |\Kk| T_s \gamma^{c}_{i}q^c_i +  T_s U_i + T_s^2  U_i^2}{2(  |\Kk| T_s \gamma^{c}_{i}q^c_i  - \lambda^c_i )}.
 \end{align*}
 Taking $\limsup$ for $T\to \infty$ of the above, we have:
 \begin{align}
 \nonumber
     \bar{Q}^c_i &= \limsup_{T\to \infty} \frac{1}{T} \sum_{t=0}^{T-1} \EE[Q^c_i(t)]\\
     &\leq \frac{ \EE[a^c_i(t)^2]- 2  \lambda^c_i  |\Kk| T_s \gamma^{c}_{i}q^c_i +  T_s U_i + T_s^2  U_i^2}{2(  |\Kk| T_s \gamma^{c}_{i}q^c_i  - \lambda^c_i )}.
 \end{align}
By Little's Theorem, the average queueing delay $W^c_i$ thus satisfies:
\begin{align*}
    W^c_i \leq  \frac{ \EE[a^c_i(t)^2]- 2  \lambda^c_i  |\Kk| T_s \gamma^{c}_{i}q^c_i +  T_s U_i + T_s^2  U_i^2}{2 \lambda^c_i (  |\Kk| T_s \gamma^{c}_{i}q^c_i  - \lambda^c_i )},
\end{align*}
which is the required statement \eqref{queue_delay_bound} of the Theorem. 
Finally, to obtain a sufficient condition for $W^c_i \leq W^*$, we set the RHS of \eqref{queue_delay_bound} to be less than or equal to $W^*$:
\begin{align*}
     &\frac{ \EE[a^c_i(t)^2]- 2  \lambda^c_i  |\Kk| T_s \gamma^{c}_{i}q^c_i +  T_s U_i + T_s^2  U_i^2}{2 \lambda^c_i (  |\Kk| T_s \gamma^{c}_{i}q^c_i  - \lambda^c_i )} \leq W^*\\
     \iff & 2W^*\lambda^c_i (  |\Kk| T_s \gamma^{c}_{i}q^c_i  - \lambda^c_i ) \geq \EE[a^c_i(t)^2]- 2  \lambda^c_i  |\Kk| T_s \gamma^{c}_{i}q^c_i \\
     &\quad \quad  \quad  \quad \quad \quad \quad \quad \quad \quad \quad  +  T_s U_i + T_s^2  U_i^2 \\
     \iff& 2\lambda^c_i |\Kk | T_s (W^*+1) \gamma^{c}_{i}q^c_i\geq \EE[a^c_i(t)^2]+  T_s U_i \\
     & \quad \quad  \quad  \quad \quad \quad \quad \quad \quad \quad + T_s^2  U_i^2 +2(\lambda^{c}_i)^2 W^*\\
    \iff &\gamma^c_i q^c_i   \geq \frac{ \EE[a^c_i(t)^2] + T_s^2 U_i^2 + T_s U_i + 2(\lambda^{c}_i)^2 W^*   }{2\lambda^c_i |\Kk | T_s (W^*+1)}, 
\end{align*}
which is the condition \eqref{queue_delay_suff1} as stated in the Theorem.

\section{ Omitted Proofs}

\subsection{Proof of Lemma \ref{lma_bound_h}}
\label{appen_lma_bound_h}
For any $ \pP(t) \in [0, 1]^{T_s |E| |\Cc|}$, if $\pP(t) \not\in S$, then $\bar{h}_t(\pP(t)) = -\infty$ by \eqref{indicator1} and the bound \eqref{bound_h1} trivially holds. Else if $\pP(t) \in S$, then $  \mathbbm{1}_S(\pP(t)) =0$ and thus $ \bar{h}_t(\pP(t)) = h_t(\pP(t))$. We have:
\begin{align*}
   \bar{h}_t(\pP(t))  &= \sum_{c\in \Cc}\sum_{i\in I} Q^c_i(t) \sum_{\tau = 0}^{T_s -1 }\sum_{k\in N(i)} r_{ki} p^c_{ki}(t, \tau) + V  f( \pP(t))\\
    &\leq  \| \QQ(t) \|_\infty  \sum_{c\in \Cc}\sum_{i\in I} \sum_{\tau = 0}^{T_s -1 }\sum_{k\in N(i)} p^c_{ki}(t, \tau)  + V f_{max}\\
    &\overset{\eqref{prob_constraint}}{=} \| \QQ(t) \|_\infty  T_s  |\Kk| + V f_{max}, 
\end{align*}
where in the above, we use $ \| \QQ(t) \|_\infty = \max_{i\in \Ii, c\in \Cc}  Q^c_i(t) $ and $f(\pP(t)) \in [f_{min}, f_{max}]$.

\subsection{Proof of Lemma \ref{penalty_diff}}
\label{appen_penalty_diff}

Applying Lemma \ref{min_ineq} and using \eqref{penalty}, we obtain that:
\begin{align*}
      | P_{r, \gma}(\pP(t)) -  P_{r, \gma'}(\pP(t))| &\leq r   \sum_{(i, c)\in \Ss} |\gamma^c_i  -{\gamma^{c}_i}'  |\\
      &= r  \|\gma  -\gma'  \|_1.
\end{align*}

\subsection{Proof of Lemma \ref{lma_penalty_main} }
\label{appen_lma_penalty_main}

% \ref{lma_2nd_slater}

Assumption \ref{assum_slater} ensures that  the Slater condition holds for  $M( \gma)$ with constant $\zeta$. 
Furthermore, since $\gamma^c_i = {\gamma^c_i}' + \frac{K_1}{\sqrt{T_s}} > {\gamma^c_i}' $,  the Slater condition  also holds for $M( \gma')$ with constant $\zeta$, as formally stated in Lemma \ref{lma_2nd_slater}. 
From the proof therein,  there exists some $\bar{\pP}(t) \in \widetilde{\Dd}(\gma, \qQ) \subseteq  \widetilde{\Dd}(\gma', \qQ) $ such that the Slater condition with constant $\zeta$ holds for both $M(\gma) $ and $ M(\gma')$. Thus, $\bar{h}_t(\bar{\pP}(t)) \overset{\eqref{def_h_bar_t}}{=} {h}_t(\bar{\pP}(t)) \geq V 
 f(\bar{\pP}(t)) \geq V f_{min}$.
Now, note that: 
\begin{align*}
    K_3&= \frac{\| \QQ(t) \|_\infty  |\Kk| T_s  }{\zeta} + \frac{V (f_{max} - f_{min})}{ \zeta} \\
    &\geq \frac{\| \QQ(t) \|_\infty  |\Kk| T_s  + V f_{max} - \bar{h}_t(\bar{\pP}(t))}{\zeta},
\end{align*}
where, by Lemma \ref{lma_bound_h}, $\| \QQ(t) \|_\infty  |\Kk| T_s  + V f_{max}$ serves as  an upper bound for the optimal values of either problems \eqref{proove_distort1} or \eqref{proove_distort1}.
By the result on linear exterior penalty \cite{penalty1}, for any $r\geq K_3$, the identities \eqref{lma_penalty_main_eq1} and \eqref{lma_penalty_main_eq2} hold.

% Furthermore, the next Lemma establishes that given \eqref{slater_cond},  the Slater condition also holds for $\widetilde{\Dd}(\gma', \qQ)$ with constant $\frac{\zeta}{2}$. 
% Furthermore, as illustrated in the following Lemma, the Slater condition also holds for $\widetilde{\Dd}(\gma', \qQ)$ with constant $\frac{\zeta}{2}$. 

\subsection{Proof of Lemma \ref{lma_bound_f}}
\label{appen_lma_bound_f}

For any $ \pP \in [0, 1]^{T_s |E| |\Cc|}$, if $\pP \not\in S(\lmb)$, then $\bar{f}(\pP) = -\infty$ by \eqref{indicator1.5} and the bound \eqref{bound_f_eq} trivially holds. Else if $\pP \in S(\lmb)$, then $  \mathbbm{1}_{S(\lmb)}(\pP) =0$ and thus $ \bar{f}(\pP) = f(\pP) \leq f_{max}$.

% \subsection{ Proof of Lemma \ref{lma_big_penalty}   }
% \label{appen_big_penalty}
% Applying Lemma \ref{min_ineq} and using \eqref{penalty_num}, we obtain that:
% \begin{align}
% \nonumber
%     &| P_{r, \lmb, \gma}(\pP) -  P_{r, \lmb', \gma'}(\pP)| \\
%     \nonumber
%     &\leq r \sum_{i \in I, c \in \Cc} | \lambda^c_i -\lambda'^c_i| + r \sum_{i \in I, c \in \Cc} | \gamma^c_i -\gamma'^c_i|\\
%       \nonumber
%     &= r \big( \|\lmb - \lmb'  \|_1 +   \|\gma  -\gma'  \|_1\big).
% \end{align}

\subsection{Proof of Lemma \ref{lma_penalty_num}}
\label{appen_lma_penalty_num}

Assumption \ref{assum_slater} ensures that  the Slater condition holds for  $M( \gma)$ with constant $\zeta$. 
Furthermore, since $\gamma^c_i = {\gamma^c_i}' + \frac{K_1}{\sqrt{T_s}} > {\gamma^c_i}' $,  the Slater condition  also holds for $M( \gma')$ with constant $\zeta$, as formally stated in Lemma \ref{lma_2nd_slater}. 
From the proof therein,  there exists some $\bar{\pP} \in  \Mm(\lmb, \gma, \qQ) \subseteq  \Mm(\lmb, \gma', \qQ)$ such that the Slater condition with constant $\zeta$ holds for both $M(\gma) $ and $ M(\gma')$. Thus, $\bar{f}(\bar{\pP}) \overset{\eqref{def_f_bar}}{=} f(\bar{\pP}) \geq f_{min}$. Now, note that: 
\begin{align*}
    \frac{ f_{max} - f_{min}}{\zeta} \geq \frac{ f_{max} -  \bar{f}(\bar{\pP}) }{\zeta},
\end{align*}
where, by Lemma \ref{appen_lma_bound_f}, $f_{max}$ serves as an upper bound for the optimal values of either problems \eqref{proof_lma_num_distort} or \eqref{proof_lma_num_distort2}.
By the result on linear exterior penalty \cite{penalty1}, for any $r\geq \frac{f_{max} -f_{min}}{\zeta}$, the identities \eqref{lma_penalty_num_eq1} and \eqref{lma_penalty_num_eq2} hold.

\subsection{Proof of Lemma \ref{lma_bound_f_tilde}}
\label{appen_lma_bound_f_tilde}

For any $ \pP \in [0, 1]^{T_s |E| |\Cc|}$, if $\pP \not\in \widetilde{\Dd}(\gma', \qQ)$, then $\tilde{f}(\pP) = -\infty$ by \eqref{indicator2} and the bound \eqref{bound_f_tilde_eq} trivially holds. Else if $\pP \in \widetilde{\Dd}(\gma', \qQ)$, then $  \mathbbm{1}_{\widetilde{\Dd}(\gma', \qQ)}(\pP) =0$ and thus $ \tilde{f}(\pP) = f(\pP) \leq f_{max}$.

\subsection{Proof of Lemma \ref{lma_final_bounds}}
\label{appen_lma_final_bounds}

% Note that $\lmb \leq \lmb_\varepsilon$

First, we have the following Lemma. 

% Consider the following quantities:
% \begin{align}
%  &\gamma_{min} = \min_{i\in \Ii, c\in \Cc} \{\gamma^c_i\} > 0,\\
%  &\lambda_{min} = \min_{i\in \Ii, c\in \Cc: \lambda^c_i >0} \{\lambda^c_i\} > 0. 
% \end{align}
\begin{lemma}
\label{lma_distance_from_stability_region}
Consider the quantity:
\begin{align}
    \varepsilon' =  0.5 \sqrt{T_s} K_4 |\Kk| \delta_{min} . \label{def_eps_prime}
\end{align}
Then, we have $\lmb + \varepsilon' \eE \in  \widetilde{\Lmb}(\gma, \qQ)$ and $\lmb_\varepsilon + \varepsilon' \eE \in  \widetilde{\Lmb}(\gma, \qQ)$, where $\eE = \{ \mathbbm{1}\{\lambda^c_i >0\}\}_{i\in \Ii, c\in \Cc}$. 
\end{lemma}
\begin{proof}
  As $\lmb \leq \lmb_\varepsilon$, it suffices to show that  $\lmb_\varepsilon + \varepsilon' \eE \in  \widetilde{\Lmb}(\gma, \qQ)$.  
  Since  $\lmb_\varepsilon = \frac{1}{1-0.5 \varepsilon} \lmb \subseteq \frac{1-\varepsilon}{1-0.5 \varepsilon} \widetilde{\Lmb}(\gma, \qQ)$ or equivalently $ \frac{1-0.5\varepsilon}{1- \varepsilon}\lmb_\varepsilon \in \widetilde{\Lmb}(\gma, \qQ)$, by Theorem \ref{thm_charac_stability_region} there exists some $\pP \in  \widetilde{\Dd}(\gma, \qQ)$ such that $\forall i\in \Ii, c\in \Cc$:
  % $= \lmb_\varepsilon + \frac{0.5 \varepsilon}{(1-\varepsilon)(1-0.5\varepsilon)} \lmb $
  \begin{align}
      \frac{1-0.5\varepsilon}{1- \varepsilon}&[\lmb_\varepsilon]^c_i \leq  \sum_{k\in N(i)} \sum_{\tau=0}^{T_s-1} r_{ki} p^c_{ki}(\tau) \label{proof_lma_distance_from_stability_region0} \\ 
      \nonumber 
      &[\lmb_\varepsilon]^c_i \leq \Big(1 -\frac{0.5\varepsilon}{1-\varepsilon} \Big) \,  \Big( \sum_{\tau=0}^{T_s-1} r_{ki} p^c_{ki}(\tau)\Big) \\
      \therefore  &[\lmb_\varepsilon]^c_i + 0.5\varepsilon  \Big( \sum_{\tau=0}^{T_s-1} r_{ki} p^c_{ki}(\tau)\Big) \leq   \sum_{\tau=0}^{T_s-1} r_{ki} p^c_{ki}(\tau) .\label{proof_lma_distance_from_stability_region1}
  \end{align}
 % By Lemma \ref{lma_lower_bound_service}, we have: 
 %  \begin{align}
 %  &\sum_{k\in N(i)} \sum_{\tau=0}^{T_s-1} r_{ki}  p^c_{ki}( \tau)\geq|\Kk| T_s \gamma_{min}, \forall (i, c)\in \Ss. 
 %  \end{align}
From \eqref{proof_lma_distance_from_stability_region0} and noting that $ \delta(\lmb) = \min_{i\in \Ii, c\in \Cc: \lambda^c_i > 0}\big\{ \frac{\lambda^c_i}{|\Kk| T_s}\big\} \geq \delta_{min}$, we have:
\begin{align}
\nonumber
  \sum_{k\in N(i)} \sum_{\tau=0}^{T_s-1} r_{ki}  p^c_{ki}( \tau)&\geq    \frac{1-0.5\varepsilon}{1- \varepsilon}[\lmb_\varepsilon]^c_i =   \frac{1}{1- \varepsilon}\lambda^c_i \geq \geq |\Kk| T_s \delta(\lmb) \\
  &\geq |\Kk| T_s \delta_{min}, \forall i\in \Ii, c\in \Cc: \lambda^c_i > 0.\label{proof_lma_distance_from_stability_region2}
\end{align}
Applying \eqref{proof_lma_distance_from_stability_region2} to the LHS of \eqref{proof_lma_distance_from_stability_region1} and using $\varepsilon = \frac{K_4}{\sqrt{T_s}}$ and \eqref{def_eps_prime}, we obtain that:
\begin{align}
   [\lmb_\varepsilon]^c_i+ \varepsilon' \leq \sum_{k\in N(i)} \sum_{\tau=0}^{T_s-1} r_{ki}  p^c_{ki}( \tau), \forall i\in I, c\in \Cc: \lambda^c_i > 0.\label{proof_lma_distance_from_stability_region3}
\end{align}
Since $\pP \in  \widetilde{\Dd}(\gma, \qQ)$ satisfies \eqref{proof_lma_distance_from_stability_region3}, by the characterization of $\widetilde{\Lmb}(\gma, \qQ)$ as in Theorem \ref{thm_charac_stability_region}, we conclude that $\lmb_\varepsilon + \varepsilon' \eE \in  \widetilde{\Lmb}(\gma, \qQ)$.
\end{proof}

% Since Assumption \ref{assum_slater} holds for 

Back to the main proof, we now proceed to derive bounds for  $  F(\pP^*) -  F(\pP^{\pi_1})$ and $ F(\pP^{\pi_1}) -  F(\pP^{\pi_\varepsilon})$.

\textbf{Bounding $  F(\pP^*) -  F(\pP^{\pi_1})$:} From \eqref{thrm_mdp1}, we have:
\begin{align}
     F(\pP^*) -  F(\pP^{\pi_1}) \leq  F(\pP^{\pi_2}) -  F(\pP^{\pi_1}). \label{proof_appen_lma_final_bounds1}
\end{align}
Recall that $\pi_1$ and $\pi_2$ respectively are optimal policies for $(\AOP_{\lmb, \gma, \qQ})$ and 
 $(\AOP_{\lmb, \gma - \frac{K_1}{\sqrt{T_s}} \one, \qQ})$. 
 Furthermore, we have $\lmb + \varepsilon' \one \in  \widetilde{\Lmb}(\gma, \qQ)$ with $\varepsilon'$ defined in \eqref{def_eps_prime} from Lemma \ref{lma_distance_from_stability_region}.
 By invoking Lemma \ref{lma_num_distort}, we have:
 \begin{align}
 \nonumber
      F(\pP^{\pi_2}) -  F(\pP^{\pi_1}) &\leq \frac{f_{max} - f_{min}}{\zeta} \, \|\gma - (\gma - \frac{K_1}{\sqrt{T_s}} \one) \|_1 \\
      \nonumber
      &= \frac{f_{max} - f_{min}}{\zeta} \, \frac{K_1 |\Ii| |\Cc|}{\sqrt{T_s}}\\
      &= O\bigg(\frac{1}{\sqrt{T_s}}\bigg). \label{proof_appen_lma_final_bounds2}
 \end{align}
 Plugging \eqref{proof_appen_lma_final_bounds2} into \eqref{proof_appen_lma_final_bounds1}, we  conclude the required statement \eqref{lma_final_bounds_1} of the Lemma.

\textbf{Bounding $ F(\pP^{\pi_1}) -  F(\pP^{\pi_\varepsilon})$:} Recall that $\pi_1$ and $\pi_\varepsilon$ respectively are optimal policies for $(\AOP_{\lmb, \gma, \qQ})$ and 
 $(\AOP_{\lmb_\varepsilon, \gma, \qQ})$. Furthermore, we have $\lmb + \varepsilon' \one \in  \widetilde{\Lmb}(\gma, \qQ)$ and $\lmb_\varepsilon + \varepsilon' \one \in  \widetilde{\Lmb}(\gma, \qQ)$ with $\varepsilon' $ defined in \eqref{def_eps_prime} from Lemma \ref{lma_distance_from_stability_region}. By invoking Lemma \ref{lma_num_distort} and  using $\varepsilon = \frac{K_4}{\sqrt{T_s}}$, we have:
 \begin{align*}
 \nonumber
     F(\pP^{\pi_1}) -  F(\pP^{\pi_\varepsilon}) &\leq \frac{f_{max} - f_{min}}{\varepsilon'} \, \|\lmb - \lmb_\varepsilon \|_\infty \\
     &= \frac{f_{max} - f_{min}}{  0.5 \sqrt{T_s} K_4 |\Kk| \delta_{min} }  \, \|\lmb - \frac{1}{1-0.5\varepsilon} \lmb\|_\infty\\
     & = \frac{f_{max} - f_{min}}{    0.5 \sqrt{T_s} K_4 |\Kk| \delta_{min} }  \, \frac{0.5\varepsilon}{1-0.5\varepsilon} \| \lmb\|_\infty \\
     &= \frac{(f_{max} - f_{min})  \| \lmb\|_\infty }{   0.5 {T_s} K_4 |\Kk| \delta_{min}(1-0.5\varepsilon)} \\
     &= O\bigg(\frac{1}{{T_s}}\bigg),
 \end{align*}
which concludes the required statement \eqref{lma_final_bounds_2} of the Lemma.

\section{Supplementary Lemmas}

% \begin{lemma}
% For any $i \in \Ii, c\in \Cc$, we have:
% \begin{align}
%     \sum_{i\in \Ii, c\in \Cc} Q^c_i(t+1)^2 - Q^c_i(t)^2 \leq (A_{max} + )^2
% \end{align}
% \end{lemma}

% First, we consider the following Lemma.
\begin{lemma}
\label{min_ineq}
For any $ x, y \in \Br$, we have:
\begin{align}
    | \min\{0, x\} - \min\{0, y\} | \leq |x-y|.
\end{align}    
\end{lemma}
\begin{proof}
WLOG, we assume that $x\geq y$. If $x\geq y\geq 0$, we have $ | \min\{0, x\} - \min\{0, y\} | = 0 \leq   |x-y|$. If $x \geq 0 \geq y$, we have $| \min\{0, x\} - \min\{0, y\} | = -y \leq x-y = |x-y|$. If $0\geq x \geq y$, we have $ | \min\{0, x\} - \min\{0, y\} | = |x-y|$.
\end{proof}

\begin{lemma}
\label{lma_cauchy}
 We have the following bound:
 \begin{align}
     \EE[\|\QQ(T)\|_1] \leq \sqrt{2|\Cc| |\Ii| \, \EE[\Ll(\QQ(T))]}  \label{cauchy}
 \end{align}
\end{lemma}
\begin{proof}
     By Holder's inequality, we first have:
 \begin{align}
 \nonumber
 \|\QQ(T)\|_1 = \sum_{i\in \Ii, c\in \Cc} Q^c_i(T) &\leq \sqrt{|\Cc| |\Ii|  \cdot  \big(  \sum_{i\in \Ii, c\in \Cc} Q^c_i(T)^2   \big)}\\
 &= \sqrt{2 |\Cc| |\Ii|  \cdot  \Ll(\QQ(T)) }. \label{proof_lma_cauchy1}
 \end{align}
 By Jensen's inequality, we obtain that:
 \begin{align*}
      \EE\big[\|\QQ(T)\|_1\big]^2&\leq  \EE[\|\QQ(T)\|_1^2  ] \overset{\eqref{proof_lma_cauchy1}}{\leq} 2 |\Cc| |\Ii|  \, \EE[  \Ll(\QQ(T))]  \\
     \therefore  \EE[\|\QQ(T)\|_1] &\leq \sqrt{2|\Cc| |\Ii| \, \EE[\Ll(\QQ(T))]} ,
 \end{align*}
 which concludes the proof of the Lemma.
\end{proof}

\begin{lemma}
\label{lma_2nd_slater}
    Under Assumption \ref{assum_slater}, there exists $\pP(t) \in \widetilde{\Dd}(\gma', \qQ)$ such that $\forall i\in \Ii, c\in \Cc$:
\begin{align}
      |\Kk|T_s ({\gamma^c_i}' + {\zeta}) \leq \sum_{k\in N(i)} \sum_{\tau=0}^{T_s-1} p^c_{ki}(t, \tau) - B^c_i(\Gamma^c_i, \pP(t)).
\end{align}
\end{lemma}
\begin{proof}
Since Assumption \ref{assum_slater} holds, there exists $\pP(t) \in \widetilde{\Dd}(\gma, \qQ) \subseteq \widetilde{\Dd}(\gma', \qQ) $  that  satisfies \eqref{prob_constraint}, \eqref{pmax1} and \eqref{slater}. Noting that $\gma' = \gma - \frac{K_1}{\sqrt{T_s}} \one $, the inequality \eqref{slater} directly implies that:
\begin{align*}
      |\Kk|T_s ( {\gamma^c_i}' + \zeta ) \leq \sum_{k\in N(i)} \sum_{\tau=0}^{T_s-1} p^c_{ki}(t, \tau) - B^c_i(\Gamma^c_i, \pP(t)),
\end{align*}
which concludes the proof of the Lemma.
\end{proof}

% \begin{lemma}
% \label{lma_bound_B}
% We have the following:
% \begin{align}
%      \frac{1}{2}  \Gamma^c_i  \leq &B^c_i(\Gamma^c_i, \pP ) \leq \Gamma^c_i. \label{bound_B}
% \end{align}
% \end{lemma}
% \begin{proof}
%     Since, by definition,  $\hat{p}^c_{ki}(\tau)= \max\{p^c_{ki}(\tau), 1- p^c_{ki}(\tau)\} \in [\frac{1}{2}, 1]$, we  obtain the following bound in view of \eqref{big_b1}:
% \begin{align}
% \nonumber
%        \frac{1}{2}\lfloor \Gamma^c_i\rfloor + \frac{1}{2} (\Gamma^c_i - \lfloor \Gamma^c_i\rfloor) \leq &B^c_i(\Gamma^c_i, \pP) \leq \lfloor \Gamma^c_i\rfloor + (\Gamma^c_i - \lfloor \Gamma^c_i\rfloor) \\
%    \therefore  \frac{1}{2}  \Gamma^c_i  \leq &B^c_i(\Gamma^c_i, \pP) \leq \Gamma^c_i. 
% \end{align}
% \end{proof}

\begin{lemma}
\label{lma_lower_bound_service}
For $\pP \in  \widetilde{\Dd}(\gma, \qQ)$, we have the following:
\begin{align}
    \sum_{k\in N(i)} \sum_{\tau=0}^{T_s-1} r_{ki}  p^c_{ki}( \tau)\geq |\Kk| T_s \gamma_{min} , \forall (i, c)\in \Ss. \label{lower_bound_service}
\end{align}
\end{lemma}
\begin{proof}
Since   $\pP \in  \widetilde{\Dd}(\gma, \qQ)$, we have $\forall (i, c)\in \Ss$:
      \begin{align*}
  \nonumber
      & |\Kk|T_s \gamma^c_i \leq \sum_{k\in N(i)} \sum_{\tau=0}^{T_s-1} r_{ki}  p^c_{ki}( \tau) - B^c_i(\Gamma^c_i, \pP) \\
      \nonumber
       \therefore  &\sum_{k\in N(i)} \sum_{\tau=0}^{T_s-1} r_{ki}  p^c_{ki}( \tau)\geq |\Kk| T_s \gamma^c_{i} \geq |\Kk| T_s \gamma_{min},
  \end{align*}
  which concludes the proof of the Lemma. 
\end{proof}

% \begin{lemma}
% \label{lma_lower_bound_service}
% For $\pP \in  \widetilde{\Dd}(\gma, \qQ)$, we have the following:
% \begin{align}
%     \sum_{k\in N(i)} \sum_{\tau=0}^{T_s-1} r_{ki}  p^c_{ki}( \tau)\geq \sqrt{\frac{1}{2} |\Kk| T_s \log\big( (1-q_{min})^{-1} \big)}. \label{lower_bound_service}
% \end{align}
% \end{lemma}
% \begin{proof}
% Since   $\pP \in  \widetilde{\Dd}(\gma, \qQ)$, we have:
%       \begin{align*}
%   \nonumber
%       & |\Kk|T_s \gamma^c_i \leq \sum_{k\in N(i)} \sum_{\tau=0}^{T_s-1} r_{ki}  p^c_{ki}( \tau) - B^c_i(\Gamma^c_i, \pP) \\
%       \nonumber
%        \therefore  &\sum_{k\in N(i)} \sum_{\tau=0}^{T_s-1} r_{ki}  p^c_{ki}( \tau)\geq B^c_i(\Gamma^c_i, \pP) \overset{\eqref{bound_B}}{\geq} \frac{1}{2}\Gamma^c_i\\
%        \nonumber
%        &\quad \quad \quad \quad\quad\quad\quad\quad\overset{\eqref{gamma_choice}}{\geq}  \sqrt{\frac{1}{2} |\Kk| T_s \log\big( (1-q^c_i)^{-1} \big)} \\
%        &\quad \quad \quad\quad \quad \quad\quad\quad\geq \sqrt{\frac{1}{2} |\Kk| T_s \log\big( (1-q_{min})^{-1} \big)},
%   \end{align*}
%   which concludes the proof of the Lemma. 
% \end{proof}

\section{Properties of Convex Optimization over Linear Constraints}

Let $ f:\mathbb{R}^n \to \mathbb{R}\cup\{+\infty\} $ be a proper, convex, and lower semicontinuous function, and let $ A\in\mathbb{R}^{m\times n} $ and $ \bB\in\mathbb{R}^m $. Consider the primal problem:
\[
\begin{array}{ll}
\text{(Primal)} \quad \min\limits_{x\in\mathbb{R}^n} & f(\xX)\\[1mm]
\quad \text{subject to} & A\xX \le \bB.
\end{array}
\]
Suppose there exists a strictly feasible point \(\bar{\xX}\) satisfying:
\begin{align}
A\bar{\xX} &\le \bB - \zeta\,\mathbf{1} \label{slate_appen1}
\end{align}
for some $\zeta > 0$, i.e. the Slater condition holds.
We consider an optimal solution \( \xX^* = \argminE_{\xX \in \mathbb{R}^n: A\xX \leq \bB} f(\xX)\). 
% and  \( \tilde{x}^*   = \argminE_{Ax \leq \tilde{b}} f(x)\). 

\begin{lemma}
\label{lem:fenchel_duality}
Let $f^*(.)$ be the Frenchel conjugate of $f(.)$.  The dual problem of (Primal) can be written in its Fenchel form as
\[
\text{(Dual)} \quad \sup_{\lmb\ge 0} \left\{ g(\lmb) \triangleq -\bB^\top \lmb - f^*\bigl(-A^\top\lmb\bigr) \right\},
\]
and strong duality holds; that is, the optimal value of (Primal) equals the optimal value of (Dual).
\end{lemma}

\begin{proof}
Define the Lagrangian for (P) by:
\[
\mathcal{L}(\xX,\lmb) = f(\xX) + \lmb^\top (A\xX - \bB), \quad \lmb \ge 0.
\]
The dual function is then given by:
\[
g(\lmb) = \inf_{\xX\in\mathbb{R}^n} \Bigl\{ f(\xX) + \lmb^\top (A\xX - \bB) \Bigr\} 
= \inf_{\xX\in\mathbb{R}^n} \Bigl\{ f(\xX) + \lmb^\top A\xX \Bigr\} - \bB^\top \lmb.
\]
Recall the definition of the Fenchel conjugate of \(f(.)\):
\[
f^*(\yY) = \sup_{\xX\in\mathbb{R}^n} \Bigl\{ \xX^\top \yY - f(\xX) \Bigr\}.
\]
Thus, we have:
\[
\inf_{\xX\in\mathbb{R}^n} \Bigl\{ f(\xX) + \lmb^\top A\xX \Bigr\}
= -\sup_{\xX\in\mathbb{R}^n} \Bigl\{ -\lmb^\top A\xX - f(x) \Bigr\}
= -f^*\bigl(-A^\top\lmb\bigr).
\]
Thus, the dual function becomes:
\[
g(\lmb) = -\bB^\top\lmb - f^*\bigl(-A^\top\lmb\bigr).
\]
The dual problem is to maximize \( g(\lmb) \) over \( \lmb \ge 0 \), namely,
\[
\sup_{\lmb\ge 0} \Bigl\{ -\bB^\top\lmb - f^*\bigl(-A^\top\lmb\bigr) \Bigr\}.
\]

Under the stated assumptions (convexity of \(f\), lower semicontinuity, and the Slater condition \(A\bar{\xX} < b\)), standard results in convex optimization \cite{Boyd_Vandenberghe_2004} guarantee that strong duality holds. Hence, the optimal values of the primal and dual problems are equal.
\end{proof}

\begin{theorem}[Bound on the Optimal Dual Solution]
\label{thm:dual_bound}
 Let \( \lmb^*\ge 0 \) be an optimal  solution to (Dual).  We have:
\begin{align*}
\|\lmb^*\|_1 \leq \frac{f(\bar{\xX}) - f(\xX^*)}{\zeta}.
\end{align*}
\end{theorem}

\begin{proof}
% Since \(\bar{x}\) is strictly feasible, for each \(i = 1, \dots, m\) we have
% \begin{align*}
% b_i - [A\bar{x}]_i \ge \epsilon.
% \end{align*}
% Let \( x^* \) be an optimal solution with \( p^* = f(x^*) \). 
% By weak duality and the properties of Lagrange multipliers, we obtain

Since strong duality holds, we have:
\begin{align*}
    f(\xX^*)  = g(\lmb^*) = \inf_{\xX\in\mathbb{R}^n} \mathcal{L}(\xX,\lmb^*) &\leq \mathcal{L}(\bar{\xX},\lmb^*) = f(\bar{\xX}) + \lmb^\top (A \bar{\xX} - \bB)\\
\therefore f(\bar{\xX}) - f(\xX^*) &\ge \sum_{i=1}^m \lambda^*_i \Bigl( b_i - [A\bar{\xX}]_i \Bigr) \\
&\overset{\eqref{slate_appen1}}{\geq } \zeta \sum_{i=1}^m \lambda^*_i = \zeta \,\|\lmb^*\|_1.
\end{align*}
Thus, it follows that $\|\lmb^*\|_1 \le \frac{f(\bar{\xX}) -f(\xX^*)}{\zeta}$.

\end{proof}

% , i.e. \eqref{slate_appen1},

\begin{theorem}
\label{thm_perturbation_bound_no_v_in_statement}
Besides $\xX^* = \argminE_{\xX \in \mathbb{R}^n: A\xX \leq \bB} f(\xX)$, we let  \( \tilde{\xX}^*   = \argminE_{A\xX \leq \tilde{\bB}} f(\xX)\)  be  the solution to the primal problem with the perturbed boundary vector $\tilde{\bB} \in\mathbb{R}^m  $. Assume that the Slater condition holds with constant $\zeta$ for both right-hand sides $\bB$  and $\tilde{\bB}$, i.e. $A\bar{\xX} \le \tilde{\bB} - \zeta\,\mathbf{1}$ and $A\bar{\xX}' \le {\bB} - \zeta\,\mathbf{1}$ for some $\bar{\xX}$ and $\bar{\xX}'$. 
% If $f(\xX) \leq f_{bound}, \forall \xX \in \{\xX \in \mathbb{R}^n: A \xX \leq \bB \text{ or } A\xX \leq \tilde{\bB} \}$, we have the following:
Let $\kappa = \frac{1}{\zeta}\max\big\{f(\bar{\xX}) - f(\xX^*), f(\bar{\xX}') - f(\tilde{\xX}^*) \big\}$. We have the following bound:
\[
0\leq  |f(\xX^*) - f(\tilde{\xX}^*)|  \le \kappa \,  \|\tilde{\bB} - \bB\|_\infty.
\]
\end{theorem}

\begin{proof}

% By strong duality, if \( x^* \) is an optimal solution for the problem with right-hand side \( b \), then \( f(x^*) = v(b) \). Similarly, if \( \tilde{x}^* \) is optimal for the perturbed problem, then \( f(\tilde{x}^*) = v(\tilde{b}) \).
% As  $\tilde{\bB} \geq \bB$, $\xX^*$ is also a feasible solution to $\{\xX \in \mathbb{R}^m: A \xX\leq \tilde{\bB}$\}. Thus, we have $f(\tilde{\xX}^*) \leq f(\xX^*)$, or equivalently $0 \leq f(\xX^*) - f(\tilde{\xX}^*)$. To prove the remaining part of the required statement, we 

Let $\tilde{\lmb}^*$ be the optimal dual solution to the problem $\min_{A\xX \leq \tilde{\bB}} f(\xX) $.
Define the optimal value function given the boundary vector $\bB$ as:
\[
v(\bB)=\min\{\,f(\xX):\ A\xX \le \bB\,\}.
\]
Since \( v(\bB) \) is convex in \( \bB \) and \( \lambda^* \) is an optimal dual solution for the problem with right-hand side \( \bB \), it follows that \( \lambda^* \) is a subgradient of \( v \) at \( \bB \). By  the convexity of \( v \), we have:
\begin{align}
\nonumber
&v(\tilde{\bB}) \ge v(\bB) + \lambda^{*\top} (\tilde{\bB} - \bB).
\end{align}
Similarly, we have:
\begin{align}
\nonumber
&v(\bB) \ge v(\tilde{\bB}) + \tilde{\lambda}^{*\top} ( \bB - \tilde{\bB}).
\end{align}
Combining the above two inequalities, we obtain that:
\begin{align}
    &\bigl|v(\tilde{\bB}) - v(\bB)\bigr| \le \max\{ \lambda^{*\top} ( \bB - \tilde{\bB} ), \tilde{\lambda}^{*\top} (\tilde{\bB}-\bB )\}. \label{perturbation_bound_no_v_in_statement_1}
\end{align}
By Holder's inequality and Theorem \ref{thm:dual_bound}, we have:
\begin{align}
\label{perturbation_bound_no_v_in_statement_2}
  \bigl|\lambda^{*\top} (\tilde{\bB} - \bB)\bigr| \leq  \|\lambda^*\|_1 \, \| \tilde{\bB} - \bB \|_\infty  \leq \frac{f(\bar{\xX}) -f(\xX^*)}{\zeta} \,  \|\bB-\tilde{\bB}\|_\infty,\\
   \bigl|\tilde{\lambda}^{*\top} (\tilde{\bB} - \bB)\bigr| \leq  \|\tilde{\lambda}^*\|_1 \, \| \tilde{\bB} - \bB \|_\infty  \leq \frac{f(\bar{\xX}') -f(\tilde{\xX}^*)}{\zeta} \,  \|\bB-\tilde{\bB}\|_\infty. \label{perturbation_bound_no_v_in_statement_3}
\end{align}
Applying \eqref{perturbation_bound_no_v_in_statement_2} and \eqref{perturbation_bound_no_v_in_statement_3} to \eqref{perturbation_bound_no_v_in_statement_1} and noting that  \( f(\xX^*) = v(\bB) \) and \( f(\tilde{\xX}^*) = v(\tilde{\bB}) \), we conclude the required statement.
\end{proof}

\end{document}